\documentclass[11pt,letterpaper]{article}   \usepackage{psmacros}
\date{\today}

\newcommand{\E}{\ensuremath{\mathbb E}}

    \def\a{{\mathbf{\alpha}}}    
  \def\beq{\begin{eqnarray}} \def\eeq{\end{eqnarray}} \def\ben{\begin{enumerate}}
\def\een{\end{enumerate}}
 \def\bit{\begin{itemize}}
\def\eit{\end{itemize}}
 \def\beqs{\begin{eqnarray*}} \def\eeqs{\end{eqnarray*}} \def\bel{\begin{lemma}} \def\eel{\end{lemma}}

   \newcommand{\p}{\mathbb{P}}
  \newcommand{\HH}{\mathcal H}  \newcommand{\one}{\mathrm{1}}
 \newcommand{\MM}{\mathcal M} \newcommand{\la}{\lambda}  
 \def\a{{\mathbf{\alpha}}}  \def\eps{{\epsilon}}  \def\ie{i.\,e.\,} \def\g{G}

\renewcommand{\a}{\alpha}

\newcommand{\dist}{dist}

\newcommand{\beqn}{\begin{equation}}
\newcommand{\eeqn}{\end{equation}}

\newcommand{\FF}{\ensuremath{\mathcal{F}}}
\newcommand{\vol}[2][]{\ensuremath{\textup{vol}_{#1}\inp{#2}}}
\newcommand{\surfarea}[2][]{\ensuremath{\textup{surface-area}_{#1}\inp{#2}}}
\newcommand{\sidelen}[1]{\ensuremath{\textup{sidelength}\inp{#1}}}
\newcommand{\bdry}[1]{\ensuremath{\partial{}#1}}

\newcommand{\sbad}{inside-heavy}

\newcommand{\asbad}{an inside-heavy}

\newcommand{\stub}[2][]{\ensuremath{{\textup{stub}}_{#1}\inp{#2}}}

\newcommand{\nlab}[1][]{\ensuremath{\textup{label}_{#1}}}

  \newcommand{\grantacknowledgement}{We thank the anonymous referees of this paper
    for their helpful comments. HN and PS acknowledge support from the
    Department of Atomic Energy, Government of India, under project no. RTI4001,
    the Ramanujan Fellowship of SERB, and the Infosys-Chandrasekharan virtual
    center for Random Geometry.  HN acknowledges support from a Swarna Jayanti
    fellowship.  PS acknowledges support from Adobe Systems Incorporated via a
    gift to TIFR. AR acknowledges support from an Akamai Presidential Fellowship.  The contents of this paper do not necessarily reflect the
    views of the funding agencies listed above.}

\usepackage{booktabs} \usepackage{array} \usepackage{amssymb}
\usepackage{fancyhdr}   
\newcommand{\Kb}{\R^n\setminus{(K^\circ)}}

\DeclareMathOperator{\diam}{diam}
\renewcommand{\dist}{\ensuremath\operatorname{dist}}
\newcommand{\len}{\ensuremath\operatorname{length}}

\newcommand{\cntr}{\operatorname{center}}
\newcommand{\CHR}{\operatorname{CHR}}
\newcommand{\bS}{\mathbb{S}}
\newcommand{\cube}{\mathcal{Q}}

\usepackage{float}
\newcommand{\Depth}{3}
\newcommand{\Height}{3}
\newcommand{\Width}{3}

\title{Sampling from convex sets with a cold start\\ using multiscale
  decompositions} \author{Hariharan Narayanan\thanks{Tata Institute of
    Fundamental Research, Mumbai. Email:
    \texttt{hariharan.narayanan@tifr.res.in}.} \and Amit Rajaraman\thanks{MIT. Email: \texttt{amit\_r@mit.edu}. Much of this work was done while this author was a student at the Indian Institute of Technology Bombay.} \and Piyush Srivastava\thanks{Tata Institute of Fundamental Research,
    Mumbai. Email: \texttt{piyush.srivastava@tifr.res.in}.}}
\date{} 
\begin{document}
\maketitle

\begin{abstract}
  A standard approach for sampling approximately uniformly from a convex body
  $K \subseteq \R^n$ is to run a random walk within $K$.  The requirement is
  that starting from a suitable initial distribution, the random walk should
  ``mix rapidly'', i.e., after a number of steps that is polynomial in $n$ and
  the aspect ratio $R/r$ (here, $K$ is assumed to contain a ball of radius $r$
  and to be contained within a ball of radius $R$), the distribution of the
  random walk should come close to the uniform distribution $\pi_K$ on
  $K$. Different random walks differ in aspects such as the ease of
  implementation of each step, or suitability for a specific class of convex
  bodies.  Therefore, the rapid mixing of a wide variety of random walks on
  convex bodies has been studied.

  Many proofs of rapid mixing of such random walks however require that the
  initial distribution of the random walk is not too different from the target
  distribution $\pi_K$.  In particular, they require that the probability
  density function of the initial distribution with respect to the uniform
  distribution $\pi_K$ on $K$ must be bounded above by $\mathrm{poly}(n)$: this
  is called a \emph{warm start}. Achieving such a warm start often requires a
  non-trivial pre-processing step before the random walk can be started.  This
  motivates the problem of proving rapid mixing from ``cold starts'', i.e.,
  when the density of the initial distribution with respect to $\pi_K$ can be as
  high as $\exp(\mathrm{poly}(n))$.  In contrast to warm starts, a cold start is
  usually trivial to achieve. However, rapid mixing from a cold start may not
  hold for every random walk, e.g., the well-known ``ball walk'' does not
  have rapid mixing from an arbitrary cold start.  On the other hand, for the
  ``hit-and-run'' random walk, Lovász and Vempala proved rapid mixing from a
  cold start.  For the related \emph{coordinate} hit-and-run (CHR) random walk,
  which has been found to be promising in computational experiments, a rapid
  mixing result starting from a warm start was proven only recently, while the
  question of whether CHR mixes rapidly from a cold start remained open.

  In this paper, we construct a family of Markov chains inspired by classical
  multiscale decompositions of subsets of $\mathbb{R}^n$ into countably many
  axis-aligned cubes. We show that even with a cold start, the mixing times of
  these chains are bounded by a polynomial in $n$ and the aspect ratio of the
  body.  Our main technical ingredient is an isoperimetric inequality for $K$
  for a metric that magnifies distances between points that are close to the
  boundary of $K$.  As a byproduct of the analysis of this new family of chains,
  we show that the coordinate hit-and-run (CHR) random walk also mixes rapidly
  from a cold start, and also from any point that is not too close to the
  boundary of the body.
\end{abstract}
\thispagestyle{empty}
\newpage
\setcounter{tocdepth}{2}
\tableofcontents
\thispagestyle{empty}
\newpage
\setcounter{page}{1}
\section{Introduction}

The problem of generating a point distributed (approximately) uniformly over a
convex set $K \subseteq \R^n$ is an important algorithmic primitive.  It is
usual to assume that the body $K$ is presented by means of a ``well guaranteed
membership oracle'', i.e., a membership oracle for $K$, along with values
$R > r > 0$ such that the body is contained in the radius $R$ Euclidean ball and
also contains the radius $r$ Euclidean ball.  The ratio $R/r$ is then referred
to as the \emph{aspect ratio} of the body.

The first provably polynomial time algorithm for this problem was given by Dyer,
Frieze and Kannan~\cite{DFK91}: their algorithm used a random walk on a
uniformly-spaced lattice of points in a suitable ``smoothed'' version of the
original body $K$.  More refined analyses of such lattice walks were given in
subsequent works \cite{AK91,DF91,LS90}: we refer to \cite{LS93} for a more
complete discussion of the history.  Soon after, Lovász~\cite{L90-ICM} and
Lovász and Simonovits~\cite{LS93} considered a more geometric random walk not
supported on a discrete lattice: the so-called \emph{ball walk}.  Here, one
fixes a radius parameter $\delta$, and given a current point $x \in K$, proposes
a next point $y$ from the Euclidean ball of radius $\delta$ centered at $x$, and
moves to $y$ if $y \in K$.  They prove (see \cite[Remark on p.~398]{LS93}) that when
$\delta$ is chosen appropriately, the lazy\footnote{A random walk is called
  \emph{lazy} if the probability that it stays at its current state after one
  step is at least $1/2$. A lazy version of any random walk $W$ can be obtained
  by considering the random walk in which at each step, the walk simply stays at
  the current state with probability $1/2$, and takes a step according to $W$
  with probability $1/2$.  Considering only lazy versions of walks is a standard
  device for avoiding pathological periodicity issues, and therefore we will
  always work with lazy walks in this paper.} version of the ball walk \emph{mixes rapidly}, i.e.,
it reaches a distribution that is $\epsilon$-close in total variation distance
to the uniform distribution $\pi_K$ on $K$, after a number of steps which is
polynomial in $n, 1/\epsilon$ and $R/r$, \emph{provided} that the initial point
of the random walk is chosen according to a $\poly{n}$-\emph{warm} start.  (A
distribution $\mu$ supported on $K$ is said to be \emph{$M$-warm} if the density
function of $\mu$ with respect to $\pi_K$ is bounded above by $M$.)  Another
natural geometric random walk is the \emph{hit-and-run} walk (see
\cite{smith_efficient_1984}, where it is attributed to earlier work of Boneh and
Golan, and of Smith).  Here, if the current state is $x \in K$, then the next
point $y$ is sampled by first choosing a uniformly random direction $\hat{u}$
from the unit sphere $\bS^{{n-1}}$, and then picking $y$ uniformly at random
from the chord of $K$ in direction $\hat{u}$ passing through $x$.
Lovász~\cite{L99} proved that the lazy hit-and-run walk also mixes in time
polynomial in $n$, $1/\epsilon$ and $R/r$, again assuming that the initial point
is sampled from a $\poly{n}$-warm start.

While a $\poly{n}$-warm start can be achieved in polynomial time, it requires
sophisticated pre-processing. In contrast, a ``cold start'', i.e., an $M$-warm
start where $M$ can be as large as $\exp(\poly{n})$, is very easy to generate
when $R/r$ is at most $\exp(\poly{n})$: one can simply sample the initial point
uniformly at random from the radius $r$ Euclidean ball. The first polynomial
time mixing time result for the hit-and-run walk from such a cold start, without
the need for any further pre-processing, was proved by Lovász and
Vempala~\cite{lovasz_hit-and-run_2006}.

An interesting variant of the hit-and-run walk is the \emph{coordinate}
hit-and-run (CHR) walk, where the direction $\hat{u}$ is chosen uniformly at
random from one of the coordinate directions.  The CHR walk is attractive in
part because the implementation of each step of the chain can potentially be
quite efficient: Smith~\cite[pp.~1302-1303]{smith_efficient_1984} already
mentioned some preliminary computational experiments of Telgen supporting such
an expectation in the important special case when $K$ is a polytope described by
a small number of sparse inequalities.  More recent computational work has also
explored the CHR walk in various application
areas~\cite{H+17,fallahi_comparison_2020}.  However, few theoretical guarantees
were known for the CHR walk, and it was only recently that Laddha and
Vempala~\cite{laddha_convergence_2021} and Narayanan and
Srivastava~\cite{narayanan_srivastava_2022} proved that with a $\poly{n}$-warm
start, the lazy CHR walk mixes in polynomial time.  The question of its mixing
time from a ``cold start'', i.e., from a $\exp(\poly{n})$-warm start, however
has remained open.

\subsection{Contributions} We construct a new family of Markov chains inspired
by classical multiscale decompositions of bounded sets of $\R^n$ into axis-aligned dyadic (i.e., of sidelength equal to a integral power of two) cubes.
Our chains $\MM_p$ are parameterized by the $\ell_p$ norms on $\R^n$,
$1 \leq p \le \infty$.  Our first contribution is to show that all of these
chains require only a polynomial (in $n$ and the aspect ratio $R/r$, as before)
number of steps to come within $\epsilon$ total variation distance of the
uniform distribution $\pi_K$, even when started with an $\exp(\poly{n})$-warm
start.  However, before describing the $\MM_p$ chains and our mixing result for
them in detail, we state the following result that we obtain as a byproduct of
their analysis. (Given the special status of the coordinate directions in the
coordinate hit-and-run walk, we parametrize the aspect ratio in terms of the
$\ell_\infty$ unit ball $B_\infty$ rather than in terms of the Euclidean unit
ball $B_2$.)
\begin{theorem}[see \cref{thm:chr-l2-mixing}]
  \label{thm:intro-chr}
  Let $K \subseteq \R^n$ be a convex body such that
  $r\cdot B_{\infty} \subseteq K \subseteq R\cdot B_\infty$.  Then starting from
  an $M$-warm start, the lazy coordinate hit-and-run walk comes within total
  variation distance at most $\epsilon$ of the uniform distribution $\pi_K$ on
  $K$ after $O\inp{n^9(R/r)^2\log(M/\epsilon)}$ steps.
\end{theorem}
The above result shows that the coordinate hit-and-run (CHR) random walk also
mixes in polynomial (in $n$ and the aspect ratio $R/r$) time even from ``cold'',
i.e., $\exp(\poly{n})$-warm starts.  As described above, polynomial time mixing
for the CHR walk had only been proved so far starting from a $\poly{n}$-warm
start~\cite{narayanan_srivastava_2022,laddha_convergence_2021}: the dependence
on $M$ in the mixing time bounds obtained in
\cite{narayanan_srivastava_2022,laddha_convergence_2021} are proportional to
$\poly{M}$, as compared to the $\log M$ dependence in our \cref{thm:intro-chr}.
The above result can also be extended to show that CHR mixes after a polynomial
number of steps even when the starting distribution is concentrated on a single
point in $K$, provided that the point is not too close to the boundary of $K$ --
its $\ell_\infty$-distance from the boundary of $K$ should be least
$R\exp(-\poly{n})$.  See \cref{sec:chr-point} for further discussion of this
extension.

We now proceed to describe our main technical result: the construction of the
$\MM_p$ random walks and their rapid mixing from a cold start.  The random walks
$\MM_p$ are inspired from the classical decomposition of bounded subsets of
$\R^n$ into axis-aligned cubes with disjoint interiors.  Such decompositions
have been used since the work of Whitney~\cite{Whitney} (see, e.g.,
\cite{Feff0,Feff} for more recent examples of their use).  We now informally
describe the decomposition of $K$ that we use for the $\MM_{p}$ chain.  For
simplicity, we assume that $K$ is contained in the interior of the
$\ell_{\infty}$ ball of radius $1$. We start with the standard tiling of
$\R^{n}$ by unit cubes with vertices in $\Z^n$, and also consider all scalings
of this tiling by factors of the form $2^{-k}$, where $k$ is a positive integer.
Our decomposition $\FF = \FF^{(p)}$ of $K$ into cubes with disjoint interiors is
then obtained by considering these cubes in decreasing order of sidelength and
including those cubes $Q$ for which
\begin{enumerate}
\item $Q$ is contained within $K$, and in fact, relative to
  its own diameter, $Q$ is ``far away'' from the exterior $\R^n \setminus K$ of
  $K$: the $\ell_{p}$-distance of the center of $Q$ from $\R^n\setminus K$ is at
  least \emph{twice} the $\ell_{p}$-diameter of $Q$, and
\item no ``ancestor'' cube of $Q$, i.e., a cube containing $Q$ is part of the
  decomposition $\FF$.
\end{enumerate}
A formal description of the construction of $\FF$ is given in
\cref{sec:whitney-cubes}, where it is also shown that such a decomposition fully
covers the interior $K^\circ$ of $K$, and also that if two cubes in $\FF$ abut
along an $(n-1)$-dimensional facet, then their sidelengths must be within a
factor of two of each other.  We note that this ``bounded geometry'': namely
that the ratio of the side lengths of abutting cubes are within a factor of two
of each other (see \cref{fig2}), is a very useful feature of this construction for our
purposes. In particular, this feature plays an important role in relating the
properties of the $\MM_p$ chains to the coordinate hit-and-run random walk.

\begin{figure}[H]
  \centering
  \begin{tikzpicture}
\coordinate (A1) at (0, 0, 0);
    \coordinate (A2) at (\Depth, 0, 0);
    \coordinate (A3) at (\Depth, 0, \Height);
    \coordinate (A4) at (0, 0, \Height);
    \coordinate (B1) at (0, \Width, 0);
    \coordinate (B2) at (\Depth, \Width, 0);
    \coordinate (B3) at (\Depth, \Width, \Height);
    \coordinate (B4) at (0, \Width, \Height);
    
    \draw (A1) -- (A2) -- (A3) -- (A4) -- cycle; \draw (A1) -- (B1) -- (B4) -- (A4) -- cycle; \draw (A1) -- (A2) -- (B2) -- (B1) -- cycle; \draw[opacity=0.8] (B1) -- (B2) -- (B3) -- (B4) -- cycle; \draw[opacity=0.8] (A2) -- (A3) -- (B3) -- (B2) -- cycle; \draw[opacity=0.8] (A3) -- (B3) -- (B4) -- (A4) -- cycle; 

    \coordinate (A01) at (-\Depth+0, 0, 0);
    \coordinate (A02) at (-\Depth+\Depth, 0, 0);
    \coordinate (A03) at (-\Depth+\Depth, 0, \Height);
    \coordinate (A04) at (-\Depth+0, 0, \Height);
    \coordinate (B01) at (-\Depth+0, \Width, 0);
    \coordinate (B02) at (-\Depth+\Depth, \Width, 0);
    \coordinate (B03) at (-\Depth+\Depth, \Width, \Height);
    \coordinate (B04) at (-\Depth+0, \Width, \Height);
    
    \draw (A01) -- (A02) -- (A03) -- (A04) -- cycle; \draw (A01) -- (B01) -- (B04) -- (A04) -- cycle; \draw (A01) -- (A02) -- (B02) -- (B01) -- cycle; \draw[opacity=0.8] (B01) -- (B02) -- (B03) -- (B04) -- cycle; \draw[opacity=0.8] (A02) -- (A03) -- (B03) -- (B02) -- cycle; \draw[opacity=0.8] (A03) -- (B03) -- (B04) -- (A04) -- cycle; 

    \coordinate (A11) at (\Depth, 0, 0);
    \coordinate (A12) at (\Depth+0.5*\Depth, 0, 0);
    \coordinate (A13) at (\Depth+0.5*\Depth, 0, 0.5*\Height);
    \coordinate (A14) at (\Depth, 0, 0.5*\Height);
    \coordinate (B11) at (\Depth, 0.5*\Width, 0);
    \coordinate (B12) at (\Depth+0.5*\Depth, 0.5*\Width, 0);
    \coordinate (B13) at (\Depth+0.5*\Depth, 0.5*\Width, 0.5*\Height);
    \coordinate (B14) at (\Depth, 0.5*\Width, 0.5*\Height);

    \draw (A11) -- (A12) -- (A13) -- (A14) -- cycle; \draw (A11) -- (B11) -- (B14) -- (A14) -- cycle; \draw (A11) -- (A12) -- (B12) -- (B11) -- cycle; \draw[opacity=0.6] (A13) -- (B13) -- (B14) -- (A14) -- cycle; \draw[opacity=0.6] (B11) -- (B12) -- (B13) -- (B14) -- cycle; \draw[opacity=0.6] (A12) -- (A13) -- (B13) -- (B12) -- cycle; 

    \coordinate (A21) at (\Depth+0,           0.5*\Width+0,          0);
    \coordinate (A22) at (\Depth+0.5*\Depth,  0.5*\Width+0,          0);
    \coordinate (A23) at (\Depth+0.5*\Depth,  0.5*\Width+0,          0.5*\Height);
    \coordinate (A24) at (\Depth+0,           0.5*\Width+0,          0.5*\Height);
    \coordinate (B21) at (\Depth+0,           0.5*\Width+0.5*\Width, 0);
    \coordinate (B22) at (\Depth+0.5*\Depth,  0.5*\Width+0.5*\Width, 0);
    \coordinate (B23) at (\Depth+0.5*\Depth,  0.5*\Width+0.5*\Width, 0.5*\Height);
    \coordinate (B24) at (\Depth+0,           0.5*\Width+0.5*\Width, 0.5*\Height);

    \draw (A21) -- (A22) -- (A23) -- (A24) -- cycle; \draw (A21) -- (B21) -- (B24) -- (A24) -- cycle; \draw (A21) -- (A22) -- (B22) -- (B21) -- cycle; \draw[opacity=0.6] (A23) -- (B23) -- (B24) -- (A24) -- cycle; \draw[opacity=0.6] (B21) -- (B22) -- (B23) -- (B24) -- cycle; \draw[opacity=0.6] (A22) -- (A23) -- (B23) -- (B22) -- cycle; 

    \coordinate (A31) at (\Depth+0,           0,          0.5*\Height+0);
    \coordinate (A32) at (\Depth+0.5*\Depth,  0,          0.5*\Height+0);
    \coordinate (A33) at (\Depth+0.5*\Depth,  0,          0.5*\Height+0.5*\Height);
    \coordinate (A34) at (\Depth+0,           0,          0.5*\Height+0.5*\Height);
    \coordinate (B31) at (\Depth+0,           0.5*\Width, 0.5*\Height+0);
    \coordinate (B32) at (\Depth+0.5*\Depth,  0.5*\Width, 0.5*\Height+0);
    \coordinate (B33) at (\Depth+0.5*\Depth,  0.5*\Width, 0.5*\Height+0.5*\Height);
    \coordinate (B34) at (\Depth+0,           0.5*\Width, 0.5*\Height+0.5*\Height);

    \draw (A31) -- (A32) -- (A33) -- (A34) -- cycle; \draw (A31) -- (B31) -- (B34) -- (A34) -- cycle; \draw (A31) -- (A32) -- (B32) -- (B31) -- cycle; \draw[opacity=0.6] (A33) -- (B33) -- (B34) -- (A34) -- cycle; \draw[opacity=0.6] (B31) -- (B32) -- (B33) -- (B34) -- cycle; \draw[opacity=0.6] (A32) -- (A33) -- (B33) -- (B32) -- cycle; 

    \coordinate (A41) at (\Depth+0,           0.5*\Width+0,          0.5*\Height+0);
    \coordinate (A42) at (\Depth+0.5*\Depth,  0.5*\Width+0,          0.5*\Height+0);
    \coordinate (A43) at (\Depth+0.5*\Depth,  0.5*\Width+0,          0.5*\Height+0.5*\Height);
    \coordinate (A44) at (\Depth+0,           0.5*\Width+0,          0.5*\Height+0.5*\Height);
    \coordinate (B41) at (\Depth+0,           0.5*\Width+0.5*\Width, 0.5*\Height+0);
    \coordinate (B42) at (\Depth+0.5*\Depth,  0.5*\Width+0.5*\Width, 0.5*\Height+0);
    \coordinate (B43) at (\Depth+0.5*\Depth,  0.5*\Width+0.5*\Width, 0.5*\Height+0.5*\Height);
    \coordinate (B44) at (\Depth+0,           0.5*\Width+0.5*\Width, 0.5*\Height+0.5*\Height);

    \draw (A41) -- (A42) -- (A43) -- (A44) -- cycle; \draw (A41) -- (B41) -- (B44) -- (A44) -- cycle; \draw (A41) -- (A42) -- (B42) -- (B41) -- cycle; \draw[opacity=0.6] (A43) -- (B43) -- (B44) -- (A44) -- cycle; \draw[opacity=0.6] (B41) -- (B42) -- (B43) -- (B44) -- cycle; \draw[opacity=0.6] (A42) -- (A43) -- (B43) -- (B42) -- cycle; \end{tikzpicture}

  \caption{Local geometry of Whitney decompositions: the sidelengths of adjacent
    cubes are within a factor of two of each other.}
  \label{fig2}
\end{figure}
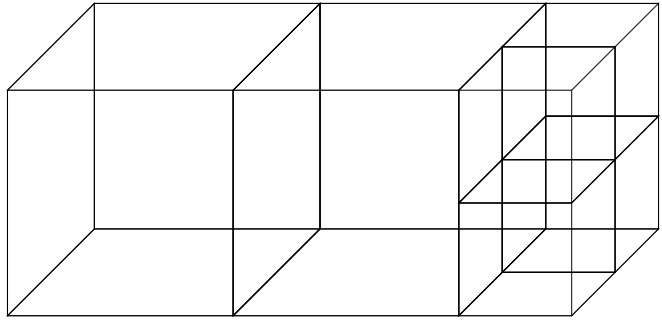

The chain $\MM_p$ can be seen both as a random walk on $K$ and also as a random
walk on the countably infinite set of cubes in the Whitney decomposition
$\FF = \FF^{(p)}$ of $K$ described above, but the latter view is easier to describe
first.  The stationary distribution $\pi$ of $\MM_p$ is given by
$\pi(Q) = \vol{Q}/\vol{K}$ for each cube in $\FF^{(p)}$.  Given the current cube
$Q$, the walk chooses to stay at $Q$ with probability $1/2$.  With the remaining
probability, it performs the following step. Pick a point $x$ uniformly at
random from the boundary $\bdry{Q}$ of the cube $Q$. With probability $1$, there
is a unique cube $Q' \neq Q$ in $\FF^{(p)}$ to which $x$ belongs.  The walk proposes
a move to this cube $Q'$, and then accepts it based on a standard Metropolis
filter with respect to $\pi$.  The Metropolis filter ensures that the walk is in
fact \emph{reversible} with respect to $\pi$, i.e.,
$\pi(Q)P_{\MM_p}(Q, Q') = \pi(Q')P_{\MM_p}(Q', Q)$ where $P_{\MM_p}(Q, Q')$ is
the probability of transitioning to cube $Q'$ in one step,
when starting from cube $Q$.  This implies that $\pi$ is a stationary
distribution of $\MM_p$ (see \cref{sec:markov-chain} for details).

As stated above, $\MM_p$ can equivalently be seen as a Markov chain on $K$
itself.  To see this, note that corresponding to any probability distribution
$\nu$ on $\FF^{(p)}$, there is a probability distribution $\nu_K$ on $K$
obtained by sampling a cube $Q$ from $\FF^{(p)}$ according to $\nu$, and then a
point $x$ uniformly at random from $Q$.  It is easy to see that the uniform
distribution $\pi_K$ on $K$ can be generated from the distribution $\pi$ above
in this fashion.  Further, one can also show that the total variation distance
between the distributions $\nu_K$ and $\pi_K$ on $K$ is at most the total
variation distance between the distributions $\nu$ and $\pi$ on $\FF^{(p)}$
(this follows directly from the definition, and the easy details are given in
the proof of \cref{thm:mp-mixing-intro} on page~\pageref{eq:91}).  Similarly,
given a probability distribution $\nu_K$ on $K$ that is $M$-warm with respect to
$\pi_K$, one can obtain a distribution $\nu$ on $\FF^{(p)}$ that is $M$-warm
with respect to $\pi$. This is done as follows. Sample a point $x$ according to
$\nu_K$. With probability $1$, $x$ lies in the interior of some cube
$Q \in \FF^{(p)}$ (this follows because $\nu_K$ is $M$-warm with respect to $\pi_K$
and because the probability measure under $\pi_K$ of the union of the boundaries
of the countably many cubes in $\FF^{(p)}$ is zero). $\nu$ is then defined to be
the probability distribution of this random cube $Q$.  Then
$\nu(Q) = \nu_K(Q) \leq M\pi_K(Q) = M\vol{Q}/\vol{K} = M\pi(Q)$.

Our main theorem for $\MM_p$ chains is the following.  Here,
$B_p \defeq \inb{x \in \R^n \st \norm[p]{x} < 1}$ is the unit $\ell_p$-ball in
$\R^n$: note that the requirement $r\cdot B_{p} \subseteq K$ is weaker than the
requirement $r\cdot B_{\infty} \subseteq K$ when $p < \infty$.
\begin{theorem}[see \cref{cor:mp-mixing}] Fix $1 \leq p \leq \infty$. Let
  $K \subseteq \R^n$ be a convex body such that
  $r\cdot B_{p} \subseteq K \subseteq R\cdot B_\infty$.  Then, starting from an
  $M$-warm start, the $\MM_p$ random walk on $K$ comes within total variation
  distance at most $\epsilon$ of the uniform distribution $\pi_K$ on $K$ after
  $O\inp{n^{4 + \frac{2}{p}}\cdot(R/r)^2\cdot \log (M/\eps)}$
  steps. \label{thm:mp-mixing-intro}
\end{theorem}
\begin{remark}
  In \cref{thm:mp-mixing-refined}, we use a more refined analysis to establish
  essentially the same dependence of the mixing time on $n$ and $R/r$ even when
  $\MM_p$ is started from a starting distribution supported on a single cube
  $Q \in \FF^{(p)}$ (or equivalently, in light of the discussion in the previous
  paragraph, a point $x \in Q$) that is at least $\frac{1}{\poly{n}}$ away from
  the boundary of $K$ (a direct application of \cref{thm:mp-mixing-intro} would
  lose an extra factor of $\tilde{O}(n)$ in this setting).
\end{remark}
 
\paragraph{Algorithmic implementation of $\MM_p$} We note that the tools we
develop for the analysis of our multiscale chains $\MM_p$ play a crucial rule in
our result for the CHR walk (\cref{thm:intro-chr}).  In addition, $\MM_p$ chains
are of algorithmic interest in their own right.  However, it may not be
immediately clear how to algorithmically implement each step of the $\MM_p$
chain from the above description of the chain and its state space $\FF^{(p)}$ of
Whitney cubes.  We show in \cref{sec:finding-whitney-cube} that each step of the
$\MM_1$ chain can be algorithmically implemented in $O(n)$ time using only a
membership oracle for $K$.  When $p > 1$, algorithmically implementing one step
of the $\MM_p$ chain requires access to an oracle for the $\ell_p$-distance of a
point $x \in K$ to the boundary $\bdry{K}$ of $K$: such oracles can be
implemented efficiently for polytopes.  We describe this construction as well in
\cref{sec:finding-whitney-cube}.

We now proceed to discuss the context for our results in the light of existing
literature. Following this, we give a overview of our results and proof
techniques in \cref{sec:technical-overview}.
\subsection{Discussion}
The notion of \emph{conductance} has played a central role in most rapid mixing
results for random walks on convex bodies.  For the discussion below, we fix a
convex body $K \subseteq \R^n$ such that $rB_2 \subseteq K \subseteq
RB_2$. Given a random walk $\mathcal{W}$ with stationary distribution as the
uniform distribution $\pi_K$ on $K$, the conductance $\Phi_{\mathcal{W}}(S)$ of
a subset $S \subseteq K$ is defined as the probability of the following randomly
chosen point lying in $K \setminus S$: choose a point uniformly at random from
$S$, and then take a step according to $\mathcal{W}$.  It follows from standard
results in the theory of Markov chains~\cite{LS93} that if
$\Phi_{\mathcal{W}}(S) \geq 1/\poly{n, R/r}$ for every measurable
$S \subseteq Q$, then the random walk $\mathcal{W}$ mixes rapidly from a
$\exp(\poly{n})$-warm start.  However, in several cases, one only gets the
weaker result that only large enough subsets have good conductance: a
formalization of this is through the notion of
\emph{$s$-conductance}~\cite[p.~367]{LS93}, which can capture the phenomenon
that, roughly speaking, the lower bound obtained on the conductance of $S$
\emph{degrades} as the volume $s$ of the set $S$ becomes smaller.  Under such a
bound, one usually only gets rapid mixing from a $\poly{n}$-warm start~(see,
e.g.,~\cite[Corollary 1.6]{LS93}). The reason that one can only get a lower
bound on the
conductance of large sets may have to do with the properties of the walk
$\mathcal{W}$ itself (which is the case with the ball walk).  However, it may
also be an artefact of the proof method rather than a property of the walk
itself.  For example, the original proof of Lovász~\cite{L99} for the rapid
mixing of the hit-and-run walk was built upon an $s$-conductance lower bound
that approached zero as the size parameter $s$ approached zero~\cite[Theorem 3]{L99}, and therefore
required a $\poly{n}$-warm start. In contrast, the later proof by Lovász and
Vempala~\cite{lovasz_hit-and-run_2006} established a conductance bound for the
same chain and thereby achieved rapid mixing from a cold start.

Rapid mixing proofs of random walks on convex sets often follow the plan of
establishing a conductance (or $s$-conductance) lower bound of the chain using
an \emph{isoperimetric inequality} for an appropriate metric (roughly speaking,
an isoperimetric inequality puts a lower bound on
$\vol{K \setminus (S_1 \cup S_2)}$ proportional to the product of volumes of
$S_1$ and $S_2$ and the distance $\delta$ between $S_1$ and $S_2$, at least when
$\delta$ is a sufficiently small positive number).\footnote{A notable exception
  to this general strategy is the work of Bubley, Dyer and
  Jerrum~\cite{bubley_elementary_1998}, discussed in more detail later in the
  introduction.}  A unifying theme in the analysis of many random walks for
sampling from convex sets, starting from the work of Lov\'{a}sz \cite{L99}, has
been to prove such an isoperimetric inequality when the underlying metric is
non-Euclidean. For example, the underlying metric in \cite{L99} is the
\emph{Hilbert metric} defined using the logarithm of certain cross-ratios.  This
isoperimetric inequality was then used to give an inverse polynomial lower bound
for the $s$-conductance of the hit-and-run walk that degraded gracefully to zero
as the size parameter $s$ approached zero, thereby leading to a rapid mixing
result for the hit-and-run walk under a warm start.  In later work, Lovász and
Vempala~\cite{lovasz_hit-and-run_2006} obtained an inverse-polynomial lower
bound on the conductance of the hit-and-run walk by refining the isoperimetric
inequality for the Hilbert metric proved in \cite{L99}: this improvement in the
isoperimetric inequality thus led to a rapid mixing result for the hit-and-run
walk without the need of a warm start.

The Hilbert metric also appears in the analysis by Kannan and
Narayanan~\cite{kannan2012random} of another random walk, called the \emph{Dikin
  walk}, on polytopes.  The Dikin walk was generalized by Narayanan~\cite{HN2}
to more general convex sets equipped with a weighted combination of logarithmic,
hyperbolic and self-concordant barriers, and was analysed using a different
Riemannian metric whose metric tensor is derived from the Hessian of the
combined barrier.  The isoperimetric properties of this Riemannian metric were
established by comparison to the Hilbert metric.  Improvements on this walk
with better mixing times have been obtained by Chen, Dwivedi, Wainwright and
Yu~\cite{chen2018fast} and by Laddha, Lee, and Vempala~\cite{Laddha}.  The
geodesic walk of Lee and Vempala~\cite{lee2017geodesic} uses geodesics of the
Riemannian metric associated with the logarithmic barrier to define a walk on
polytopes, whose properties again hinge on the isoperimetric properties of the
convex set equipped with the Hilbert metric and the uniform measure.

Beyond proving the isoperimetric inequality, there is also the need to relate
these Markov chains to the reference metric introduced. This was done for
hit-and-run in \cite{L99} using in part the well-known theorem of Menelaus in
Euclidean geometry.  This step for the
Dikin walk used facts from interior point methods developed by Nesterov and
Nemirovski. The analogous analysis was particularly involved in
\cite{lee2017geodesic} and used Jacobi fields among other tools.
For a more detailed discussion of these and related developments, we refer to
the recent survey~\cite{LV-survey} by Lee and Vempala.

Unfortunately, it has not been possible to exploit the Hilbert metric to analyze
the \emph{coordinate} hit-and-run (CHR) walk.  However, in recent work, Laddha
and Vempala~\cite{laddha_convergence_2021} showed how to implement the program
of proving an $s$-conductance bound for the CHR walk using an isoperimetric
inequality for an appropriate metric: they proved rapid mixing for the CHR walk
from a warm start via an isoperimetric inequality for subsets of $K$ that are
far in the $\ell_0$-metric and that are not too small in volume (the
$\ell_0$-distance between two points in $\R^n$ is the number of coordinates on
which they differ).  

Our result for the CHR walk (\cref{thm:intro-chr}) also hinges on a similar
$\ell_0$-isoperimetric inequality, \cref{theo: axis-disjoint isoperimetry},
which however extends to sets of all volumes (including arbitrarily small
volumes). This is the main technical ingredient that allows us to remove the
requirement of a warm start in \cref{thm:intro-chr}.

The proof of \cref{theo: axis-disjoint isoperimetry} itself goes via the proof
of a conductance lower bound for the $\MM_p$ chains on Whitney decompositions of
$K$ that we introduced above. The conductance analysis of the $\MM_p$ chains, in
turn, proceeds by introducing a kind of degenerate Finsler metric on $K$ (see
\cref{sec:an-isop-ineq}), which is a scaled version of $\ell_{\infty}$ that
magnifies distances in the vicinity of a point $x$ in $K$ by a factor of
$1/\dist_{\ell_p}(x, \R^n \setminus
K)$. Our main technical ingredient is a new isoperimetric inequality (\cref{theo
  isoperimetry}) for any convex body $K$ under such a metric. Part of the proof
of this inequality requires an existing isoperimetric inequality for convex sets
in normed spaces proved by Kannan, Lov\'{a}sz and Montenegro \cite{KLM06}, but
the bulk of the proof is handled by a detailed analysis of ``needles'' analogous
to those in the celebrated localization lemma of Lovász and
Simonovits~\cite{LS93}.  In the more refined analysis
(\cref{thm:mp-mixing-refined}) of the $\MM_p$ chain from a fixed state that we
alluded to in the remark following \cref{thm:mp-mixing-intro}, we also use
results of Lov\'{a}sz and Kannan~\cite{LK} relating rapid mixing to \emph{average conductance} rather
than worst-case conductance, thereby saving ourselves a factor of $\tilde{O}(n)$
in the mixing time. This in turn is made possible by the fact that for the
degenerate Finsler metric we introduce, the lower bounds we can prove on the
isoperimetric profile of small sets are actually \emph{stronger} than those we
can prove for large sets.

We now proceed to give a more detailed overview of our techniques.

 \subsection{Technical overview}
 \label{sec:technical-overview}
 Our result follows the general schema of establishing a conductance lower bound
 for the chain using an isoperimetric inequality for an appropriate metric.  As
 discussed above, the requirement of a warm start in rapid mixing proofs is
 often a consequence of the fact that non-trivial bounds for the conductance of
 the chain are available only for sets of somewhat large volumes.  This in turn is often
 due to having to ``throw away'' a part of the volume of $K$ that is close to
 the boundary $\bdry{K}$ of $K$ before applying the isoperimetric inequality:
 this is the case, for example, with the original warm start rapid mixing proof
 of the hit-and-run walk~\cite{L99}.  The same issue also arose in two different
 proofs of rapid mixing for the coordinate hit-and-run (CHR) walk starting with
 a warm start~\cite{laddha_convergence_2021,narayanan_srivastava_2022}: in both
 these proofs, an isoperimetric inequality could only be applied after excluding
 a part of $K$ close to $\bdry{K}$.

 Our motivation for considering a multiscale walk comes partly from the desire
 to avoid this exclusion of the part of $K$ close to its boundary.  Notice that
 as our multiscale chain $\MM_p$ approaches the boundary of $K$, the underlying
 cubes also become proportionately smaller, and the chain can still make
 progress to neighboring cubes at a rate that is not much worse than what it
 would be from larger cubes in the deep interior of the body.  Note, however,
 that this progress cannot be captured in terms of usual $\ell_p$ norms: while
 the chain does move to adjacent cubes, the distances between the centers of
 these adjacent cubes shrink as the chain comes closer to the boundary of $K$.
 Thus, it seems unlikely that isoperimetric results for $\ell_p$-norms
 alone (e.g., those in \cite{LS93,KLM06}) would be able to properly account for
 the progress the multiscale chain makes when it is close to the boundary of
 $K$.

\paragraph{A metric and an isoperimetry result}
In order to properly account for this progress, we introduce metrics that
magnify distances close to the boundary $\bdry{K}$ of $K$.  More concretely, to
analyze the chain $\MM_p$, we consider the metric $g_p$ which magnifies
$\ell_\infty$-distances in the vicinity of a point $x \in K$ by a factor of
$1/\dist_{\ell_p}(x, \bdry{K})$ (see \cref{sec:an-isop-ineq} for the formal
definition of the metric $g_p$).  Because of this scaling, this metric captures
the intuition that the chain's progress close to the boundary is not much worse than
what it is in the deep interior of $K$.  Our main technical result is an
isoperimetry result for $K$ endowed with the $g_p$ metric and the uniform
(rescaled Lebesgue) probability measure.  We show that
$\vol{K \setminus (S_1 \cup S_2)}$ is significant in proportion to
$\min\inb{\vol{S_1}, \vol{S_2}}$ whenever $S_1$ and $S_2$ are subsets of $K$
that are far in the $g_p$ distance: see \cref{theo isoperimetry} for the
detailed statement.

Our proof of \cref{theo isoperimetry} is divided into two cases depending upon
whether $S_1$, the smaller of the sets $S_1$ and $S_2$, has a significant mass
close to the boundary of $K$ or not.  The easy case is when $S_1$ does
\emph{not} have much mass close to the boundary, and in this case we are able to
appeal to a isoperimetric inequality of Kannan, Lovász and
Montenegro~\cite{KLM06} for the standard $\ell_p$ norms: this is Part 1
(page~\pageref{eq:83}) of the proof of \cref{theo isoperimetry}.

The case that requires more work is Part 2 (page~\pageref{eq:47}), which is when a large
constant fraction (about 0.95 in our proof) of the volume of $S_1$ lies within
$\ell_p$-distance $C_1/n$ of the boundary $\bdry{K}$ of $K$ for some parameter
$C_1$.  Our proof of this part is inspired by the localization idea of
\cite{LS93}, but we are unable to directly apply their localization
lemma in a black box manner. Instead, we proceed by radially fibering the
body $K$ into one-dimensional needles (see \cref{def:needle}), where the needles
correspond to radial line segments in a spherical polar coordinate system
centered at a point $x_0$ in the deep interior of $K$.  The intuition is that
since $S_1$ and $S_2$ are at distance at least $\delta > 0$ in the $g_p$ metric,
a large fraction of these needles contain a large segment intersecting
$S_3 = K \setminus (S_1 \cup S_2)$.  This intuition however runs into two
competing requirements.
\begin{enumerate}
\item First, the $S_3$-segment in a needle cannot be be too close to the
  boundary $\bdry{K}$. This is because the $g_p$ metric magnifies distances
  close to $\bdry{K}$, so that a segment that is close to $\bdry{K}$ and is of
  length $\delta$ in the $g_p$ metric may have a much smaller length in the
  usual Euclidean norm.  The contribution to the volume of $S_3$ of such a
  segment would therefore also be small.
\item Second, neither can the $S_3$-segment in a needle be too \emph{far} from
  the boundary $\bdry{K}$.  This is because, by definition, a needle $N$ is a
  radial line in a polar coordinate system centered at a point $x_0$ deep
  inside $K$, so that the measure induced on $N$ by the standard Lebesgue
  measure is proportional to $t^{n-1}$, where $t$ is the Euclidean distance from
  $x_0$.  Thus, the measure of an $S_3$-segment that lies close to the center
  $x_0$ of the polar coordinate system may be attenuated by a large factor
  compared to the measure of a segment of the same Euclidean length that lies
  closer to $\bdry{K}$.
\end{enumerate}

For dealing with these two requirements together, we consider the outer ``stub''
of each needle, which is the part of the needle starting from $\bdry{K}$ up to a
$C_2/n$ distance along the needle, where $C_2$ is an appropriate factor that
depends upon the needle (see \cref{eq:36} and \cref{def:needle} for the formal
definition).  For an appropriate choice of $C_1$ and $C_2$, we can show that for
at least a constant fraction of needles (see the definition of good needles
in \cref{eq:31}), the following conditions are simultaneously satisfied:
\begin{enumerate}
\item The stub of the needle contains a non-zero volume of $S_1$.
\item A large fraction of the inner part of the stub (i.e, the part farthest
  from the boundary) is not in $S_1$.
\end{enumerate}
For a formal description, see \cref{eq:46,eq:55} in conjunction with
\cref{fig:needle}.  Together, these facts can be used to show that the inner
part of the stub contains a large segment of $S_3$ (see page~\pageref{eq:57}).
This achieves both the requirements above: the segment of $S_3$ found does not
lie too close to the boundary (because it is in the inner part of the stub), but
is not too far from the boundary either (because the stub as a whole is quite
close to $\bdry{K}$ by definition).

\paragraph{Mixing time for the $\MM_p$ chains} We then show in
\cref{sec:bounding-conductance} that the isoperimetric inequality above implies
a conductance lower bound for the $\MM_p$ chain, in accordance with the
intuition outlined for the definition of the $g_p$ metric.  Rapid mixing from a
cold start (\cref{cor:mp-mixing}) then follows immediately from standard theory.
In \cref{sec:mixing-time-from}, we show that the fine-grained information that
one obtains about the conductance profile of the $\MM_p$ chain can be used to
improve the mixing time from a fixed state by a factor of $\tilde{O}(n)$ over
what the vanilla mixing time result from a cold start (\cref{cor:mp-mixing})
would imply.  We also show in \cref{prop:conductacne-lower-bound} that the
conductance lower bound we obtain for the $\MM_{p}$ chain is tight up to a
logarithmic factor in the dimension.

\paragraph{Rapid mixing from cold start for coordinate hit-and-run} Finally, we
prove rapid mixing from cold start for the coordinate hit-and-run (CHR) walk in
\cref{sec:coordinate-hit-run}.  As described above, two different proofs were
recently given for the rapid mixing for this chain from a warm
start~\cite{laddha_convergence_2021,narayanan_srivastava_2022}, and in both of
them, the bottleneck that led to the requirement of a warm start was a part of
the argument that had to ``throw away'' a portion of $K$ close to $\bdry{K}$.
In \cref{sec:coordinate-hit-run}, we show that the conductance (even that of
arbitrarily small sets) of the CHR walk can be bounded from below in terms of
the conductance of the multiscale chain $\MM_\infty$ (\cref{thm:CHR-cond}). As
discussed above, the conductance of the latter can be bounded from below using
the isoperimetry result for the $g_\infty$ metric. Together, this gives a rapid
mixing result for the CHR walk from a cold start
(\cref{thm:chr-l2-mixing,thm:intro-chr}).  To prove \cref{thm:CHR-cond}, we
build upon the notion of \emph{axis-disjoint} sets introduced by Laddha and
Vempala~\cite{laddha_convergence_2021}, who had proved an
``$\ell_0$-isoperimetry'' result for such sets.  However, as discussed above,
their isoperimetry result gives non-trivial conductance lower bounds only for
sets of somewhat large volume.  This was in part due their result being based on
a (partial) tiling of $K$ by cubes of \emph{fixed} sidelength, thereby
necessitating the exclusion of a part of the volume of the body close to the
boundary.  The main technical ingredient underlying our result for the CHR chain
is a new $\ell_0$-isoperimetry result for axis-disjoint sets (\cref{theo:
  axis-disjoint isoperimetry}) that applies to sets of all sizes, and that
involves the conductance of the multiscale $\MM_\infty$ chain described above.

Finally, we show in \cref{thm:chr-mixing-point} that the mixing result for CHR
from a cold start can be extended to show that CHR mixes in polynomially many
steps even when started from a point that is not too close to the boundary.
Roughly speaking, the mixing time scales with $\log(R/\delta)$ where $\delta$ is
the $\ell_\infty$-distance of the starting point to the boundary of the body $K$
and where $K \subseteq R\cdot B_\infty$.  This extension formalizes the
intuition that after about $O(n \log n)$ steps of CHR (i.e., when the chain has
had the opportunity, with high probability, to have made a step in each of the
$n$ coordinate directions), the resulting distribution is close to a ``cold
start'' in the sense of the mixing result in \cref{thm:intro-chr}.

\subsection{Open problems}
We conclude the introduction with a discussion of some directions for future
work suggested by this work.  The natural question raised by the application of
the $\MM_\infty$ walk to the analysis of the coordinate hit-and-run walk is whether
the $\MM_p$ chains, or the notion of multiscale decompositions in general, can
be used to analyze the rapid mixing properties of other random walks on convex
sets.

Another natural question is whether tools similar to ours can be used to further
improve the bound obtained on the mixing time of CHR -- we do not believe our
bounds to be tight. We lose a factor of $n^2$ in transferring an isoperimetric inequality for
axis-disjoint sets to a conductance bound for coordinate hit-and-run. A factor
of $n^3$ is lost when using the conductance bound of the $\MM_\infty$ chain to
prove an isoperimetric inequality for axis-disjoint sets; in particular, it
might be possible to tighten this to $n$ by circumventing the usage of Harper's
Theorem in Case 2(b) of the proof of \Cref{theo: axis-disjoint isoperimetry}.
This factor of $n^3$ comes partly from an invocation of a result of Laddha and
Vempala~\cite{laddha_convergence_2021} (\cref{lemma: cube isoperimetry} below).
In a manuscript that appeared after this paper had been circulated, Fernandez
V~\cite{FV23} has reported an improvement of this latter result, leading to an
improvement from a loss of $n^3$ to a loss of only $n^2$ in this step.

While the conductance bound we obtain for the $\MM_\infty$ chain is tight up to logarithmic factors (\Cref{prop:conductacne-lower-bound}), this does not preclude the possibility of an alternate Whitney cube-like decomposition which allows a less lossy reduction from coordinate hit-and-run. At a high level, the $\MM_\infty$ Markov chain is related to coordinate hit-and-run because when a point is far away from the boundary of the body, it is likelier to make large steps. In particular, when the sidelength of the Whitney cube containing a point is large, coordinate hit-and-run tends to make steps that are also (at least) this large. However, this does not capture the fact that coordinate hit-and-run can make larger steps in some directions. This leads to the question of whether it is possible to construct so-called ``Whitney cuboids'', where the cuboid $Q$ containing a point $x$ is such that the sidelength of $Q$ in the $i$th direction is comparable to the distance of $x$ from the boundary in the $i$th direction (for comparison, the Whitney cubes we use have sidelengths comparable to the \emph{minimum} over $i$ of the distance of $x$ from the boundary in the $i$th direction).

Another avenue is to explore the conductance profile of the $\MM_{p}$ chains at
sets of small volume.  In particular, if the right hand side of \cref{eq:10} in
\cref{theo isoperimetry} could be strengthened further to have an additional
factor of $\sqrt{\log (\vol{K}/\vol{S_3})}$, one could hope to use ``average
conductance''~\cite{LK,MS01,KLM06} methods to improve the dependence upon the
dimension $n$ in the proof of our mixing time result for the $\MM_p$ chains
starting from a point (\cref{thm:mp-mixing-refined}).  For the CHR chain, our
analysis currently does not use the conductance profile information we have for
the $\MM_p$ chains: it uses only the worst case conductance.  Improving the
analysis to take advantage of this extra information could also be a potential
direction for future work.\footnote{We thank an anonymous reviewer for
  suggesting this possibility.}

An alternative to the approach of using isoperimetric inequalities for analyzing
mixing times for random walks on convex sets is suggested by an interesting
paper of Bubley, Dyer and Jerrum \cite{bubley_elementary_1998}, where a certain
gauge transformation is used to push forward the uniform measure on a convex set
on to a log concave measure supported on $\R^n,$ whereafter a
Metropolis-filtered walk is performed using Gaussian steps. The analysis of this
walk (which mixes in polynomial time from a cold start, or even from the image,
under the gauge transformation, of a fixed point not too close to the boundary)
proceeds via a coupling argument, and does not use the program of relating the
conductance of the chain to an isoperimetric inequality.  Such coupling
arguments have also been very successful in the analysis of a variety of Markov
chains on finite state spaces.  It would be interesting to explore if a coupling
based analysis can be performed for the $\MM_p$ random walks or for the CHR
random walk.  Another possible approach to attacking these questions on rapid
mixing could be the localization schemes framework of Chen and
Eldan~\cite{CE22}.

Another direction for investigation would be to make the implementation of each
step of $\MM_p$, especially in the case $p > 1$, more efficient.  In the current
naive implementation of a step of the $\MM_p$ chain (when $p > 1$) on polytopes
that is described in \cref{sec:finding-whitney-cube}, the distances of a given
point $x$ to all the facets of the polytope are computed in order to find the
Whitney cube which contains $x$. In principle, it may be possible to ignore far
away facets as has been done by Mangoubi and Vishnoi~\cite{Mangoubi} in the
context of the ball walk, leading to savings in the implementation time.

One can also ask whether the dependence in \Cref{thm:intro-chr} on the aspect
ratio $R/r$ can be removed to get a polynomial dependence on only $n$ (and $\log
(M/\epsilon)$). Prior works do this using rounding procedures that linearly
transform the body to make it ``well-rounded'' with a polynomial aspect
ratio~(see, e.g., p.~409 of \cite{lovasz_simulated_2006}). However, one reason one might want to use coordinate hit-and-run specifically, as opposed to other Markov chains that mix in polynomial time, is that the coordinate directions have special significance in the context of the application (which might make each step of the Markov chain cheaper to implement, say). In particular, if we apply a rotational transformation to the body, coordinate hit-and-run in the transformed body might become more expensive to implement. Nevertheless, one can ask whether it is possible to construct a rounding scheme such that coordinate hit-and-run steps remain easy to implement.

Finally, it is natural to ask if this analysis of the multiscale chains and
coordinate hit-and-run can be extended to sampling log-concave densities.
Indeed, since log-concave measures can be supported on all of $\R^n$, extending
the proof in this direction axis would require a redefinition of the
notion of Whitney cubes we
use. However, we note here that in the special case of a log-concave measure
supported on a convex body for which the logarithm of the density has a
Lipschitz constant bounded by a polynomial in the dimension $n$, one can obtain a
$1/\poly{n}$ lower bound on the conductance of a natural Markov chain that uses
a Metropolis filter after using one step of (say) the $\MM_1$ chain as the
proposal distribution. We leave as an open problem whether the condition
of boundedness of the Lipschitz constant of the log-density can be removed from
this argument.
\section{Preliminaries}

\subsection{Markov chains}
We follow mostly the Markov chain notation used by Lovász and
Simonovits~\cite{LS93}, which we reproduce here for reference. For the following
definitions, let $\MM$ be a Markov chain on a state space $\Omega$, and let
$P(\cdot, \cdot) = P_{\MM}(\cdot, \cdot)$ denote the transition kernel of the
Markov chain.  Let $\pi$ be the stationary distribution of the chain: this means
that for any measurable subset $A \subseteq \Omega$,
\begin{equation}
  \label{eq:90}
  \int\limits_{x \in \Omega}\pi(dx)P(x, A) = \pi(A).
\end{equation}
A Markov chain is said to be \emph{lazy} if for every $x \in \Omega$,
$P(x, {x}) \geq 1/2$.  All Markov chains we consider in this paper will be lazy.
\begin{definition}[Ergodic flow and reversible chains]
  Given measurable subsets $A$ and $B$ of $\Omega$, the \emph{ergodic flow}
  $\Psi_{\MM}(A, B)$ is defined as
  \begin{equation}
    \Psi_{\MM}(A, B) \defeq \int\limits_{x \in A}\pi(dx) P(x, B).\label{eq:5}
  \end{equation}
  Informally, the ergodic flow from $A$ to $B$ is the probability of landing in
  $B$ after the following process: first sample a point from $A$ with ``weight''
  proportional to $\pi$, and then take one step of the chain.

  We also denote $\Psi_{\MM}(A, \Omega \setminus A)$ as $\Psi_{\MM}(A)$. Note
  that for \emph{any} Markov chain on $\Omega$ with stationary distribution
  $\pi$, $\Psi_{\MM}(A) = \Psi_{\MM}(\Omega \setminus A)$ (see, e.g.,
  \cite[Section 1.C]{LS93}). $\MM$ is said to be \emph{reversible} with respect
  to $\pi$ if $\Psi_{\MM}(A, B) = \Psi_{\MM}(B, A)$ for all measurable
  $A, B \subseteq \Omega$.\end{definition}

\begin{definition}[Conductance] Given a measurable subset $A$ of $\Omega$, the
  \emph{conductance} $\Phi_\MM(A)$ of $A$ is defined as $\Psi_\MM(A)/\pi(A)$.
  The \emph{conductance} $\Phi_\MM$ of $\MM$ is defined as the infimum of
  $\Phi_\MM(S)$ over all measurable $S \subseteq \Omega$ such that
  $\pi(S) \leq \frac{1}{2}$:
  \begin{equation}
    \label{eq:94}
    \Phi_{\MM} \defeq \inf_{\substack{S \st \pi(S) \leq 1/2}}\Phi_\MM(S).
  \end{equation}

\end{definition}\begin{definition}[Conductance profile]

  For $\a \in (0, 1/2]$, we define the value $\Phi_{\a, \MM}$ of the
  \emph{conductance profile} of $\MM$ at $\a$ as the infimum of
  $\Phi_{\MM}(S)$ over all measurable $S \subseteq \Omega$ such that
  $\pi(S) \leq \a$.

\end{definition}
When the underlying chain $\MM$ is clear from the context, we will drop the
subscript $\MM$ from the quantities in the above definitions.

\begin{definition}[Density and warmth]
  Given probability distributions $\pi$ and $\nu$ on $\Omega$, we say that $\nu$
  has \emph{density} $f$ with respect to $\pi$ if there is a measurable function
  $f: \Omega \rightarrow [0, \infty)$ such that for every measurable subset $A$
  of $\Omega$,
  \begin{equation}
    \label{eq:95}
    \nu(A) = \int\limits_{x \in A}f(x)\pi(dx).
  \end{equation}
  We will also use the notation $f\pi$ to denote the probability distribution
  that has density $f$ with respect to $\pi$.  Note that this implicitly
  requires that $\E_{X \sim \pi}[f(X)] = \nu(\Omega) = 1$.

  A probability distribution $\nu$ is said to be \emph{$M$-warm} with respect to
  $\pi$ if it has a density $f$ with respect to $\pi$ such that $f(x) \leq M$
  for all $x \in \Omega$.  Note also that $\pi$ has as density the constant
  function $\mathbf{1}$ with respect to itself.
\end{definition}

Given a probability distribution $\pi$ on $\Omega$, one can define the norms
$L^p(\pi)$, $1 \leq p \leq \infty$ on the set of bounded measurable real valued
functions on $\Omega$ as follows:
\begin{equation}
  \label{eq:96}
  \norm[L^p(\pi)]{f} \defeq  \inp{\;\;\int\limits_{x \in \Omega}\abs{f(x)}^p\pi(dx)}^{1/p}
  = \E_{X \sim \pi}\insq{\abs{f(X)}^p}^{1/p}.
\end{equation}
We will need only the norms $L^1(\pi)$ and $L^2(\pi)$ in this paper. If $\nu$ has
density $f$ with respect to $\pi$, then the \emph{total variation distance}
$d_{TV}(\nu, \pi)$ between $\nu$ and $\pi$ can be written as
\begin{equation}
  \label{eq:97}
  d_{TV}(\nu, \pi) = \sup_{A \subseteq Q}\abs{\pi(A) - \nu(A)} = \frac{1}{2}\norm[L^{1}(\pi)]{f - \mathbf{1}}.
\end{equation}
From Jensen's inequality we also have that for every bounded measurable $f$,
\begin{equation}
  \label{eq:98}
  \norm[L^1(\pi)]{f} \leq \norm[L^2(\pi)]{f}.
\end{equation}
Corresponding to $L^2(\pi)$, we also have the inner product
\begin{equation}
  \label{eq:99}
  \ina{f, g}_{L^2(\pi)} \defeq \int\limits_{x \in \Omega}f(x)g(x)\pi(dx) = \E_{X \sim \pi}[f(X)g(X)],
\end{equation}
so that $\ina{f, f}_{L^2(\pi)} = \norm[L^2(\pi)]{f}^2$.  Note that any Markov
chain $\MM$ can be seen as a linear operator acting on probability measures
$\nu$ on $\Omega$ as
\begin{equation}
  \label{eq:100}
  (\nu\MM)(A) \defeq \int\limits_{x\in\Omega}P_{\MM}(x, A)\nu(dx),\;\text{ for every measurable $A \subseteq \Omega$,}
\end{equation}
and also on real valued function $f$ on $\Omega$ as
\begin{equation}
  \label{eq:101}
  (\MM{}f)(x) \defeq \int\limits_{y\in\Omega}f(y)P_{\MM}(x, dy),\;\text{ for every $x \in \Omega$.}
\end{equation}
When $\MM$ is \emph{reversible} with respect to $\pi$, we have (see, e.g.,
\cite[eq.~(1.2)]{LS93})
\begin{equation}
  \label{eq:102}
  \ina{\MM{}f, g}_{{L^2(\pi)}} =   \ina{f, \MM{}g}_{{L^2(\pi)}},
\end{equation}
and also that the probability distribution $\nu\MM$ has density $\MM{}f$
with respect to $\pi$ when the probability distribution $\nu$ has density $f$
with respect to $\pi$.  We will need the following result of Lovász and
Simonovits~\cite{LS93} connecting the mixing properties of reversible chains to
their conductance (the result builds upon previous work of Jerrum and Sinclair
for finite-state Markov chains~\cite{jerrum_conductance_1988}).

\begin{lemma}[\textbf{\cite[Corollary 1.8]{LS93}}]\label{lem:ls-l2-mixing}
  Suppose that the lazy Markov chain $\MM$ on $\Omega$ is reversible with
  respect to a probability distribution $\pi$ on $\Omega$.  Let $\nu_{0}$ have
  density $\eta_{0}$ with respect to $\pi$, and define $\eta_{t}$ to be the
  density of the distribution $\nu_{t} = \nu_0\MM^{t}$ obtained after $t$ steps
  of the Markov chain starting from the initial distribution $\nu_0$.  Then
  \begin{displaymath}
    \norm[L^{2}(\pi)]{\eta_{t} - \mathbf{1}}^2
    \leq \inp{1 - \frac{\Phi^2}{2}}^{2t}\norm[L^{2}(\pi)]{\eta_{0} - \mathbf{1}}^2,
  \end{displaymath}
  where $\Phi$ is the conductance of $\MM$.
\end{lemma}
\begin{proof}
  Note that since $\MM$ is reversible with respect to $\pi$, the density
  $\eta_{t}$ of $\nu_t = \nu_0\MM^{t}$ with respect to $\pi$ is
  $\MM^{t}\eta_{0}$.  We now apply Corollary 1.8 of \cite{LS93} with $f$ in the
  statement of that corollary set to $\eta_0 - \mathbf{1}$, and $T$ set to $2t$.
  This ensures that $\E_{X \sim \pi}[f(X)] = 0$.  The corollary then gives
  \begin{equation}
    \label{eq:103}
    \ina{\eta_0- \mathbf{1}, \MM^{2t}(\eta_0 - \mathbf{1})}
    \leq \inp{1 - \frac{\Phi^2}{2}}^{2t}\norm[L^{2}(\pi)]{\nu_{0} - \mathbf{1}}^2.
  \end{equation}
  Now, by reversibility of $\MM$, we get (from \cref{eq:102}) that
  \begin{equation}
    \label{eq:104}
    \ina{\eta_0- \mathbf{1}, \MM^{2t}(\eta_0 - \mathbf{1})} = \ina{\MM^{{t}}(\eta_0- 1), \MM^{t}(\eta_0 - 1)}.
  \end{equation}
  The claim now follows since, as observed above, the reversibility of $\MM$
  implies that $\eta_t = \MM^t\eta_{0}$.
\end{proof}
\subsection{Geometric facts}
\paragraph{Notation} For any subset $S \subseteq \R^{n}$ we will denote by
$S^\circ$ its open interior (i.e., the union of all open sets contained in $S$),
and by $\bdry{S}$ the \emph{boundary} of $S$, defined as
$\bar{S} \setminus S^{\circ}$, where $\bar{S}$ is the closure of $S$.  Note that
$\bdry{S} \subseteq S$ if and only if $S$ is closed.  A \emph{convex body} in
$\R^{n}$ is a closed and bounded convex subset of $\R^{n}$ that is not contained
in any proper affine subspace of $\R^{n}$. By a \emph{cuboid} in $\R^n$, we mean
an $n$-dimensional hyper-rectangle, and by a \emph{cube} in $\R^n$ we mean an
$n$-dimensional hyper-rectangle all whose sides have the same length.  In this
paper, the cubes and cuboids we consider will typically be \emph{axis-aligned}:
their sides will all be parallel to the coordinate axes.

We will also need the following standard fact.

  \begin{lemma}
    \label{lem: l1 concave}
    Fix $p \geq 1$ and let $K$ be any convex body.  Then, the function
    $f : K \to \R$ defined by $f(x) = \dist_{\ell_p}(x,\partial K)$ is
    concave.
  \end{lemma}
  \begin{proof}
    Consider $x, y \in K$ such that $f(x) = \dist_{{\ell_p}}(x, \bdry{K}) = a$
    and $f(y) = \dist_{{\ell_p}}(y, \bdry{K}) = b$.
    If both $a$ and $b$ are zero then there is nothing to prove. 
    Otherwise, for any
    $\lambda \in (0, 1)$, let $z = \lambda x + (1-\lambda)y$.  Now, for any
    vector $v$ such that $\norm[p]{v} \leq \lambda a + (1-\lambda)b$, we have
    $z + v = \lambda x' + (1-\lambda)y'$ where
    $x' \defeq x + \frac{a}{\lambda a + (1-\lambda)b}\cdot v$ and
    $y' \defeq y + \frac{b}{\lambda a + (1-\lambda)b}\cdot v$.  By construction,
    $\dist_{\ell_p}(x, x') \leq a$ and $\dist_{\ell_{p}}(y, y') \leq b$ so that
    $x', y'$, and therefore $z + v$ are all elements of $K$. Since $v$ was an
    arbitrary vector with $\norm[p]{v} \leq \lambda a + (1-\lambda)b$, this
    shows that
    $f(z) = \dist_{{\ell_p}}(z, \bdry{K}) \geq \lambda a + (1-\lambda)b =
    \lambda f(a) + (1-\lambda)f(b)$.
\end{proof}

  In the proof of \cref{thm:mp-mixing-refined}, we will need the following
  well-known direct consequence of Cauchy's surface area formula (see, e.g.,
  \cite{tsukerman_brunn-minkowski_2017}).
\begin{proposition}\label{prop:cauchy-surface}
  Let $K$ and $L$ be convex bodies in $\R^n$ such that $K \subseteq L$.  Then
  $\vol[{{n-1}}]{\bdry{K}} \leq \vol[{{n-1}}]{\bdry{L}}$.
\end{proposition}

\section{Whitney decompositions}
\label{sec:whitney-cubes}
Hassler Whitney introduced a decomposition of an open set in a Euclidean space
into cubes in a seminal paper \cite{Whitney}.  The goal of this work was to
investigate certain problems involving interpolation. Such decompositions were
further developed by Calder\'{o}n and Zygmund \cite{CZ}. For more recent uses of
decompositions of this type, see Fefferman \cite{Feff0} and Fefferman and
Klartag \cite{Feff}. We begin with the procedure for constructing a Whitney
decomposition of $K^\circ$, \ie the interior of $K$, for the $\ell_p$-norm,
along the lines of Theorem 1, page 167 of \cite{Stein}.

As in the statement of \cref{thm:mp-mixing-intro} we assume that
$K \subseteq \{x \st \|x\|_\infty < R_\infty\}$ where $R_\infty < 1$ is a
positive real, and that $K \supseteq \{x \st \|x\|_p < r_p\}$ for some positive
real $r_p$.  The assumption $R_{\infty} < 1$ is made for notational
convenience, and can be easily enforced by scaling the body if necessary. We
discuss in a remark following \cref{theo 1.1} below as to how to remove this
assumption.

Consider the lattice of points in $\R^n$ whose coordinates are
integral. This lattice determines a mesh $\cube_0$, which is a collection of
cubes: namely all cubes of unit side length, whose vertices are points of the
above lattice. The mesh $\cube_0$ leads to an infinite chain of such meshes
$\{\cube_k\}^\infty_{0}$, with $\cube_k = 2^{-k} \cube_0$. Thus, each cube in
the mesh $\cube_k$ gives rise to $2^n$ cubes in $\cube_{{k+1}}$ which are termed
its \emph{children} and are obtained by bisecting its sides. The cubes in the
mesh $\cube_k$ each have sides of length $2^{-k}$ and are thus of
$\ell_p$-diameter $n^{\frac{1}{p}}2^{-k}$.

We now inductively define sets $\FF_{i} = \FF_{i}^{(p)}$, $i \geq 0$ as follows.  Let $\FF_0$
consist of those cubes $Q \in \cube_0$ for which
$\dist_{\ell_p}(\cntr(Q), K) \leq \frac{n^{1/p}}{2}$. Fix $\la = 1/2$. A cube
$Q$ in $\FF_k$ is subdivided into its children in $\cube_{k+1}$ if
\begin{equation}
  \la\dist_{\ell_p}(\cntr(Q), \partial K) < \diam_{\ell_p}(Q),\label{eq:Whitney}
\end{equation}
which are then declared to belong to $\FF_{k+1}$. Otherwise $Q$ is not divided and its children are not in $\FF_{k+1}$.

Let $\FF = \FF^{(p)} = \{Q_1, Q_2, \ldots, Q_k, \ldots\}$ denote the set of all
cubes $Q$ such that \ben \item There exists a $k$ for which
$Q \in \FF_k = \FF_{k}^{{(p)}}$ but the children of $Q$ do not belong to
$\FF_{k+1} = \FF_{{k+1}}^{{(p)}}$.
\item $\cntr(Q) \in K^\circ$. \een We will refer to $\FF^{(p)}$ as a
  \emph{Whitney decomposition} of $K$, and the cubes included in $\FF^{(p)}$ as
  \emph{Whitney cubes}.  In our notation, we will often suppress the dependence
  of $\FF^{(p)}$ on the underlying $\ell_p$ norm when the value of $p$ being
  used is clear from the context.  The following theorem describes the important
  features of this construction.
\begin{restatable}{theorem}{whitneythm}
\label{theo 1.1}
Fix $p$ such that $1 \leq p \leq \infty$. Let $R_{\infty} < 1$ and let $K
\subseteq R_{\infty}\cdot
B_{\infty}$ be a convex body. Then, the following statements hold true for the
Whitney decomposition $\FF = \FF^{{(p)}}$ of $K$.  \ben
\item \label{item:Whitney-cubes} $\bigcup_{Q \in \FF} Q = K^\circ$.  Further, if
  $Q \in \FF$, then $Q \not\in \cube_0$.
\item \label{item:disjoint-Whitney} The interiors $Q_k^\circ$ are mutually disjoint.
\item \label{item: diameter of cube center}For any Whitney cube $Q \in \FF$,
\[ 2\diam_{\ell_p}(Q) \le \dist_{\ell_p}(\cntr(Q),\Kb) \le \frac{9}{2} \diam_{\ell_p}(Q).\]
\item \label{item:diamater-of-cube} For any Whitney cube
  $Q \in \FF$ and $y \in Q$,
\[ \frac{3}{2}\diam_{\ell_p}(Q)\leq \dist_{\ell_p}(y, \Kb) \leq 5 \diam_{\ell_p}(Q).\]
In particular, this is true when $\dist_{\ell_p}(y,\Kb) = \dist_{\ell_p}(Q,\Kb)$.
\item\label{item:side-length-ratio} The ratio of sidelengths of any two abutting
  cubes lies in $\inb{1/2,1,2}$.
\een
\end{restatable}
The proof of this theorem can be found in \cref{sec:prop-whitn-cubes}.
\begin{remark}
  For notational simplicity, we described the construction of Whitney cubes
  above under the assumption that $K \subseteq R_\infty\cdot B_{\infty}$ with
  $R_\infty < 1$.  However, it is easy to see that this assumption can be done
  away with using a simple scaling operation.  If $R_\infty > 1$, let $2^a$ be
  the smallest integral power of two that is larger than $R_\infty$.  For any
  $p$ such that $1 \leq p \leq \infty$, denote by $\FF^{(p)}(K/2^{a})$ the
  Whitney decomposition of the scaled body $K/2^{a}$ (which can be constructed
  as above since $R_{\infty}/2^k<1$).  Now scale each cube in the decomposition
  $\FF^{(p)}(K/2^{a})$ up by a factor of $2^a$, and declare this to be the
  Whitney decomposition $\FF^{(p)}$ of $K$.  Since only linear scalings are
  performed, all properties guaranteed by \cref{theo 1.1} for
  $\FF^{(p)}(K/2^{a})$ remain true for $\FF^{(p)}$, except possibly for the
  property that unit cubes $Q \in \mathcal{Q}_0$ do not belong to $\FF^{(p)}$.
  Henceforth, we will therefore drop the requirement that $K$ has to be strictly
  contained in $B_{\infty}$ for it to have a Whitney decomposition
  $\FF^{(p)}$.
\end{remark}

\section{Markov chains on Whitney decompositions}
\label{sec:markov-chain}

Fix a convex body $K$ as in the statement of \cref{theo 1.1}, and a $p$ such
that $1 \leq p \le \infty$.  We now proceed to define the Markov chain $\MM_p$.

\paragraph{The state space and the stationary distribution} The state space of
the chain $\MM_p$ is the set $\FF = \FF^{(p)}$ as in the statement of \cref{theo
  1.1}.  The stationary distribution $\pi$ is defined as
\begin{equation}
  \label{eq:1}
  \pi(Q) \defeq \frac{\vol{Q}}{\vol{K}} \text{ for every $Q \in \FF$.}
\end{equation}
\paragraph{Transition probabilities}
In describing the transition rule below, we will assume that given a point $x$
which lies in the interior of an unknown cube $Q$ in $\FF^{{(p)}}$, we can
determine $Q$.  The details of how to algorithmically perform this operation are
discussed in the next section.

The transition rule from a cube $Q \in \FF$ is a lazy Metropolis filter,
described as follows.  With probability $1/2$ remain at $Q$. Else, pick a
uniformly random point $x$ on the boundary of $Q$. \Cref{item:diamater-of-cube}
of \cref{theo 1.1} implies that $x$ is in the interior $K^{o}$ of
$K$. Additionally, pick a point $x'$ such that
$\|x'-x\|_2 = \frac{\sidelen{Q}}{4}$ and $x'-x$ is parallel to the unique
outward normal of the face that $x$ belongs to.  With probability $1$, there is
a unique abutting cube $Q' \in \FF$ which also contains $x$.  By
\cref{item:side-length-ratio} of \cref{theo 1.1}, $Q'$ is also characterised by
being the unique cube in $\FF$ that contains $x'$ in its interior. If this
abutting cube $Q'$ has side length greater or equal to $Q$, then transition to
$Q'$.  Otherwise, do the following: with probability
$\frac{\sidelen{Q'}}{\sidelen{Q}}$ accept the transition to $Q'$ and with the
remaining probability remain at $Q$.

We now verify that this chain is reversible with respect to the stationary
distribution $\pi$ described in \cref{eq:1}.  Let $P(Q, Q') = P_{\MM_p}(Q, Q')$
denote the probability of transitioning to $Q' \in \FF$ in one step, starting
from $Q \in \FF$.  We then have
\begin{equation}
  \label{eq:2}
  P(Q, Q') = \frac{1}{2}
  \cdot \frac{
    \vol[n-1]{\bdry{Q} \cap \bdry{Q'}}
  }{
    \vol[n-1]{\bdry{Q}}
  }\cdot
  \min\inb{1,
    \frac{
      \sidelen{Q'}
    }{
      \sidelen{Q}
    }
  }.
\end{equation}
We thus have (since $\sidelen{Q}\cdot\vol[n-1]{Q} = 2n\cdot\vol[n]{Q}$)
\begin{align}
  \label{eq:3}
  \pi(Q)P(Q, Q')
  &= \frac{
    \vol{Q}
    }{
    \vol{K}
    }\cdot \frac{1}{2}
    \cdot \frac{
    \vol[n-1]{\bdry{Q} \cap \bdry{Q'}}
    }{
    \vol[n-1]{\bdry{Q}}
    }\cdot
    \min\inb{1,
    \frac{
    \sidelen{Q'}
    }{
    \sidelen{Q}
    }
    }\\
  &= \frac{1}{4n}
    \cdot
    \frac{
    \vol[n-1]{\bdry{Q} \cap \bdry{Q'}}
    }{
    \vol{K}
    }\cdot
    \min\inb{
    \sidelen{Q'},
    \sidelen{Q}
    }\\
  &= \pi(Q')P(Q', Q), \text{ by its symmetry in $Q$ and $Q'$}.
\end{align}

\subsection{Finding the Whitney cube containing a given point}
\label{sec:finding-whitney-cube}
The above description of our Markov chain assumed that we can determine the
Whitney cube $q \in \FF$ that a point $z \in K^\circ$ is contained in.  We only
needed to do this for points $z$ that are not on the boundary of such cubes, so
we assume that $z$ is contained in the interior of $q$.  In particular, this
implies that $q$ is uniquely determined by $z$ (by
\cref{item:Whitney-cubes,item:disjoint-Whitney} of \cref{theo 1.1}).

Suppose that $\sidelen{q} = 2^{-b}$, where $b$ is a currently unknown
non-negative integer.  Note that since $z$ lies in the interior of $q$, the
construction of Whitney cubes implies that given $b$ and $z$, $q$ can be
uniquely determined as follows: round each coordinate of the point $2^{b}z$ down
to its integer floor to get a vertex $v \in \Z^{n}$, and then take $q$ to be the
unique axis-aligned cube of side length $2^{{-b}}$ with center at
$2^{{-b}}(v + (1/2)\vec{1})$.  It thus remains to find $b$.

Assume now that we have access to an ``$\ell_p$-distance inequality oracle'' for
$K$, which, on input a point $x \in K$ and an algebraic number $\gamma$ answers
``YES'' if
\begin{equation*}
  \dist_{\ell_p}(x, \Kb) > \gamma
\end{equation*}
and ``NO'' otherwise, along with an ``approximate $\ell_{p}$-distance oracle'',
which outputs an $2^{\pm 0.01}$-factor multiplicative approximation $\tilde{d}$ of
$d(x) \defeq \dist_{{\ell_p}}(x, \Kb)$ for any input $x \in K^{o}$.  When
$p = 1$, such oracles can be efficiently implemented for any convex body $K$
with a well-guaranteed membership oracle.  However, they may be hard to
implement for other $p$ unless $K$ has special properties. We discuss this issue
in more detail in \cref{sec:ell_p-dist-oracl} below: here we assume that we have
access to such $\ell_p$-distance oracles for $K$.

Now, since $\tilde{d}$ is a $2^{\pm 0.01}$-factor multiplicative approximation
of $d(x)$, \cref{item:diamater-of-cube} of \Cref{theo 1.1} implies that
\begin{equation*}
  2^{-0.01}\cdot\frac{1}{5} \cdot \frac{\tilde{d}}{n^{1/p}} \le \sidelen{q}
  \le \frac{2}{3} \cdot\frac{\tilde{d}}{n^{1/p}}\cdot2^{0.01}.
\end{equation*}
Since $\sidelen{q} = 2^{-b}$, this gives
\begin{equation}
  \label{eq:84}
  \ceil{
    \log_2 \left( \frac{3n^{1/p}}{2d} \right)
    - 0.01 
  }\le
  b
  \le
  \floor{
    \log_2\inp{
      \frac{3n^{1/p}}{2d}
    }
    + \log_2\inp{
      \frac{10}{3}
    }
    + 0.01
  }.
\end{equation}
Let $b_{\min}$ and $b_{\max}$ be the lower and upper bounds in
\cref{eq:84}. Note that the range $[b_{\min},b_{\max}]$ has at most two
integers. We try both these possibilities for $b$ in decreasing order, and check
for each possibility whether the corresponding candidate $q$ obtained as in the
previous paragraph is subdivided in accordance with \cref{eq:Whitney}.  By the
construction of Whitney cubes, the first candidate $q$ that is \emph{not}
subdivided is the correct $q$ (and least one of the candidate cubes is
guaranteed to pass this check).  Note that this check requires one call to the
$\ell_p$-distance inequality oracle for $K$.

\subsubsection{\texorpdfstring{$\ell_p$-distance oracles for $K$}{lp distance
    oracles for K}}
\label{sec:ell_p-dist-oracl}
When $p = 1$, the distance oracle can be implemented to $O(2^{-L})$ precision as
follows. Given a point $x \in K^\circ$, consider for each canonical basis vector
$e_i$ and each sign $\sigma = \pm 1$, the ray
$\{y \st (y - x) \in \R_+ \cdot \sigma e_i\}$. The intersection $y_{i, \sigma}$
of this ray with the boundary of the convex set can be computed to a precision
of $O(2^{-L})$ using binary search and $L + O(1)$ calls to the the membership
oracle. The $\ell_1$ distance to the complement of $K$ from $x$ equals
$\min(\|y_{i, \sigma}- x\|_1 \st i \in [n], \sigma \in \{-1,1\})$, provided all
the $\|y_{i, \sigma} - x\|_1$ are finite and the point $x$ is not in $K$
otherwise.  The implementation of the $\ell_1$-distance inequality oracle also
follows from the same consideration: for $x \in K$,
$\dist_{\ell_1}(x, \bdry{K}) > \gamma$ if and only if all of the points
$\inb{x + \sigma \gamma e_{i} \st 1 \leq i \leq n, \sigma \in \inb{-1,1}}$ are
in $K^\circ$.

When $p >1$, and $K$ is an arbitrary convex body, there is a non-convex
optimization involved in computing the $\ell_p$-distance.  However, for
polytopes with $m$ faces with explicitly given constraints, the following
procedure may be used.

We compute the $\ell_p$-distance to each face and then take the minimum. These
distances have a closed form expression given as follows. Let $K$ be the
intersection of the halfspaces $H_i$, where $H_i$ is given by
$\{y \st a_i \cdot (y - x) \leq 1\}$. The $\ell_p$-distance of $x$ to
$\R^n \setminus H_i$ is given by \beq \inf\limits_{y \in \R^n\setminus
  H_i} \norm[p]{y - x} = \norm[q]{a_i}^{-1},\label{eq:4-nov-1-hari}\eeq for
$1/p + 1/q = 1$.  To see (\ref{eq:4-nov-1-hari}), note that for any $a_i$,
equality in $\|y - x\|_p \cdot \|a_i\|_q \geq 1$, can be achieved by some $y$ in
$\overline {\R^n\setminus H_i}$ by the fact that equality in H\"{o}lder's
inequality is achievable for any fixed vector $a_i$.

\paragraph{A note on numerical precision} Since we are only concerned with walks
that run for polynomially many steps, it follows as a consequence of the fact
that the ratios of the side lengths of abutting cubes lie in
$\{\frac{1}{2}, 1, 2\}$ (\cref{item:side-length-ratio} of \cref{theo 1.1}) that
the distance to the boundary cannot change in the course of the run of the walk
by a multiplicative factor that is outside a range of the form
$[\exp(n^{-C}), \exp(n^C)]$, where $C$ is a constant.  Due to this, the number
of bits needed to represent the side lengths of the cubes used is never more
than a polynomial in the parameters $n, R/r$ and $M$ in
\cref{thm:mp-mixing-intro}, and thus $L$ in the description above can also be
chosen to be $\poly{n, R/r, \log(M)}$ in order to achieve an ``approximate
$\ell_1$ distance oracle'' of the form considered in
\cref{sec:ell_p-dist-oracl}.

\section{Analysis of Markov chains on Whitney decompositions}
\label{sec:an-isop-ineq}
\subsection{An isoperimetric inequality}
In this subsection, we take the first step in our strategy for proving a lower
bound on the conductance of the $\ell_p$-multiscale chain $\MM_{p}$, which is to
equip $K$ with a suitable metric and prove an isoperimetric inequality for the
corresponding metric-measure space coming from the uniform measure on $K$. We
then relate (in \cref{sec:bounding-conductance}) the conductance of the chain to
the isoperimetric profile of the metric-measure space.

The metric we introduce is a kind of degenerate Finsler metric, in which the
norms on the tangent spaces are rescaled versions of $\ell_\infty$, by a factor
of $\dist_{\ell_p}(x,\partial K)^{-1}$ so that the distance to the boundary of
$K$ in the local norm is always greater than
$\Omega\left(n^{-\frac{1}{p}}\right)$. In order to prove the results we need on
the isoperimetric profile, we need to lower bound the volume of a tube of
thickness $\delta$ around a subset $S_1$ of $K$ whose measure is less than
$1/2$. This is done by considering two cases. First, if $S_1$ has a strong
presence in the deep interior of $K,$ we look at the intersection of $S_1$ with
an inner parallel body, and get the necessary results by appealing to existing
results of Kannan, Lov\'{a}sz, and Montenegro \cite{KLM06}. The case when $S_1$
does not penetrate much into the deep interior of $K$ constitutes the bulk of
the technical challenge in proving this isoperimetric inequality.  We handle
this case by using a radial needle decomposition to fiber $S_1,$ and then
proving on a significant fraction of these needles (namely those given by
\cref{eq:31}) an appropriate isoperimetric inequality from which the desired
result follows.  We now proceed with the technical details.

Equip $K$ with a family of Minkowski functionals
$F_p: K^\circ \times \R^n \rightarrow \R_+$, $p \geq 1$, defined by
\begin{equation}
  \label{eq:25}
  F_p(x, v) \defeq (\dist_{\ell_p}(x,\partial K))^{-1} \|v\|_\infty
\end{equation}
for each $x \in K^\circ$ and $v \in \R^n$.  Note that each $F_p$ is a continuous
map that satisfies $F_p(x, \alpha v) = \abs{\alpha} F_p(x, v)$ for each
$x \in K^\circ, v \in \R^n$ and $\alpha \in \R$. Given this, the length
$\len_{g_p}(\gamma)$ (for each $p \geq 1$) of any piecewise continuously differentiable curve
$\gamma: [0, 1] \rightarrow K^\circ$, is defined as
\begin{equation}
  \label{eq:15}
  \len_{g_p}(\gamma) \defeq \int\limits_0^1{F_p(\gamma(t), \gamma'(t))} \; dt.
\end{equation}
(Note that the length of a curve defined as above does not change if the curve
is re-parameterized.)  This defines a metric on $K^\circ$ as usual: for
$x, y \in K^\circ$,
\begin{equation}
  \label{eq:16}
  \dist_{g_p}(x, y) \defeq \inf_{\gamma} \len_{g_p}(\gamma),
\end{equation}
where the infimum is taken over all piecewise continuously differentiable curves
$\gamma: [0, 1] \rightarrow K^\circ$ satisfying $\gamma(0) = x$ and $\gamma(1) = y$.

We are now ready to state our isoperimetric inequality.

\begin{theorem}
  \label{theo isoperimetry}
  There exist absolute positive constants $C_{0}$, $C_1$ and $C_2$ such that the
  following is true.  Let $K$ be a convex body such that
  $r_{p}B_p \subseteq K \subseteq R_{\infty}B_{\infty}$.  Let $S_1,S_2,S_3$ be a
  partition of $K$ into three parts such that $\dist_{g_p}(S_1,S_2) > \delta$, and
  $\vol{S_1} \leq \frac{1}{2} \vol{K}$. Define
  $\rho_p \defeq r_p/R_{\infty} \leq 1$. Then, for $\delta \leq 1$, we have the
  following:
if $\vol{S_1} > \exp(-C_{0}n) \cdot \vol{K}$ then
    \begin{align}
      \vol{S_3}
      & \geq C_1\cdot\frac{\rho_p}{n}
        \cdot \delta
        \cdot \vol{S_1}
        \cdot \log\inp{1 + 0.9\frac{\vol{K}}{\vol{S_1}}}.\label{eq:24}
    \end{align}
    and if $\vol{S_1} \leq \exp(-C_{0}n)\cdot \vol{K}$ then
    \begin{align}
      \vol{S_3}
      & \geq C_2\cdot{\rho_p}
        \cdot \delta
        \cdot \vol{S_1}.\label{eq:10}
    \end{align}
  \end{theorem}

  In the proof of \cref{theo isoperimetry}, we will need to consider needles
  analogous to those that appear in the localization lemma of \cite{LS93}. As
  described above, however, we will have to analyze such needles in detail in
  part 2 of the proof. We therefore proceed to list some of their properties
  that will be needed in the proof.
  
  \begin{definition}[\textbf{Needle}]\label{def:needle} Fix some $x_0 \in K$
    such that
    $\dist_{\ell_p}(x_0,\partial K) = \max_{x \in K} \dist_{\ell_p}(x,\partial
    K) \geq r_p$.  By a \emph{needle}, we mean a set
    \[ N_{u} \defeq K \cap \left\{ x_0 + tu \st t \ge 0 \right\} \] where
    $u \in \bS^{n-1}$ is a unit vector. Also define
    $\ell_2(N_u) \defeq \sup\{ t \st x_0 + tu \in K \}$ and in general,
    $\ell_p(N_u) \defeq \ell_2(N_u) \cdot \|u\|_p$.
  \end{definition}
  Note that by the choice of $x_{0}$, ${\ell_p(N) \ge r_p}$ for any $p$ and
  any needle $N$. Similarly, we also have $\ell_{\infty}(N) \leq 2R_{\infty}$.

  Let $\mathcal{N}$ denote the set of all needles.  Clearly, $\mathcal{N}$ is in
  bijection with $\bS^{n-1}$, and we will often identify a needle with the
  corresponding element of $\bS^{{n-1}}$. Let $\sigma$ denote the uniform (Haar)
  probability measure on $\bS^{n-1}$, and $\omega_n$ the $(n-1)$-dimensional
  surface area of $\bS^{n-1}$.  Then, for any measurable subset $S$ of $K$, we
  can use a standard coordinate transformation to polar coordinates followed
  by Fubini's theorem (see, e.g., \cite[Corollary 2.2]{burgisser13:_condit} and
  \cite[Theorem 3.12]{coarea}) to write
  \begin{align}
    \vol{S}
    &= \int\limits_{x \in \R^n \setminus \inb{x_0}}I_S(x) dx \label{eq:73}\\
    &=\omega_{n}\int\limits_{r = 0}^{\infty}
      \int\limits_{\hat{u} \in \bS^{n-1}}
      I_{S}(x_{0} + r \hat{u}) r^{{n-1}} dr \,\sigma(d\hat{u})\label{eq:74}\\
    &=\frac{\omega_{n}}{n}
      \int\limits_{\hat{u} \in \bS^{n-1}}\ell_{2}\inp{{N_{\hat{u}}}}^{n}
      \cdot \mu_{{N_{\hat{u}}}}(S \cap {N_{\hat{u}}}) \sigma(du),\label{eq:75}
  \end{align}
  where for any needle $N_{\hat{u}}$, the probability measure
  $\mu_{{N_{\hat{u}}}}$ on $N_{\hat{u}}$ is defined as
  \begin{align}
    \label{eq:72}
    \mu_{{N_{\hat{u}}}}(S \cap {N_{\hat{u}}})
    &\defeq
      \frac{n}{\ell_{2}\inp{N_{{\hat{u}}}}^{n}}
      \cdot \int\limits_{r = 0}^{\infty}I_{S}(x_{0} + r \hat{u}) r^{{n-1}} dr\\
    &=
      \frac{n}{\ell_{2}\inp{N_{{\hat{u}}}}^{n}}
      \cdot \int\limits_{r = 0}^{\ell_2(N_{\hat{u}})}I_{S}(x_{0} + r \hat{u})
      r^{{n-1}} dr
      \quad
      \text{for every measurable } S \subseteq K.
  \end{align}
  More generally, Fubini's theorem also yields the following.  Let $A$ be a
  Haar-measurable subset of $\bS^{n-1}$, $S$ a measurable subset of $K$, and set
  \begin{equation}
    \label{eq:76}
    T = S \cap \bigcup_{\hat{u} \in A}N_{\hat{u}}.
  \end{equation}
  Then $T$ itself is measurable and
  \begin{align}
    \frac{\omega_{n}}{n}
    \int\limits_{\hat{u} \in \bS^{{n-1}}}I_{A}(\hat{u})\cdot \ell_{2}\inp{{N_{\hat{u}}}}^{n}
    \cdot \mu_{{N_{\hat{u}}}}(S \cap {N_{\hat{u}}}) \sigma(d\hat{u})
    &=\omega_{n}\int\limits_{r = 0}^{\infty}
      \int\limits_{\hat{u} \in \bS^{n-1}}
      I_{A}(\hat{u}) I_{S}(x_{0} + r \hat{u}) r^{{n-1}} dr \sigma(d\hat{u})\label{eq:77}\\
    &=\omega_{n}\int\limits_{r = 0}^{\infty}
      \int\limits_{\hat{u} \in \bS^{n-1}}
      I_{T}(x_{0} + r \hat{u}) r^{{n-1}} dr \sigma(d\hat{u})\label{eq:78}\\
    &= \vol{T}.\label{eq:79}
  \end{align}

  The following alternative description of $\mu_{N_{{\hat{u}}}}$ will be useful.
  Let us label the points in $N = N_{{\hat{u}}}$ by their $\ell_p$ distance,
  along $N$, from the boundary $\bdry{K}$: thus $x_0$ is labelled $\ell_p(N)$.  In
  what follows, we will often identify, without comment, a point $x \in N$ with
  its label $\nlab[N](x) \in [0, \ell_p(N)]$. Note that
  $\nlab[N](x_{0} + t\hat{u}) = \ell_{p}(N)\cdot(1 - t/\ell_{2}(N))$.  By a
  slight abuse of notation, we will denote the inverse of the bijective map
  $\nlab[N]$ by $N$.  Thus, the label of the point $N(x) \in N$, where
  $x \in [0, \ell_p(N)]$, is $x$.  The pushforward of $\mu_{N}$ under the
  bijective label map $\nlab[N]$ is then a probability measure on
  $[0, \ell_{p}(N)]$ with the density
  \begin{equation}
    \label{eq:42}
    \tilde{\mu}_{N}(x) = \frac{n}{\ell_p(N)} \inp{1 - \frac{x}{\ell_p(N)}}^{n-1}.
  \end{equation}

  We are now ready to begin with the proof of \cref{theo isoperimetry}.
  \begin{proof}[Proof of \cref{theo isoperimetry}]
Let $c_1 \leq c_2 < 1$ and $\beta < \alpha \leq 1/2$ be positive constants
    to be fixed later.  The proof is divided into two parts, based on the value
    of $\Pr[x \sim S_1]{\dist_{\ell_p}(x,\partial K) \ge c_1 r_p/n}$.
\paragraph*{Part 1: }
Suppose that $\Pr[x \sim S_1]{\dist_{\ell_p}(x,\partial K) \ge c_1 r_p/n} \geq \beta$.\\
In this case, let
    \begin{equation}
      K' \defeq \{ x \in K :
      \dist_{\ell_p}(x,\partial K) \geq c_1 r_p/n
      \}.\label{eq:83}
  \end{equation}
  \Cref{lem: l1 concave} then implies that $K'$ is convex, and further that $x_0
  + \inb{\inp{1 - c_1/n}(z - x_0) \st z \in K} \subseteq
  K'$.  This inclusion implies that
    \begin{equation}
      \label{eq:17}
      \vol{K'} \geq 0.95\,\vol{K}
    \end{equation}
    {provided $c_1 \leq 0.05$}. Now, by the assumption in this part,
    \begin{equation}
      \label{eq:19}
      \frac{\vol{S_1 \cap K'}}{\vol{S_1}} = \p_{x \sim S_1}[x \in K'] \ge \beta.
    \end{equation}
    If $\vol{S_3} \ge \frac{1}{2} \vol{S_1}$, then we already have the required
    lower bound on the volume of $S_3$.  So we assume that
    $\vol{S_3} \le \frac{1}{2} \vol{S_1}$, and get
    $\vol{S_1 \cup S_3} \leq \frac{3}{4}\vol{K}$ (where we use the fact that
    $\vol{S_1} \leq \frac{1}{2}\vol{K}$). From \cref{eq:17}, we therefore get
    \begin{equation}
      \label{eq:18}
      \vol{S_2 \cap K'} \ge \vol{K'}  - \vol{S_1 \cup S_3} \ge \frac{1}{8} \vol{K}.
    \end{equation}
    Note also that since $\dist_{\ell_p}(x, \bdry{K}) \geq c_1r_p/n$ for every
    $x \in K'$, we have
    \begin{equation}
      \label{eq:20}
      \dist_{\ell_\infty}(S_1 \cap K', S_2 \cap K')
      \geq \frac{c_1r_p}{n}
      \cdot \dist_{g_p}(S_1 \cap K', S_2 \cap K')
      \geq \frac{c_1r_p}{n}
      \cdot\dist_{g_p}(S_1, S_2) \geq \frac{c_1r_p\delta}{n}.
    \end{equation}
The isoperimetric constant of $K'$ can now be bounded from below using a
    ``multiscale'' isoperimetric inequality of Kannan, Lovász and Montenegro
    (Theorem 4.3 of \cite{KLM06}), applied with the underlying norm being the
    $\ell_\infty$ norm.\footnote{ The inequality in \cite{KLM06} is stated for
      distance and diameters in $\ell_2$ only, but exactly the same proof works
      when the distance and the diameters are both in $\ell_\infty$ (or in any
      $\ell_p$ norm), because the final calculation is on a straight line, just
      as is the case for \cite{LS93}.  See \cref{sec:isop-ineq-kann} for
      details.}  Applying this result to $K'$, along with
    \cref{eq:20,eq:19,eq:18} we get
    \begin{align}
      \label{eq:21}
      \vol{S_3}
      &\ge  \vol{S_3 \cap K'}\\
      &\ge \frac{c_1r_p \delta}{2nR_\infty}\cdot
        \frac{\vol{S_1\cap K'}\vol{S_2\cap K'}}{\vol{K'}}
        \cdot \log\inp{
        1 +
        \frac{\vol{K'}^2}{\vol{S_1\cap K'}\vol{S_2\cap K'}}
        }\\
      &\geq \frac{\beta c_1r_p \delta}{16nR_\infty}\cdot
        \vol{S_1}
        \cdot \log\inp{
        1 +
        \frac{\vol{K'}}{\vol{S_1}}
        }\\
      &\geq \frac{\beta c_1r_p \delta}{16nR_\infty}\cdot
        \vol{S_1}
        \cdot \log\inp{
        1 +
        \frac{9}{10}\cdot\frac{\vol{K}}{\vol{S_1}}
        }\text{, using \cref{eq:17}.}\label{eq-isoperimetry-inner-multi-scale}
    \end{align}

\paragraph*{Part 2: }
Suppose now that $\Pr[x \sim S_1]{\dist_{\ell_p}(x,\partial K) \geq c_1 r_p /n} < \beta$.\\
In this case, for any needle $N$, define the sets
    \begin{align}
      N_{in}
      &\defeq
        N \cap \inb{x \st \dist_{\ell_p}(x, \bdry{K}) > c_1r_p/n}
        \text{, and}\label{eq:47}\\
      N_{out}
      &\defeq
        N \cap \inb{x \st \dist_{\ell_p}(x, \bdry{K}) \leq c_1r_p/n},\label{eq:48}
    \end{align}
    and consider the set $B_{\alpha, c_1}$ of \emph{\sbad{}} needles, defined
    as a subset of $\bS^{n-1}$ as
    \begin{equation}
      B_{\alpha, c_1} \defeq \inb{\hat{u} \st 
        \mu_{N_{\hat{u}}}\inp{S_1 \cap N_{\hat{u},in}}
        \geq \alpha \cdot \mu_{N_{\hat{u}}}(N_{\hat{u}} \cap S_1)}.\label{eq:27}
    \end{equation}
    Note that $B_{{\alpha, c_{1}}}$ is a measurable subset of $\bS^{{n-1}}$.
    Integrating this inequality over all such needles using the formula in
    \cref{eq:79}, we get
    \begin{align}
      \alpha \cdot \vol{
        S_{1} \cap \bigcup_{\hat{u} \in B_{\alpha, c_1}}N_{\hat{u}}
      }
      & \leq \vol{
        S_1
        \cap
        \inb{
          x \st
          \dist_{\ell_p}(x, \bdry{K}) > {c_1r_p}/{n}
        }
        \cap
        \bigcup_{\hat{u} \in B_{\alpha, c_1}}N_{\hat{u}}
      }\\
      &\leq \vol{
        S_1
        \cap
        \inb{
          x \st
          \dist_{\ell_p}(x, \bdry{K}) > {c_1r_p}/{n}
        }
      }.\label{eq:28}
    \end{align}
    By our assumption for this case, we have
    \begin{equation}
      \label{eq:29}
      \frac{\vol{
        S_1
        \cap
        \inb{
          x \st
          \dist_{\ell_p}(x, \bdry{K}) > {c_1r_p}/{n}
        }}
      }{
        \vol{S_1}
      }
      = \Pr[x \sim S_1]{\dist_{\ell_p}(x,\partial K) > c_1 r_p /n}
      < \beta.
    \end{equation}
Substituting \cref{eq:29} in \cref{eq:28} gives
    \begin{equation}
      \label{eq:30}
      \vol{S_1 \cap \bigcup_{\hat{u} \in B_{\alpha, c_1}}N_{\hat{u}}}
      \leq \frac{\beta}{\alpha} \cdot \vol{S_1}.
    \end{equation}
    which shows that when {$\beta$ is small enough compared to $\alpha$}, a
    point sampled randomly from $S_1$ is unlikely to land in \asbad{} needle.
Let $G_{\alpha, c_1}$ be the set of \emph{good} needles defined as follows
    (again, as a subset of $\bS^{n-1}$):
    \begin{equation}
      \label{eq:31}
      G_{\alpha, c_1}
      \defeq
      \bS^{n-1} \setminus B_{\alpha,c_1}
      = \left\{\hat{u} \in \bS^{n-1}  \st
        \mu_{N_{\hat{u}}} \inp{N_{\hat{u}, in} \cap S_1}
        < \alpha\mu_{N_{\hat{u}}}(N_{\hat{u}} \cap S_1) \right\}  .
\end{equation}
\Cref{eq:30} then gives that
\begin{equation}
      \vol{S_1 \cap \bigcup_{\hat{u} \in G_{\alpha, c_1}}N_{\hat{u}}} \ge \left( 1 -
        \frac{\beta}{\alpha} \right) \vol{S_1}.\label{eq:33}
    \end{equation}
Our goal now is to show that for every $\hat{u} \in G_{\alpha, c_1}$, we
    have for the needle $N = N_{\hat{u}}$
    \begin{equation}
      \label{eq:32}
      \mu_N(S_3 \cap N) \geq C \cdot \mu_N(S_1 \cap N)
    \end{equation}
    for some $C = C(K)$.  Indeed, given \cref{eq:32}, we get that (identifying
    each needle with the corresponding element in $\bS^{n-1}$ for ease of
    notation)
    \begin{align}
      \vol{S_3}
      &\ge \vol{S_3 \cap \bigcup_{N \in G_{\alpha,c_1}} N} \\
      &\stackrel{\text{\cref{eq:79}}}{=}
        \frac{\omega_n}{n} \int_{G_{\alpha,c_1}} \ell_2(N)^n \cdot \mu_N(S_3 \cap N) \sigma(dN) \\
      &\stackrel{\text{\cref{eq:32}}}{\ge}
        C \cdot \frac{\omega_n}{n} \int_{G_{\alpha,c_1}} \ell_2(N)^n \mu_N(S_1 \cap N) \sigma(dN) \\
      &\stackrel{\text{\cref{eq:79}}}{=}
        C \cdot \vol{S_1 \cap \bigcup_{N \in G_{\alpha,c_1}} N} \\
&\stackrel{\text{\cref{eq:33}}}{\ge}
        C \left( 1 - \frac{\beta}{\alpha} \right) \vol{S_1}. \label{eq:35}
    \end{align}
We now proceed to the proof of \cref{eq:32}.  Given a needle
    $N \in G_{\alpha, c_1}$ as above and $c_2 \geq c_1$, define \stub[c_2]{N} as
    \begin{equation}
      \label{eq:36}
\stub[c_2]{N} \defeq \inb{N(x) \st x \in [0, c_2{\ell_p(N)}/n]}.
    \end{equation}
    In particular, for every $\gamma \in [0, 1]$ we have
    \begin{equation}
      \label{eq:43}
      \frac{\gamma}{2} \leq \mu_N(\stub[\gamma]{N}) = 1 - \inp{1 -\frac{\gamma}{n}}^n \leq \gamma.
    \end{equation}
    Note also that since $N(\ell_p(N)) = x_0$ with $\dist_{\ell_p}(x_0,\partial K) \geq r_p$, and
    $N(0) \in \bdry{K}$, \cref{lem: l1 concave} implies that
    \begin{equation}
      \label{eq:44}
      \dist_{\ell_p}\inp{N\inp{\frac{\gamma\ell_p(N)}{n}}, \bdry{K}}
      \geq \frac{\gamma r_p}{n}.
    \end{equation}
    More generally, the concavity of $\dist_{\ell_p}(\cdot, \bdry{K})$ along
    with the fact that $\dist_{\ell_p}(x_0, \bdry{K}) \geq r_p > c_1r_p/n$ implies
    that the labels of the points in the sets $N_{in}$ and $N_{out}$ form a
    partition of the interval $[0, \ell_p(N)]$ into disjoint intervals, with
    $x_0 = N(\ell_p(N)) \in N_{in}$ and $N(0) \in N_{out}$.  Further, from
    \cref{eq:44}, we get that
    \begin{align}
      \label{eq:50}
      N_{out} &\subseteq N([0, c_1\ell_p(N)/n)]) \subseteq \stub[c_2]{N},\;\text{and}\\
      N_{{in}}&\supseteq N((c_{1}\ell_{p}(N)/n, \ell_{p}(N)]).\label{eq:59}
    \end{align}
    whenever {$c_1 \leq c_2$}.
\paragraph{Some estimates} We now record some estimates that follow directly
    from the above computations. Recall that {$c_{1} \leq c_{2}$}.  Let
    $\tilde{N}$ denote $\stub[c_2]{N}$.  Similarly, define (see
    \cref{fig:needle})
    \begin{equation}
      \label{eq:60}
      \tilde{N}_{{out}} \defeq \stub[c_1]{N},
      \qquad \text{and} \qquad
      \tilde{N}_{{in}} \defeq \stub[c_2]{N} \setminus \stub[c_1]{N}.
    \end{equation}
    \begin{figure}[t]
      \centering
      \includegraphics[width=\textwidth]{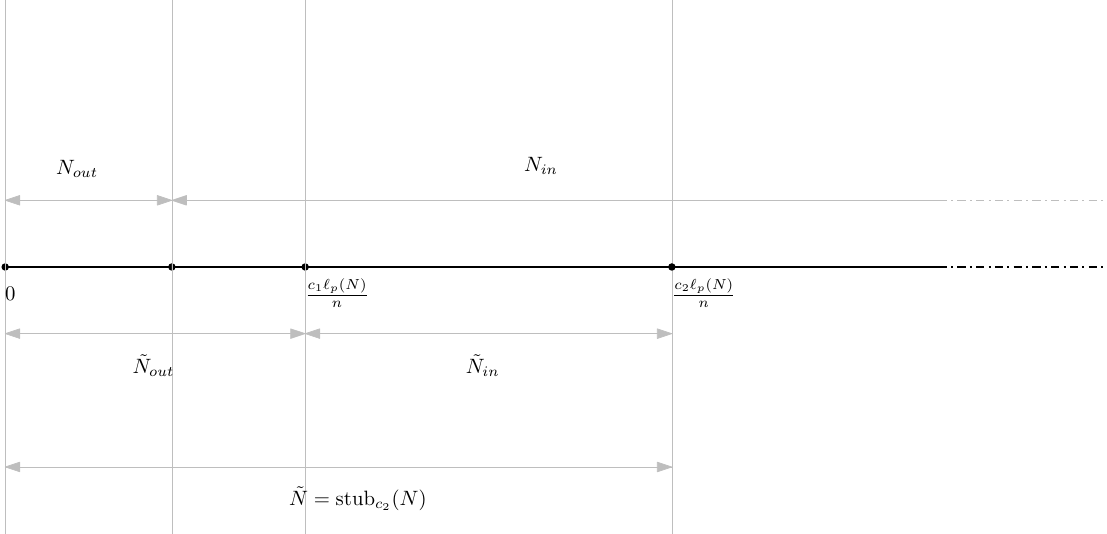}
      \caption{Various parts of a needle}
      \label{fig:needle}
    \end{figure}
    From \cref{eq:50,eq:59}, we then get that
    \begin{equation}
      \label{eq:61}
      N_{{out}} \subseteq \tilde{N}_{{out}} 
      \qquad \text{and} \qquad
      \tilde{N}_{{in}} \subseteq \tilde{N} \cap N_{in}.
    \end{equation}
    We then have
    \begin{align}
      \label{eq:51}
      \mu_N(N_{out})
      &\stackrel{\text{\cref{eq:61}}}{\leq}
        \mu_{N}(\tilde{N}_{out})
        \stackrel{\text{\cref{eq:43}}}{\leq}
        c_1,\;\\
      \label{eq:52}
      \mu_N(N_{in})
      &\stackrel{\text{\cref{eq:51}}}{\geq}
        1 - c_1,\; \text{since $N_{in}$ and $N_{out}$ form a partition of $N$, and}\\
      \label{eq:53}
      \mu_N(N_{in}\cap \tilde{N})
      &\stackrel{\text{\cref{eq:61}}}{\geq}
        \mu_{N}(\tilde{N}_{in})
        = \mu_N(\tilde{N}) - \mu_N(\tilde{N}_{out}) 
        \stackrel{\text{\cref{eq:43}}} \geq \frac{c_2}{2} - c_1.
    \end{align}
    Now, when $N \in G_{\alpha, c_{1}}$, we also have
    \begin{align}
      \label{eq:54}
      \mu_{N}(\tilde{N}_{in} \cap S_{1})
      \stackrel{\text{\cref{eq:61}}}{\leq}
      \mu_{N}(N_{{in}} \cap \tilde{N} \cap S_{1})
      &\leq \alpha, \;\text{and}\\
      \label{eq:55}
      \mu_{N}(N_{{in}} \cap \tilde{N} \cap (S_{2} \cup S_{3}))
      \stackrel{\text{\cref{eq:61}}}{\geq}
      \mu_{N}(\tilde{N}_{{in}} \cap (S_{2} \cup S_{3}))
      &\geq \frac{c_{2}}{2} - c_{1} - \alpha.
    \end{align}
    Here, \cref{eq:54} follows from \cref{eq:31} because
    $N \in G_{{\alpha, c_{1}}}$.  \Cref{eq:55} then follows from
    \cref{eq:53,eq:54}.  Along with the fact that $N_{{in}}$ and $N_{out}$ form
    a partition of $N$, \cref{eq:31} for $N \in G_{{\alpha, c_{1}}}$ also yields
    \begin{equation}
      \label{eq:56}
      \mu_{N}(S_{1} \cap N) < \frac{\mu_{N}(S_{1} \cap N_{out})}{1-\alpha}
      \leq \frac{\mu_{N}(N_{out})}{1-\alpha}
      \stackrel{\text{\cref{eq:51}}}{\leq} \frac{c_{1}}{1-\alpha}.
    \end{equation}

\paragraph{Proving \cref{eq:32}.} Since the conclusion of \cref{eq:32} is trivial
    when $\mu_N(S_1 \cap N) = 0$, we assume that $\mu_N(N \cap S_1) > 0$. As
    $N \in G_{\alpha, c_1}$, \cref{eq:56} yields
    \begin{equation}
      \label{eq:46}
      \mu_{N}(\tilde{N}_{out} \cap S_{1})
      \stackrel{\text{\cref{eq:61}}}{\geq}
      \mu_N\inp{N_{out} \cap S_1}
      \stackrel{\text{\cref{eq:56}}}{>}
      (1-\alpha)\mu_N(N \cap S_1)
      > 0.
    \end{equation}
    We now have two cases.
    \begin{description}
    \item[{Case 1:}] $\mu_{N}(\tilde{N}_{in} \cap S_{2}) = 0$.  In
      this case, \cref{eq:55} implies that
      \begin{equation}
        \label{eq:57}
        \mu_{N}(S_{3}\cap N) \geq \mu_{N}(S_{3} \cap \tilde{N}_{in})
        \geq \frac{c_{2}}{2} - c_{1} - \alpha.
      \end{equation}
      Combined with \cref{eq:56}, this gives
      \begin{equation}
        \label{eq:58}
        \mu_{N}(S_{3}\cap N) \geq (1-\alpha)\cdot\frac{c_{2} - 2(c_{1} + \alpha)}{2c_{1}}\cdot\mu_{N}(S_{1} \cap N).
      \end{equation}
    \item[{Case 2:}] $\mu_{N}(\tilde{N}_{in} \cap S_{2}) > 0$.  In this case, we
      define
      \begin{equation}
        t' \defeq \inf\inb{x \st N(x\ell_{p}(N)/n) \in \tilde{N}_{in} \cap S_{2}}
        = \inf\inb{x \in (c_{1}, c_{2}) \st N(x\ell_{p}(N)/n) \in   S_{2}}.\label{eq:62}
      \end{equation}
      and note that the assumption for the case means that $t'$ exists and
      satisfies $t' \geq c_{1}$.  We then define
      \begin{equation}
        \label{eq:49}
        s \defeq \sup\inb{x < t' \st N(x\ell_{p}(N)/n) \in S_{1}}.
      \end{equation}
      It follows from \cref{eq:46} that $s$ is well defined.  Again using
      \cref{eq:46} followed by the definition of $s$, we then have
      \begin{equation}
        \label{eq:63}
        (1-\alpha)\mu_{N}(S_{1} \cap N)
        \leq \mu_{N}(S_{1} \cap \tilde{N}_{{out}})
        \leq \mu_{N}(\stub[s]{N})
        \stackrel{\text{\cref{eq:43}}}{\leq}
        s.
      \end{equation}
      Now, define
      \begin{equation}
        t \defeq \inf\inb{x > s \st N(x\ell_{p}(N)/n) \in S_{2}}.\label{eq:66}
      \end{equation}
      Note that $s \leq t \leq t'$, and the open segment of $N$ between the
      points $N(s\ell_{p}(N)/n)$ and $N(t\ell_{p}(N)/n)$ is contained in
      $S_{3}$.  Thus,
      \begin{equation}
        \label{eq:68}
        \mu_{N}(S_{3}\cap N)
        \geq \mu_{N}\inp{
          N\inp{
            \inp{
              s\ell_{p}(N)/n,
              t\ell_{p}(N)/n
            }}}.
      \end{equation}
      Further, since $\dist_{g_{p}}(S_{1}, S_{2}) \geq \delta$, we get that the
      $g_{p}$ length of the segment from $N(s\ell_{p}(N)/n)$ to
      $N(t\ell_{p}(N)/n)$ along $N$ must also be at least $\delta$. From
      \cref{eq:44}, we see that for any point $\tau$ on this segment,
      \begin{equation}
        \label{eq:65}
        \dist_{\ell_{p}}(\tau, \bdry{K}) \geq \frac{sr_{p}}{n}.
      \end{equation}
      Using the definition of the $g_{p}$ metric (see \cref{eq:25,eq:15,eq:16})
      we therefore get
      \begin{equation}
        \label{eq:64}
        \delta
        \leq \frac{
          \dist_{{\ell_{\infty}}}\inp{N(s\ell_{p}(N)/n), N(t\ell_{p}(N)/n)}
        }{
          \frac{sr_{p}}{n}
        }
\leq
        \frac{\ell_{\infty}(N)}{r_{p}}\frac{t - s}{s},
      \end{equation}
      where $\ell_\infty(N)$ is the $\ell_\infty$ length of $N$ (this is because
      the segment from $N(s\ell_{p}(N)/n)$ to $N(t\ell_{p}(N)/n)$ of $N$
      constitutes a $(t-s)/n$ fraction of the length of $N$, in any $\ell_q$
      norm).  Rearranging, we get
      \begin{equation}
        \label{eq:67}
        t - s \geq \frac{r_{p}}{\ell_{\infty}(N)}\cdot s\delta.
      \end{equation}
      Now, a direct calculation using \cref{eq:43} and the convexity of the map
      $x \mapsto \inp{1 - x/n}^{n}$ gives
        \begin{align}
          \label{eq:69}
          \mu_{N}(S_{3}\cap N)
          &\geq
            \mu_{N}\inp{
            N\inp{\inp{
            s\ell_{p}(N)/n,
            t\ell_{p}(N)/n}
            }},\;\text{from \cref{eq:68}},\\
          &=
            (1 - s/n)^{n} - (1-t/n)^{n},\;\text{from \cref{eq:43}} \\
          &\geq (1 - t/n)^{n-1}(t - s)
            \geq (1 - c_{2})(t-s),\;\text{since }t \leq c_{2},\\
          &\geq (1-c_2)\cdot\delta\cdot\frac{r_{p}}{\ell_{\infty}(N)}\cdot s,
            \;\text{from \cref{eq:67}}\\
          &\geq (1-\alpha)\cdot(1-c_2)\cdot\delta\cdot\frac{r_{p}}{\ell_{\infty}(N)}\cdot\mu_N(S_1\cap N),
            \;\text{from \cref{eq:63}}.\label{eq:70}
        \end{align}
      \end{description}
      Combining \cref{eq:58,eq:70} we therefore get that \cref{eq:32} holds with
      \begin{equation}
        \label{eq:71}
        C = \min\inb{
          (1-\alpha)\cdot(1-c_2)\cdot\delta\cdot\frac{r_{p}}{\ell_{\infty}(N)},\;
          (1-\alpha)\cdot\frac{c_{2} - 2(c_{1} + \alpha)}{2c_{1}}
        }.
      \end{equation}
      We can now choose $c_1 = 0.05, c_2 = 0.5, \alpha = 0.1$ and
      $\beta = 0.05$.  Then, since $\delta \leq 1$, and
      $\ell_{\infty}(N) \leq 2R_{\infty}$, the right hand side above is at least
      $C' \rho_{p} \delta$ for some absolute constant $C'$ (recall from the
      statement of the theorem that $\rho_p = r_p/R_{\infty} \leq 1$).

      We now combine the results for the two parts
      (\cref{eq-isoperimetry-inner-multi-scale}, and \cref{eq:35,eq:71} and the discussion in the previous paragraph, respectively)
      to conclude that there exist positive constants $C_1'$ and $C_{2}'$ such
      that
      \begin{equation}
        \label{eq:82}
        \vol{S_3} \geq \min\inb{C_{1}', \frac{1}{n}\log\inp{1 +
            0.9\cdot\frac{\vol{K}}{\vol{S_1}}}}\cdot C_{2}' \rho_{p} \delta \cdot \vol{S_1}.
      \end{equation}
      The existence of constants $C_0, C_1$ and $C_2$ as in the statement of the
      theorem follows immediately from \cref{eq:82}, by considering when each of
      the two quantities in the minimum above is the smaller one.
    \end{proof}

\subsection{Bounding the conductance}
\label{sec:bounding-conductance}
In this subsection, we prove \cref{thm:conductance} which gives a lower bound on
the conductance of the $\MM_p$ random walks on Whitney cubes described earlier.
In the proof, we will need the following two geometric lemmas, whose proofs can
be found in \cref{sec:geometry}.  In \cref{sec:tightn-cond-bound}, we further
show that in the worst case, the conductance lower bound we obtain here for the
$\MM_{p}$ random walks is tight up to a factor of $O(\log n)$, where $n$ is the
dimension.

\begin{restatable}{lemma}{surface}
  Let $K$ be a convex body, $\FF$ a Whitney decomposition of it as described in
  \cref{sec:whitney-cubes}, and consider any set $S \subseteq \FF$.  As before,
  we identify $S$ also with the union of cubes in $S$.  For any cube $Q \in S$,
  we have \label{lem:surface-area-infty}
  \begin{equation*}
    \vol[n-1]{\bdry{S} \cap \bdry{Q}} = \lim_{\epsilon \downarrow
      0}\frac{\vol[n]{(Q + \epsilon B_\infty)\setminus S}}{\epsilon}.
  \end{equation*}
\end{restatable}
Recall that the general definition of surface area uses Minkowski sums with
scalings of the Euclidean unit ball $B_{2}$.  The reason we can work instead
with the $\ell_{\infty}$-unit ball $B_{\infty}$ in \cref{lem:surface-area-infty}
is because all the surfaces involved are unions of finitely many axis-aligned
cuboidal surfaces.

The following lemma relates distances in the $\ell_{\infty}$-norm to distances
in the $g_{p}$ metric defined before the statement of \cref{theo isoperimetry}.
\begin{restatable}{lemma}{metricdist}
  \label{lem-small-metric-distances} Fix a convex body $K \subseteq \R^n$.  Then there exists
  $\delta = \delta(n, p)$ such that for all
  $\epsilon \in [0, \delta]$, and all points $x, y \in K^\circ$ with
  \begin{equation*}
    \norm[\infty]{x - y}
    \geq \epsilon
    \cdot 
      \dist_{\ell_p}(x, \bdry{K}),
  \end{equation*}
  it holds that
  \begin{equation*}
    \dist_{g_p}(x, y) \geq \frac{\epsilon}{2}.
  \end{equation*}
\end{restatable}
We are now ready to state and prove our conductance lower bound.
\begin{theorem}\label{thm:conductance} Fix $p$ such that $1 \leq p \leq \infty$. Let $K$ be a convex
  body such that
  $r_p\cdot B_{p} \subseteq K \subseteq R_\infty\cdot B_{\infty}$. Define
  $\rho_p \defeq r_p/R_\infty \leq 1$ as in the statement of \cref{theo
    isoperimetry}.  The conductance $\Phi = \Phi_{\MM_p}$ of the chain $\MM_p$
  on the Whitney decomposition $\FF^{(p)}$ of $K$ satisfies
  \begin{equation}
    \label{eq:86}
    \Phi \geq \frac{\rho_p}{O(n^{2 + \frac{1}{p}})}.
  \end{equation}
More precisely, letting
  $C_0$ be as in the statement of \cref{theo isoperimetry},
 the  conductance profile $\Phi_\a$ for $\a > \exp(-C_{0}n)$ satisfies
  \begin{equation}
   \label{eqn: multiscale conductance profile - 1}
       \Phi_\a \geq  \frac{\rho_p}{O(n^{2 + \frac{1}{p}})}
    \cdot \log\inp{1 + \frac{0.9}{\a}},
  \end{equation}
 and
   for $\a \leq \exp(-C_0 n),$ $\Phi_\a$ satisfies
  \begin{equation}
    \label{eqn: multiscale conductance profile - 2}
    \Phi_\a \geq  \frac{\rho_p}{O(n^{1 + \frac{1}{p}})}.
  \end{equation}
\end{theorem}
\begin{proof} Let $S \subseteq \FF = \FF^{(p)}$ be such that $\pi(S) \le
  (1/2)$. We shall often also view $S$ as the subset of $K$ corresponding to the
  union of the cubes in it.  For each $Q \in S$ and $\epsilon > 0$, consider the
  set $Q_\epsilon$ defined by
  \begin{equation}
    Q_\epsilon = \bigcup_{x \in Q} \left( x + {2\epsilon\dist_{\ell_p}(x,\partial
        K)}B_\infty \right).\label{eq:9}
  \end{equation}
Then, by \cref{item:diamater-of-cube} of \Cref{theo 1.1}, we have that
  $\dist_{\ell_p}(x,\bdry{K}) \leq 5 \diam_{\ell_p}(Q)$.  Thus, we get that for
  some absolute constant $C > 0$
  \[ Q_\epsilon \subseteq Q + C\epsilon \diam_{\ell_p}(Q)B_\infty = Q + \epsilon
    \cdot C n^{1/p} \sidelen{Q}B_\infty . \] Consequently, using
  \cref{lem:surface-area-infty},
  \begin{align}
    \limsup_{\epsilon \downarrow 0} \frac{\vol{Q_\epsilon \setminus S}}{\epsilon}
    & \leq \lim_{\epsilon \downarrow 0} \frac{\vol{(Q + C\epsilon \diam_{\ell_p}(Q)B_\infty) \setminus S}}{\epsilon} \\
    &
      = C\diam_{\ell_p}(Q) \vol[n-1]{\partial Q \cap \partial S}\nonumber\\
    &
      = C n^{1/p}\sidelen{Q} \vol[n-1]{\partial Q \cap \partial S}.\label{eq:11}
  \end{align}
  Define $S_\epsilon$ similarly by 
   \begin{equation}
    S_\epsilon = \bigcup_{x \in S} \left( x + {2\epsilon\dist_{\ell_p}(x,\partial
        K)}B_\infty \right).\label{eq:9.5}
  \end{equation}
  Clearly, $S_\epsilon = \bigcup_{Q \in S} Q_\epsilon$.  Before proceeding, we
  also note that it follows from \cref{lem-small-metric-distances} that when
  $\epsilon \in (0, 1)$ is sufficiently small (as a function of $n$ and $p$)
  then
  \begin{equation}
    \dist_{g_p}(S,(K\setminus S_\epsilon)) \ge \epsilon.\label{eq:105}
  \end{equation}
  We now compute the ergodic flow of $\MM_p$ out of $S$.
  \begin{align}
    \Psi(S) &= \sum_{Q \in S} \sum_{Q' \not\in S}  \pi(Q) P(Q,Q') \\
            &= \sum_{Q \in S} \sum_{Q' \not\in S} \frac{1}{2} \cdot
              \frac{\vol{Q}}{\vol{K}} \cdot \frac{\vol[n-1]{\partial Q \cap
              \partial Q'}}{\vol[n-1]{\partial Q}} \cdot \min\left\{ 1 ,
              \frac{\sidelen{Q'}}{\sidelen{Q}} \right\}
              \quad \text{(from \cref{eq:2})} \\
            &=\frac{1}{4n\,\vol{K}}
              \sum_{Q \in S}  \sidelen{Q}
              \sum_{Q' \not\in S}\vol[n-1]{\partial Q \cap \partial Q'}
              \cdot \min\left\{ 1 , \frac{\sidelen{Q'}}{\sidelen{Q}} \right\}\\
            &\ge \frac{1}{8n\vol{K}} \sum_{Q \in S} \sidelen{Q} \cdot
              \vol[n-1]{\partial Q \cap \partial S}
              \quad \text{(from \cref{item:side-length-ratio} of \Cref{theo 1.1})} \\
            &\ge \frac{1}{8C n^{1 + 1/p}\vol{K}} \sum_{Q \in S} \limsup_{\epsilon \downarrow
              0} \frac{\vol{Q_\epsilon \setminus S}}{\epsilon}
              \quad \text{(from \cref{eq:11})}\label{eq:120} \\
            &\geq \frac{1}{8C n^{1 + 1/p}\vol{K}}
              \limsup_{\epsilon \downarrow 0}
              \sum_{Q \in S}  \frac{\vol{Q_\epsilon \setminus S}}{\epsilon}
              \quad \text{(from the reverse Fatou's lemma, see below)}\nonumber \\
            &\geq \frac{1}{8C n^{1 + 1/p}\vol{K}} \limsup_{\epsilon \downarrow 0}
              \frac{\vol{S_\epsilon \setminus S}}{\epsilon}
              \quad \text{(since $S_\epsilon = \bigcup_{Q \in S} Q_\epsilon$ )}.\label{eq:106}
  \end{align}
  Here, to interchange the limit and the sum in \cref{eq:120}, one can use,
  e.g., the reverse Fatou's lemma, after noting that
  $\sum_{Q \in S} \vol{Q} = \vol{S} < \infty$, and that for all
  $0 < \epsilon \leq 1$, and all $Q \in S$, the discussion following the
  definition of $Q_\epsilon$ in \cref{eq:9} implies that
  \begin{equation*}
    \frac{\vol{Q_\epsilon \setminus S}}{\epsilon} \leq \frac{\vol{Q_\epsilon
        \setminus Q}}{\epsilon} \leq \vol{Q}\cdot\frac{\inp{1 + 2\epsilon\cdot Cn^{1/p}}^n - 1}{\epsilon}
    \leq \vol{Q}\cdot \xi_{n,p},
  \end{equation*}
  where $\xi_{n, p} \defeq \inp{1 + 2Cn^{1/p}}^n$ is a finite positive
  number that depends only on $n$ and $p$, and not on $Q$ or $\epsilon$.  Now,
  as noted above, \cref{lem-small-metric-distances} implies that when
  $\epsilon \in (0, 1)$ is sufficiently small (as a function of $n$ and $p$)
  then $\dist_{g_p}(S,(K\setminus S_\epsilon)) \ge \epsilon$.  We can therefore
  apply \Cref{theo isoperimetry} after setting $\delta = \epsilon$, $S_1 = S$,
  $S_3 = (S_\epsilon \setminus S)$, and $S_2 = (K \setminus S_\epsilon)$.  Let
  $C_0$ be as in the statement of \cref{theo isoperimetry}. Applying \cref{theo
    isoperimetry}, we then get that if $S$ is such that
  $\pi(S) \leq \exp(-C_{0}n)$, then
  \begin{equation}
    \label{eq:107}
    \frac{\vol{S_\epsilon \setminus S}}{\epsilon} \ge \Theta(1)\cdot\rho_p\cdot\vol{S_1}.
  \end{equation}
  Substituting this in \cref{eq:106}, we thus get that for such $S$,
  \begin{equation}
    \label{eq:108}
    \Psi(S) \ge  \frac{\rho_p}{O(n^{1 + \frac{1}{p}})}
    \cdot \pi(S).
  \end{equation}
  We thus get that the value of the conductance profile $\Phi_\a$ at
  $\a \leq \exp(-C_0 n)$ is
  \begin{equation}
\Phi_\a \geq  \frac{\rho_p}{O(n^{1 + \frac{1}{p}})}.\label{eq:110}
  \end{equation}
  Similarly, when $S$ is such that $\frac{1}{2} \geq \pi(S) > \exp(-C_{0}n)$,
  \cref{theo isoperimetry} gives
  \begin{equation}
    \label{eq:80}
    \frac{\vol{S_\epsilon \setminus S}}{\epsilon} \ge \frac{\rho_p}{O(n)}
    \cdot \vol{S_1}
    \cdot \log\inp{1 + 0.9\frac{\vol{K}}{\vol{S_1}}}.
  \end{equation}
  Substituting this in \cref{eq:106}, we thus get that for such $S$,
  \begin{equation}
    \label{eq:81}
    \Psi(S) \ge  \frac{\rho_p}{O(n^{2 + \frac{1}{p}})}
    \cdot \pi(S)
    \cdot \log\inp{1 + \frac{0.9}{\pi(S)}}.
  \end{equation}
  Combining this with \cref{eq:110}, we get that the value of the conductance
  profile $\Phi_\a$ at $\a > \exp(-C_{0}n)$ therefore satisfies
\begin{equation}
  \Phi_\a \geq  \frac{\rho_p}{O(n^{2 + \frac{1}{p}})}
    \cdot \log\inp{1 + \frac{0.9}{\a}}.\label{eq:111}
  \end{equation}
  Combining \cref{eq:110,eq:111}, we also get that the conductance
  $\Phi = \Phi_{\MM_p}$ of the chain $\MM_p$ satisfies
  \[
\Phi \geq \frac{\rho_p}{O(n^{2 + \frac{1}{p}})}. \qedhere
  \]
\end{proof}

\subsection{Tightness of the conductance bound}
\label{sec:tightn-cond-bound}
We now show that at least in the worst case, the conductance lower bound proved
above for the $\MM_{p}$ chains is tight up to a logarithmic factor in the
dimension.
\begin{proposition}\label{prop:conductacne-lower-bound}
  Fix $1 \leq p \leq \infty$, and the convex body
  $K = \left[-\frac{1}{2}, \frac{1}{2}\right]^n$, and consider the Markov chain $\MM_p$ on
  the Whitney decomposition $\FF = \FF^{(p)}$ of $K$.  We then have
  $$\Phi_{\MM_p} \leq O\left(\frac{\log n}{n^{2 + \frac{1}{p}}}\right).$$
\end{proposition}
\begin{proof}
  We consider the half-cube $S_1 = K \cap H$, where $H$ is the half-space
  $\inb{x \st x_{1} \leq 0}$.  Note that the construction of $\FF$ implies that
  the boundary of $H$ does not intersect the interior of any Whitney cube in
  $\FF$. Thus, the set $S \subseteq \FF$ of Whitney cubes lying inside $S_1$ in
  fact covers $S_1$ fully.  In particular, this implies that
  $\pi(S) = \pi_{K}(S_{1}) = 1/2$.

  We now proceed to bound the ergodic flow $\Psi_{\MM_{p}}(S)$ from above.
  Let $A \subseteq S$ be the set of Whitney cubes whose boundary has a non-zero
  intersection with the boundary of $H$.  We then note that
  $P_{{\MM_{p}}}(q, \FF \setminus S) = 0$ when $q \in S \setminus A$
  (because $\MM_{p}$ only moves to abutting cubes) and also that
  \begin{equation}
    \label{eq:112}
    P_{{\MM_{p}}}(q, \FF \setminus S) \leq \frac{1}{2n} \; \text{ when } q \in A,
  \end{equation}
  since $\MM_{p}$ chooses a point uniformly at random from the boundary
  $\bdry{q}$ of $q$ to propose the next step, and only one of the $2n$ faces of
  $q$ can have a non-trivial intersection with the boundary of $H$.  It follows
  that
  \begin{equation}
    \label{eq:113}
    \Psi_{\MM_{p}}(S) \leq \frac{1}{2n} \pi(A).
  \end{equation}
  We will show now that $\pi(A) \leq O\inp{\frac{\log n}{n^{1 + \frac{1}{p}}}}$,
  which will complete the proof.
  
  To estimate $\pi(A)$, we recall that sampling a cube from the probability
  distribution $\pi$ on $\FF$ can also be described as first sampling a point
  $x$ according to the uniform distribution $\pi_K$ on $K$, and then choosing
  the cube $q \in \FF$ containing $x$ (as discussed earlier, with probability
  $1$ over the choice of $x$, there is a unique $q$ containing $x$).  Now, let
  $\tilde{A}$ denote the set of those $x \in S_1$ which are at
  $\ell_{p}$-distance at most $\frac{10\log n}{n}$ from the boundary of $K$, and
  also at $\ell_{\infty}$-distance at least $\frac{100 \log n}{n^{1 + \frac{1}{p}}}$ from the
  boundary $\bdry{H}$ of $H$.  Formally,
  \begin{equation}
    \label{eq:115}
    \tilde{A}
    \defeq
    \inb{x \st \dist_{\ell_p}(x, \bdry{K}) \leq \frac{10\log n}{n}}
    \cap \inb{x \st \dist_{\ell_\infty}(x, \bdry{H}) \geq \frac{100\log n}{n^{1+\frac{1}{p}}}}
    \cap S_1
  \end{equation}

  The first condition for $x$ being in $\tilde{A}$ implies, due to
  \cref{item:diamater-of-cube} of \cref{theo 1.1}, that the Whitney cube
  $q \in \FF$ containing $x$ has $\ell_{p}$-diameter at most
  $\frac{10\log n}{n}$. The sidelength of this $q$ is therefore at most
  $\frac{10\log n}{n^{1 + \frac{1}{p}}}$.  Combined with the second condition in
  the definition of $\tilde{A}$, this implies that if $x \in \tilde{A}$, then
  $q \in S$ and $q \not\in A$. Since $A\subseteq S$, we can now use the
  translation described above between the probability distributions $\pi$ on
  $\FF$ and $\pi_K$ on $K$ to get
  \begin{equation}
    \label{eq:114}
    \pi_{K}(\tilde{A})
    = \Pr[x \sim \pi_K]{x \in \tilde{A}}
    \leq \Pr[q \sim \pi]{q \in S \text{ and } q \not\in A}
    = \pi(S \setminus A) = \pi(S) - \pi(A).
  \end{equation}
  Now, note that since $K$ is an axis aligned cube, we have
  $\dist_{\ell_{p}}(x, \bdry{K}) = \dist_{\ell_{\infty}}(x, \bdry{K})$ for any
  $x \in K$ and any $\ell_{p}$-norm, where $1 \leq p \leq \infty$.  Using this,
  a direct calculation gives
  \begin{equation}
    \pi_{K}(\tilde{A})
    \geq \frac{1}{2} \cdot\inp{1 - \inp{1 - \frac{20\log n}{n}}^n} -
    \frac{100\log n}{n^{1 + \frac{1}{p}}} \geq \frac{1}{2}\inp{1 - n^{-20}} - \frac{100 \log
      n}{n^{1 + \frac{1}{p}}}.
  \end{equation}
  Plugging this into \cref{eq:114} and using $\pi(S) = 1/2$ gives
  $\pi(A) \leq n^{-20}/2 + \frac{100 \log n}{n^{1 + \frac{1}{p}}}$.  Using
  \cref{eq:113}, this gives
  $\Psi_{\MM_{p}}(S) \leq O\inp{\frac{\log n}{n^{2 + \frac{1}{p}}}}$.  Since
  $\pi(S) = 1/2$, this yields the claimed upper bound on the conductance
  $\Phi_{\MM_{p}}(S) = \Psi_{\MM_{{p}}}(S)/\pi(S)$ of $S$.
\end{proof}

\subsection{Rapid mixing from a cold start: Proof of \texorpdfstring{\cref{thm:mp-mixing-intro}}{Theorem 1.2}}
\label{sec:mixing-from-cold}
Given the conductance bound, the proof for rapid mixing from a cold start follows
from a result of Lovász and Simonovits~\cite{LS93} (\cite[Corollary 1.8]{LS93},
as stated in \cref{lem:ls-l2-mixing}).
Recall that we denote the multiscale chain corresponding to the $\ell_p$-norm by
$\MM_p$.  We say that a starting density $\eta_0$ is \emph{$M$-warm in the
  $L^2(\pi)$ sense }if $\|\eta_0 - \one\|_{L^2(\pi)} \leq M$. Note that if
$\|\eta_0\|_\infty \leq M - 1$, then this criterion is satisfied.
\begin{corollary} Let $0 < \eps < 1/2$.  The mixing time $T$ of $\MM_p$ to
  achieve a total variation distance of $\eps$, from any $M$-warm start (in the
  $L^2(\pi)$ sense), obeys
$$T \leq O\left(\frac{n^{4 + \frac{2}{p}}}{\rho_p^2}\log \frac{M}{\eps}\right).$$ \label{cor:mp-mixing}
\end{corollary}
\begin{proof} If $\frac{R}{r}$ is bounded above by a polynomial in $n$, so is
  $\rho_p^{-1}$. Recall that $f\pi$ denotes the probability distribution with
  density $f$ with respect to $\pi$. As in \cref{lem:ls-l2-mixing}, let
  $\eta_{T}$ denote the density (with respect to $\pi$) of the distribution
  obtained after $T$ steps of chain.  Since we have
  $$2 \, d_{TV}(\eta_T\pi, \pi) = \|\eta_T -1 \|_{L^{1}(\pi)}\leq \sqrt{\langle \eta_T - \one, \eta_T -
    \one\rangle_{L^2(\pi)}},$$ this corollary follows from
  \cref{lem:ls-l2-mixing} and the fact that
  $$\Phi \geq \frac{\rho_p}{O(n^{2 + \frac{1}{p}})},$$ as shown in
  \cref{eq:86} in \cref{thm:conductance}.
\end{proof}

\begin{proof}[Proof of \cref{thm:mp-mixing-intro}] The theorem follows
  immediately from \cref{cor:mp-mixing} after a few substitutions.  Recall that
  the notation $f\nu$ denotes the probability distribution that has density $f$
  with respect to $\nu$.  Let $\eta_0'$ be the density of the initial $M$-warm
  start (with respect to $\pi_K)$ in the statement of
  \cref{thm:mp-mixing-intro}, and let $\nu_0' = \eta_{0}'\pi_{K}$ be the
  corresponding probability distribution on $K$.  Let $\eta_0$ be the density
  (with respect to the distribution $\pi$ on $\FF^{{(p)}}$) of the probability
  distribution $\nu_0$ on $\FF^{(p)}$ obtained by first sampling a point $x$
  according to $\nu_{0}'$ and then choosing the cube in $\FF_q$ in which $x$
  lies.  As argued in the paragraph preceding \cref{thm:mp-mixing-intro}, the
  probability distribution $\nu_{0}$ is also $M$-warm with respect to $\pi$ in
  the sense that $\norm[\infty]{\eta_0} \leq M$.  This implies that
  $\norm[L^2(\pi)]{\eta_0 - 1} \leq M + 1$.  Next, note that $\rho_p = R/r$ by
  definition, where $R$ and $r$ are as in the statement of
  \cref{thm:mp-mixing-intro}. \cref{cor:mp-mixing} then implies that the total
  variation distance between the probability distributions
  $\nu_{T} \defeq \eta_{T}\pi$ and $\pi$ on $\FF^{(p)}$ is at most $\epsilon$.
  This implies that for any function $f: \FF^{(p)} \rightarrow [0, 1]$,
  \begin{equation}
    \label{eq:89}
    \E_{Q \sim \pi}[f(X)] - \E_{Q \sim \nu_T}[f(X)] \leq \epsilon.
  \end{equation}

  To prove \cref{thm:mp-mixing-intro}, we now need to show that the probability
  distribution $\nu_T' \defeq \eta_T'\pi_K$ on $K$ obtained by first sampling a
  cube $Q$ from $\nu_T = \eta_T\pi$, and then sampling a point uniformly at
  random from $Q$ (as discussed in the paragraph preceding
  \cref{thm:mp-mixing-intro}) is within total variation distance at most
  $\epsilon$ from the uniform distance $\pi_K$ on $K$.  To do this, we only need
  to show that for any measurable subset $S$ of $K$, we have
  \begin{equation}
    \label{eq:88}
    \pi_K(S) - \nu_T'(S) \leq \epsilon.
  \end{equation}
  Recall that $\pi_W$ denotes the uniform probability distribution on $W$, where
  $W$ is any measurable set. For any measurable subset $S$ of $K$ and a cube
  $Q \in \FF^{(p)}$, we then have
  $\pi_K(S \cap Q) = \pi_Q(S \cap Q)\pi_K(Q) = \pi_Q(S \cap Q)\pi(Q)$ and
  $\nu_T'(S \cap Q) = \pi_Q(S \cap Q)\nu_T(Q)$.  Now, since the cubes in
  $\FF = \FF^{(p)}$ form a countable partition of $K^\circ$
  (\cref{item:Whitney-cubes,item:disjoint-Whitney} of \cref{theo 1.1}), we get
  that
  \begin{align}
    \pi_K(S) - \nu_T'(S)
    &= \sum_{Q \in \FF}\pi_Q(S \cap Q)\pi(Q) -  \sum_{Q \in \FF}\pi_{Q}(S \cap Q)\nu_T(Q)\\
    &= \E_{Q \sim \pi}\insq{\pi_Q(S \cap Q)} - \E_{Q \sim \nu_T}\insq{\pi_{Q}(S \cap Q)}
      \stackrel{\text{\cref{eq:89}}}{\leq}
      \epsilon.\label{eq:91}
  \end{align}
  Here, the inequality in \cref{eq:89} is applied with
  $f(Q) \defeq \pi_Q(S \cap Q) \in [0, 1]$.  \Cref{eq:91} thus proves
  \cref{eq:88} and hence completes the proof of the theorem.
\end{proof}

\subsection{An extension: rapid mixing from a given state}
\label{sec:mixing-time-from}

In the following theorem, we state an upper bound on the time taken by $\MM_p$
to achieve a total variation distance of $\eps$ from the stationary distribution
$\pi$ on the set of cubes starting from a given state. By using the notion of
``average conductance'' introduced by Lov\'{a}sz and Kannan in \cite{LK}, we
save a multiplicative factor of $\tilde{O}(n)$ (assuming that the starting state
is at least $1/\poly{n}$ away from the boundary of the body $K$) from what would
be obtained from a direct application of the conductance bound above. This is
possible because our lower bound on the value of the conductance profile for
small sets is significantly larger than our lower bound on the worst case value
of the conductance.
\begin{theorem}[Mixing time from a given state] Fix $p$ such that
  $1 \leq p \leq \infty$. Let $K$ be a convex body such that
  $r_p\cdot B_{p} \subseteq K \subseteq R_\infty\cdot B_{\infty}$.  Define
  $\rho_p \defeq r_p/R_\infty \leq 1$ as in the statement of \cref{theo
    isoperimetry}. Consider the Markov chain $\MM_p$ defined on the Whitney
  decomposition $\FF^{(p)}$ of $K$.  Let $X_t$ be a Markov chain evolving
  according to $\MM_p$, where $X_0 = Q \in \FF$. Suppose that
  $\dist_{\ell_p}(\cntr(Q), \partial K) = d$. Given $\eps \in (0, 1/2)$, after
  $$T = C\log \eps^{-1}\left(\frac{n^{4 +
        \frac{2}{p}}}{\rho_p^2}\right)\left(\log \frac{nR_\infty}{d} +
    n^{-1}\log\left(\frac{n}{\rho_p\eps}\right)\right)$$ steps, the total
  variation distance $d_{TV}(X_T, \pi)$ is less than $\eps$, for a universal
  constant $C$.\label{thm:mp-mixing-refined}
\end{theorem}

In the proof of this theorem, we will need an auxiliary finite version of the
$\MM_p$ chain, that we now proceed to define.  For notational convenience, we
will also assume in the proof that the body $K$ is scaled so that
$R_\infty \leq 1$ (to translate the calculations below back to the general
setting, we will need to replace $d$ by $d/R_{\infty}$).

\subsubsection{A family of  auxiliary chains \texorpdfstring{$\MM_{p, a}$}{Mp,a}}\label{ssec:aux}
We consider a finite variant of the multiscale chain $\MM_p$, which we term $\MM_{p, a}$ where $p$ corresponds to the $\ell_p$-norm used and $a \geq 1$ is a natural number.
In this chain, all the states $Q$ in $\MM_p$ that correspond to cubes of side
length less or equal to $2^{-a}$ are fused into a single state, which we call $Q_\infty$, which we also identify with the union of these cubes. The transition probabilities $P_{p, a}$ of $\MM_{p, a}$ to and from $Q_\infty$ are defined as follows, in terms of the transition probabilities $P_p$ of $\MM_p$. Suppose (the interior of) $Q'$ is disjoint from $Q_\infty$.
$$P_{p, a}(Q', Q_\infty) \defeq \sum\limits_{\FF\ni Q \subseteq Q_\infty}  P_p(Q', Q).$$

$$P_{p, a}(Q_\infty, Q') \defeq \frac{\sum\limits_{\FF\ni Q \subseteq Q_\infty} \pi(Q) P_p(Q, Q')}{\sum\limits_{\FF\ni Q \subseteq Q_\infty} \pi(Q)}.$$
For two cubes $Q, Q'$ that are disjoint from $Q_\infty$, 
$$P_{p, a}(Q, Q') \defeq P_p(Q, Q').$$
Finally
$$P_{p, a}(Q_\infty, Q_\infty) \defeq  \frac{ \sum\limits_{\FF\ni Q'' \subseteq Q_\infty}\sum\limits_{\FF\ni Q \subseteq Q_\infty} \pi(Q) P_p(Q, Q'')}{\sum\limits_{\FF\ni Q \subseteq Q_\infty} \pi(Q)}.$$

\noindent We now proceed with the proof of \cref{thm:mp-mixing-refined}.
\begin{proof}[Proof of \cref{thm:mp-mixing-refined}]
  To avoid cluttering of notation in this proof, we will adopt the standard
  convention that different occurrences of the letters $C$ and $c$ can refer to
  different absolute constants.

  By \cref{item: diameter of cube center} of \cref{theo 1.1},  for any Whitney cube $Q \in \FF$,
\[  \dist_{\ell_p}(\cntr(Q),\Kb) \le \frac{9}{2} \diam_{\ell_p}(Q).\]
Let $2^{-b}$ denote the side length of $Q$, where $b$ is a positive natural number.
  Thus, $$2^{-b} = n^{-\frac{1}{p}}\diam_{\ell_p}(Q) \geq \frac{2dn^{-\frac{1}{p}}}{9}.$$ With the notation of \cref{lem:ls-l2-mixing}, we have 
 $$  \langle \eta_0 - \one, \eta_0 - \one\rangle_{L^2(\pi)} = 2^{-nb}(2^{2nb})(\vol{K}) - 2 (2^{-nb})(2^{nb}) + 1 <  2^{nb}(\vol{K}).$$
Thus, using the inequality $ 1 - x \leq \exp(-x)$ for positive $x$ together with \cref{lem:ls-l2-mixing}, we have 
$$\langle \eta_T - \one, \eta_T - \one\rangle_{L^2(\pi)} \leq \exp\left( - \Phi^2T\right) 2^{nb}(\vol{K}).$$
Since
$$\|\eta_T -1 \|_{L^{1}(\pi)} \leq \|\eta_T -1 \|_{L^2(\pi)} = \sqrt{\langle
  \eta_T - \one, \eta_T - \one\rangle_{L^2(\pi)}},$$ and $R_\infty \leq 1$,
we have
$$\|\eta_T -1 \|_{L^{1}(\pi)}\leq \exp\left( - \frac{\Phi^2T}{2}\right)
2^{\frac{n(b+1)}{2}}.$$ Substituting
$\Phi \geq \frac{r_p}{O(n^{2 + \frac{1}{p}}) \cdot R_\infty}$ from \cref{eq:86}
we have
    $$\|\eta_T -1 \|_{L^{1}(\pi)}\leq \exp\left( -\left( \frac{ r_p}{O(n^{2 + \frac{1}{p}}) \cdot R_\infty}\right)^2T\right) 2^{\frac{n(b+1)}{2}}.$$
    In order to have  $\|\eta_T -1 \|_{L^{1}(\pi)}\leq 2\eps$, it suffices to have $$T \geq \tilde{T}(\eps) \defeq \left(\log\left( (2\eps)^{-1}\right) + \frac{n(b+1) \ln 2}{2} \right)\left( \frac{O(n^{2 + \frac{1}{p}}) \cdot R_\infty}{ r_p}\right)^2.$$
Let $f_{p, a}$ be the function from $\MM_p$ to $\MM_{p, a}$ that maps a cube $Q \in \MM_p$ to $Q$ if $Q$ is not contained in $Q_\infty$, and otherwise maps $Q$ to $Q_\infty$.
Recall that  $X_t$ is a Markov chain evolving according to $\MM_p$, where $X_0 = Q \not\subseteq Q_\infty$. So $f_{p, a}(X_t)$ evolves according to $\MM_{p, a}$ until the random time $\tau(Q, Q_\infty)$  that it hits $Q_\infty$. Let $\tilde{a}(\eps)$ be the minimum $a$ such that $Q_\infty$ satisfies the following property:
\begin{equation}
 \p\left[\tau(Q, Q_\infty) \geq \tilde{T}\left(\frac{\eps}{2}\right) \right]
 \geq 1 - \frac{\eps}{2}.\label{eq:118}
\end{equation}
Let $\pi^a$ denote the stationary measure of the chain $\MM_{p,a}$ for any natural number $a$.
Note that $\pi^a(Q') = \pi(Q')$ for all $Q'$ that are not contained in $Q_\infty$.
Therefore, for any $T \leq \tilde{T}\left(\frac{\eps}{2}\right)$, we have the following upper bound:
\begin{equation}
  d_{TV}(X_T, \pi) \leq d_{TV}(f_{p, \tilde{a}(\eps)}(X_T), \pi^{\tilde{a}(\eps)}) + \frac{\eps}{2}. \label{eq:atilde}
\end{equation}
We will next obtain an upper bound on the right hand side by finding an upper
bound on the mixing time of $\MM_{p, a}$. The conductance profile of
$\MM_{p, a}$ dominates that of $\MM_p$ because for any subset $S$ of the states
of $\MM_{p, a}$, $f_{p, a}^{-1}(S)$ is a subset of the states of $\MM_p$ of the same
measure, and the transitions of $\MM_{p, a}$ correspond to that of a chain
obtained from fusing the states in $Q_\infty$ as stated in \cref{ssec:aux}.

In preparation for applying the average conductance result of Lovász and
Kannan~\cite{LK} to the chain $\MM_{p, \tilde{a}(\epsilon)}$, we denote \beq\HH
\defeq \frac{1}{\Phi_{1/2}} +
\int\limits_{(2^{\tilde{a}(\eps)n}\vol{K})^{-1}}^{\frac{1}{2}} \frac{d\a}{\a
  \Phi_\a^2}.\label{eq:H0}\eeq Here, $1/(2^{\tilde{a}(\eps)n}\vol{K})$ is a
lower bound on the stationary probability $\pi^{\tilde{a}(\epsilon)}(q)$ of any
state $q$ of the finite state Markov chain $\MM_{p, \tilde{a}(\epsilon)}$, and
$\Phi_\alpha = \Phi_{\MM_{p}, \alpha}$ is the conductance profile of $\MM_{p}$.
As argued above, this conductance profile is dominated by that of
$\MM_{p, \tilde{a}(\epsilon)}$, i.e.,
$\Phi_{\MM_{p, \tilde{a}(\epsilon)}, \alpha} \geq \Phi_{\MM_{p}, \alpha}$ for
every $\pi_0 \leq \alpha \leq 1/2$, where $\pi_0$ is the minimum stationary
probability of any state of $\MM_{p, \tilde{a}(\epsilon)}$.  Using the
discussion in \cite[p.~283, especially Theorem 2.2]{LK} (see also the corrected
version \cite[Theorem 2.2]{MS01}), we thus get that the mixing time to reach a
total variation distance of $\eps/2$ from any state of the finite Markov chain
$\MM_{p, \tilde{a}(e)}$ is at most (for some absolute constant $C$)
\begin{equation}
C\HH\log\left(\frac{2}{\eps}\right). \label{eq:H}
\end{equation}
Therefore, by
\cref{eq:atilde}, $C\HH\log\left(\frac{2}{\eps}\right)$ is an upper bound on the
the time needed by $\MM_p$ starting at the state $Q$ to achieve a total
variation distance of $\eps$ to stationarity. It remains to estimate
(\ref{eq:H}) from above.  As a consequence of \cref{eqn: multiscale conductance
  profile - 1,eqn: multiscale conductance profile - 2} in
\cref{thm:conductance}, \beqs \Phi_{\exp(-x)} \geq \begin{cases}
  \frac{\rho_p}{O(n^{2 + \frac{1}{p}})}
                                                     \cdot (x + \log (0.9)), & \text{if } x  < C_{0}n\\
                                                     \frac{\rho_p}{O(n^{1 +
                                                     \frac{1}{p}})}, & \text{if
                                                                       } x \geq
                                                                       C_0 n.
   \end{cases} \eeqs
   Therefore,
\beq\HH &=& \frac{1}{\Phi_{1/2}} + \int\limits_{(2^{\tilde{a}(\eps)n}\vol{K})^{-1}}^{\frac{1}{2}} \frac{d\a}{\a \Phi_\a^2}\\
& = & \frac{1}{\Phi_{1/2}} +\int\limits^{\ln\left(2^{\tilde{a}(\eps)n}\vol{K}\right)}_{\ln 2}  \Phi_{\exp(-x)}^{-2}dx\\
& \leq & \left(\frac{\rho_p}{O(n^{2 + \frac{1}{p}})}\right)^{-2} +
\tilde{a}(\eps)n\left( \frac{\rho_p}{O(n^{1 +
      \frac{1}{p}})}\right)^{-2},\label{eq:92-hari-Nov6-2022}\eeq
where in the last line we use the assumption $R_\infty \leq 1$ to upper bound
$\vol{K}$.

Lastly, we will find an upper bound for $\tilde{a}(\eps)$.
Observe that, by reversibility (see the discussion following \cref{eq:102}), we
have that for all $t \geq 0, $ $\|\eta_{t+1}\|_\infty =
\|\MM_{p, \tilde{a}(\epsilon)}\eta_{t}\|_\infty  \leq \|\eta_{t}\|_\infty$, and consequently, for all positive $t$,
\begin{equation}
\|\eta_t\|_\infty \leq \|\eta_0\|_\infty \leq 2^{nb}\vol{K}.\label{eq:116}
\end{equation}
We now proceed to upper bound the total volume of all cubes contained in
$Q_\infty$, for any value of the parameter $a$ used in its definition.  Since
any cube in $Q_\infty$ has side length at most $2^{-a}$, it follows from
\cref{item:diamater-of-cube} of \cref{theo 1.1} that any point in such a cube
must be at an Euclidean distance of at most
$5n^2 2^{-a} < C n^{C} 2^{-a}$ from $\bdry{K}$.  Define the inner
parallel body $K_r$ to be the set of all points in $K$ at a Euclidean distance
at least $r$ from $\partial K.$ This set is convex as can be seen from
\cref{lem: l1 concave}.  From the above discussion, we also
$\vol{Q_\infty} \leq \vol{K^\circ \setminus K_{C n^C 2^{-a}}}$.

Now, by the coarea formula for the closed set $\R^n\setminus K^\circ$ (see Lemma
3.2.34 on p.~271 of \cite{federer_geometric_1996}), letting $u(x)$ be the
$\ell_2$ distance of $x$ to $\R^n\setminus K^\circ$ (and therefore also to
$\bdry{K}$), we have
\begin{equation}
  \vol{K^\circ \setminus K_{r}}
  =\int\limits_{0}^{r}
  H_{n-1}(u^{-1}(t))\;dt
  =\int\limits_{0}^{r}\vol[n-1]{\bdry{K_t}}
  \,dt
  \end{equation}
  where $H_{n-1}$ is the $(n -1)$-dimensional Hausdorff measure. In particular,
  this, together with $$\vol[n-1]{\partial K_t} \leq \vol[n-1]{\partial K},$$
  (which follows from \cref{prop:cauchy-surface}) implies
  \begin{equation}
    \vol{K^\circ\setminus K_{r} }
    = \int\limits_0^{r}
    \vol[n-1]{\partial K_t} dt \leq r\cdot\vol[n-1]{\partial K}.
  \end{equation}
  Thus (again using \cref{prop:cauchy-surface}, along with the assumption
  $R_\infty \leq 1$, which implies $K \subseteq B_\infty$),
$$\vol{Q_\infty} < C2^{-a}n^C\vol[n-1]{\partial K}
< C2^{-a}n^C\vol[n-1]{\partial B_\infty} < C2^{-a + n}n^{C+1}.$$
Thus, for all $a$
such that
$$C2^{-a + n}n^{C+1} <
\frac{\eps}{2^{nb+2}\left(\tilde{T}\left(\frac{\eps}{2}\right)\right)},$$ we
have (using \cref{eq:116}, and recalling that $\eta_t\pi^{a}$
is the probability distribution with density $\eta_t$ with respect to
$\pi^{a}$)
  $$
  \p\left[\tau(Q, Q_\infty) < \tilde{T}\left(\frac{\eps}{2}\right) \right]
  \leq \sum\limits_{t =
    0}^{\tilde{T}\left(\frac{\eps}{2}\right)}(\eta_t\pi^{a})(Q_\infty)
  \leq 2^{nb}\cdot\inp{1 + \tilde{T}\left(\frac{\eps}{2}\right)}\cdot\vol{Q_\infty} < \frac{\eps}{2}.$$
Thus, from \cref{eq:118}, we see that $\tilde{a}(\eps)$ only needs to satisfy
$$2^{-\tilde{a}(\eps)} <
\frac{c\eps}{2^{n(b+1)}n^{C+1}\left(\tilde{T}\left(\frac{\eps}{2}\right)\right)}.$$
Since
$$\tilde{T}(\eps/2) \leq \left(\log \eps^{-1} + nb\right)\left(O(n^C
  \rho_p^{-2})\right),$$ we see that $\tilde{a}(\eps)$ can be chosen so that
$$ \tilde{a}(\eps) \leq \log \left(\frac{2^{n(b+1)}n^{C+1} \left(\log \eps^{-1}
      +nb\right)\left(O(n^C \rho_p^{-2})\right)}{c\eps}\right),$$ which simplifies
to \beq \tilde{a}(\eps) < C\left(nb +
  \log\left(\frac{n}{\rho_p\eps}\right)\right).\label{eq:93}\eeq Finally,
putting together (\ref{eq:92-hari-Nov6-2022}) and (\ref{eq:93}), we have for all
$\eps < \frac{1}{2}$,
\[ C\HH\log \frac{2}{\eps} \leq C\log \eps^{-1}\left(\frac{n^{4 +
        \frac{2}{p}}}{\rho_p^2}\right)\left( \log \frac{n}{d} +
    n^{-1}\log\left(\frac{n}{\rho_p\eps}\right)\right). \] In light of the
discussion following \cref{eq:H} and the choice of $\tilde{a}(\epsilon)$, this
completes the proof of the claimed mixing time for $\MM_p$ as well (recall that
we assumed by scaling the body that $R_\infty \leq 1$).
 \end{proof}

\section{Coordinate hit-and-run}
\label{sec:coordinate-hit-run}
Given a convex body $K$ in $\R^n$ and $x_1 \in K$, the steps $x_1,x_2,\ldots,$
of the Coordinate Hit-and-Run (CHR) random walk are generated as follows. Given
$x_i$, with probability $1/2$, we stay at $x_{i}$. Otherwise, we uniformly
randomly draw $j$ from $[n]$ and let $\ell$ be the chord $(x_i + e_j\R) \cap K$,
and then set $x_{i+1}$ to be a uniformly random point from this segment $\ell$.
In this section, we prove \cref{thm:intro-chr}, which shows that the CHR random
walk on convex bodies mixes rapidly even from a cold start. Our main technical
ingredient is an improvement (\cref{theo: axis-disjoint isoperimetry}) of an
isoperimetric inequality of Laddha and Vempala~\cite{laddha_convergence_2021}.

\subsection{Isoperimetric inequality for axis-disjoint sets}
We need the following definition, due to Laddha and
Vempala~\cite{laddha_convergence_2021}.
\begin{definition}[\textbf{Axis-disjoint sets~\cite{laddha_convergence_2021}}]
  Subsets $S_1,S_2$ of $\R^n$ are said to be \emph{axis-disjoint} if for all
  $i \in [n]$, $(S_1 + e_i\R) \cap S_2 = \emptyset$, where $e_i$ is the standard
  unit vector in the $i$th coordinate direction.  In other words, it is not
  possible to ``reach'' $S_2$ from $S_1$ by moving along a coordinate direction.
\end{definition}

The main technical result of this section is the following
isoperimetric inequality for axis-disjoint sets.
\begin{theorem}
  \label{theo: axis-disjoint isoperimetry}
  Let $K$ be a convex body in $\R^n$.  Denote by $\Phi_{\MM_\infty}$ the
  conductance of the Markov chain $\MM_{\infty}$ defined on the Whitney
  decomposition $\FF^{(\infty)}$ of $K$.  Suppose that
  $K = S_1 \cup S_2 \cup S_3$ is a partition of $K$ into measurable sets such
  that $S_1,S_2$ are axis-disjoint.  Then,
  \[ \vol{S_3} \ge \Omega\left(\frac{\Phi_{\MM_\infty}}{n^{3/2}}\right) \cdot \min\{\vol{S_1} , \vol{S_2}\}. \]
\end{theorem}

Combined with the results already proved for the multiscale chain $\MM_\infty$,
this implies a conductance bound (\cref{thm:CHR-cond}), followed by rapid mixing
from a cold start (\cref{thm:intro-chr}), for the CHR walk.  \cref{theo:
  axis-disjoint isoperimetry} should be compared against the main isoperimetric
inequality of Laddha and Vempala~\cite[Theorem 3; Theorem 2 in the arXiv
version]{laddha_convergence_2021}.  The result of \cite{laddha_convergence_2021}
essentially required the sets $S_1$ and $S_2$ to be not too small: they proved
that for any $\epsilon > 0$ and under the same notation as in \cref{theo:
  axis-disjoint isoperimetry},
\begin{equation}
  \vol{S_3} \ge \epsilon\cdot\Omega\left(\frac{r}{n^{3/2}\cdot R}
  \right) \cdot \inp{\min\{\vol{S_1} , \vol{S_2}\} - \epsilon\cdot\vol{K}},\label{eq:92}
\end{equation}
when the body $K$ satisfies $rB_2 \subseteq K \subseteq RB_2$.  Such an
inequality gives a non-trivial lower bound on the ratio of $\vol{S_3}$ and
$\min\{\vol{S_1} , \vol{S_2}\}$ only when the latter is at least
$\epsilon \cdot \vol{K}$.  Further, due to the $\epsilon$ pre-factor, the volume
guarantee that it gives for $\vol{S_3}$ as a multiple of
$\min\{\vol{S_1} , \vol{S_2}\}$ degrades as the lower bound imposed on the
volumes of the sets $S_1$ and $S_2$ is lowered.  Thus, it cannot lead to a
non-trivial lower bound on the conductance of arbitrarily small sets.  As
discussed in the technical overview, this was the main bottleneck leading to the
rapid mixing result of Laddha and Vempala~\cite{laddha_convergence_2021}
requiring a warm start.  The proof strategy employed by Narayanan and
Srivastava~\cite{narayanan_srivastava_2022} for the polynomial time mixing of
CHR from a warm start was different from that of \cite{laddha_convergence_2021},
but still faced a similar bottleneck: non-trivial conductance bounds could be
obtained only for sets with volume bounded below.  In contrast, \cref{theo:
  axis-disjoint isoperimetry} allows one to prove a non-trivial conductance
bound for sets of arbitrarily small size: see the proof of \cref{thm:CHR-cond}.

The proof of the inequality in \cref{eq:92} by Laddha and
Vempala~\cite{laddha_convergence_2021} built upon an isoperimetry result for
cubes (\cref{lemma: cube isoperimetry} below).  At a high level, they then
combined this with a tiling of the body with a lattice of \emph{fixed} width, to
reduce the problem to a classical isoperimetric inequality for the Euclidean
metric~\cite{LS93}. In part due the fact that they used a lattice of fixed
width, they had to ``throw away'' some of the mass of the $K$ lying close to the
boundary $\bdry{K}$, which led to the troublesome $-\eps\vol{K}$ term in
\cref{eq:92} above.  The inequality in \cref{theo: axis-disjoint isoperimetry}
is able to overcome this barrier and provide a non-trivial conductance bound
even for small sets based on the following two features of our argument. First, at a superficial
level, the multiscale decomposition ensures that we do not have to throw away
any mass (on the other hand, not having a tiling of $K$ by a fixed lattice makes
the argument in the proof of \cref{theo: axis-disjoint isoperimetry} more
complicated).  Second, and more fundamentally, the multiscale decomposition
allows us to indirectly use (through the connection to the conductance of the
$\MM_\infty$ chain) our isoperimetric inequality (\cref{theo isoperimetry}),
which is especially oriented for handling sets with a significant amount of mass
close to the boundary $\bdry{K}$.

We now proceed to the proofs of \cref{theo: axis-disjoint
  isoperimetry,thm:CHR-cond}.  We begin by listing some results and observations
of Laddha and Vempala~\cite{laddha_convergence_2021} about axis-disjoint sets.

\begin{lemma}[\textbf{Laddha and Vempala~\cite[Lemma
    2.4]{laddha_convergence_2021}}]
    \label{lemma: cube isoperimetry}
    Let $S_1,S_2$ be measurable axis-disjoint subsets of an axis-aligned cube $Q$.
    Set $S_3 \defeq Q \setminus (S_1 \cup S_2)$. Then,
    \[ \pi_Q(S_3) \ge \frac{1}{2n} \min\{\pi_Q(S_1), \pi_Q(S_2)\}
      \geq \frac{1}{2n} \cdot \pi_Q(S_1) \cdot \pi_Q(S_2). \] In
    particular, if $\pi_Q(S_1) \le (2/3)$, then
    \[ \pi_Q(S_3) \ge \frac{1}{8n} \pi_Q(S_1). \]
  \end{lemma}
  For completeness, we include the proof.\footnote{In a manuscript that appeared
    after this paper had been posted online, Fernandez V \cite{FV23} has
    reported an improvement to this result, leading to a consequent improvement
    by a factor of $n$ over the mixing time we claim in \cref{thm:intro-chr}.}
  \begin{proof}[Proof of Lemma 2.4 in \cite{laddha_convergence_2021}] Without loss
    of generality, we assume that $Q = [0,1]^n$ and $\vol{S_1} \le \vol{S_2}$,
    so that $\pi_Q(S) = \vol{S_1} \le (1/2)$. Denote by
    $\mathrm{proj}_j(S_1) \subseteq \R^{n-1}$ the $(n-1)$-dimensional projection
    \[ \mathrm{proj}_j(S_1) = \{ (x_1,\ldots,x_{j-1},x_{j+1},\ldots,x_n) \st x \in
      S_1 \} \] of $S_1$ onto the $j$th hyperplane. Since $S_{1}$ and $S_{2}$
    are axis-disjoint, we have, for every $1 \leq j \leq n$,
    \[ S_1 \subseteq \left\{ x \in Q \st (x_1,\ldots,x_{j-1},x_{j+1},\ldots,x_n)
        \in \mathrm{proj}_j(S_1) \right\} \subseteq (S_1 \cup S_3), \] so that
    $\vol{S_3} \ge \vol[n-1]{\mathrm{proj}_j(S_1)} - \vol{S_1}$. Averaging this
    over $j$,
    \begin{align*}
      \pi_Q(S_3) =
      \vol{S_3}
      &\ge \frac{1}{n} \sum_{j=1}^n \left(\vol[n-1]{\mathrm{proj}_j(S_1)} - \vol{S_1}\right) \\
      &\ge \left(\prod_{j=1}^n \vol[n-1]{\mathrm{proj}_j(S_1)}\right)^{1/n} - \vol{S_1} \qquad \text{(AM-GM inequality)} \\
      &\ge \vol{S_1}^{1 - (1/n)} - \vol{S_1} \qquad \text{(Loomis-Whitney inequality)} \\
      &\ge \vol{S_1} \left( 2^{1/n} - 1 \right) \qquad\text{(because $\vol{S_1} \leq 1/2$)} \\
      &\ge \frac{1}{2n} \vol{S_1} = \frac{1}{2n} \min\{\pi_Q(S_1), \pi_Q(S_2)\}.
    \end{align*}
    The final comment in the statement follows by considering separately the
    cases $\pi_Q(S_2) \leq 1/4$ and $\pi_Q(S_2) > 1/4$.
  \end{proof}

  We will need the following corollary.
  \begin{lemma}
    \label{cor: minilemma}
    Let $S$ be a countable union of axis-aligned cubes in $\R^n$ with disjoint
    interiors. Let $S_1, S_2 \subseteq S$ be axis-disjoint and define
    $S_3 \defeq S \setminus (S_1 \cup S_2)$. Suppose that
    $\pi_{Q}(S_1) \le (2/3)$ for each cube $Q \in S$ and
    $\pi_S(S_1 \cup S_3) \ge \eta$. Then,
    \[ \pi_S(S_3) \ge \frac{\eta}{16n}. \]
  \end{lemma}
  \begin{proof}
    If $\pi_S(S_1) \le (\eta/2)$, we are done. If $(\eta/2) \le \pi_S(S_1)$, the previous lemma gives that
    \[ \pi_S(S_3) = \sum_{Q \in S} \pi_Q(S_3) \pi_S(Q) \ge \frac{1}{8n} \sum_{Q \in S} \pi_Q(S_1) \pi_S(Q) = \frac{1}{8n} \pi_S(S_1) \ge \frac{\eta}{16n}. \qedhere \]
  \end{proof}

  We will also need the following simple observation.
  \begin{lemma}
    \label{lem: axis-disjoint reachability}
    Let $q,q'$ be axis-aligned cuboids with a common facet, and $S_1,S_2,S_3$ a partition of $(q \cup q')$ such that $S_1$ and $S_2$ are axis-disjoint. Then, $\pi_{q'}(S_1 \cup S_3) \ge \pi_q(S_1)$.
  \end{lemma}
  \begin{proof}
    Assume without loss of generality that the standard unit vector $e_1$ is
    normal to the common facet $f$. Consider the set
    $T \defeq (S_1 \cap q) + e_1\R$, i.e., the cylindrical set obtained by
    translating the set $S_1 \cap q$ along the direction perpendicular to the
    common facet $f$.  Since $S_1$ and $S_2$ are axis disjoint, $S_2$ cannot
    intersect $T$.  We thus have
    \[ q' \cap T \subseteq q' \cap (S_1 \cup S_3) \text { and }
      q \cap T \supseteq q \cap S_1. \]
    From this, we obtain
    \begin{align*}
      \pi_{q'}(S_1 \cup S_3)
      = \frac{\vol{q' \cap (S_1 \cup S_3)}}{\vol{q'}}
      & \ge \frac{\vol{q' \cap T}}{\vol{q'}}\\
      &= \frac{\surfarea{f \cap T}}{\surfarea{f}}\\
      &= \frac{\vol{q \cap T}}{\vol{q}}
        \geq \frac{\vol{q \cap S_1}}{\vol{q}} = \pi_q(S_1).
    \end{align*}
    Here, the second and third equalities follow since, by definition, $T$ is a
    cylindrical set with axis perpendicular to the common facet $f$ of the
    cuboids $q$ and $q'$.
  \end{proof}

  We are now ready to prove \cref{theo: axis-disjoint isoperimetry}.
  \begin{proof}[Proof of \cref{theo: axis-disjoint isoperimetry}]
    Assume without loss of generality that $\vol{S_1} \le \vol{S_2}$. Let $\FF = \FF^{(\infty)}$
    be a multiscale partition of $K$ into axis-aligned cubes as described in
    \cref{sec:whitney-cubes}, and split it as
    \begin{align*}
      \FF_1 &= \left\{ Q \in \FF \st \pi_Q(S_1) < 2/3 \right\} \text{ and}  \\
      \FF_2 &= \left\{ Q \in \FF \st \pi_Q(S_1) \ge 2/3 \right\}.
    \end{align*}
    As before, given a collection of cubes $\FF' \subseteq \FF$, we
    interchangeably use it to denote $\bigcup_{Q \in \FF'} Q$. Note
    that $\vol{S_1 \cap \FF_1} + \vol{S_1 \cap \FF_2} = \vol{S_1}$.  Now, if
    $\vol{S_1 \cap \FF_1} \ge (1/2) \vol{S_1}$, then
    \begin{align*}
      \vol{S_3} &\ge \sum_{Q \in \FF_1} \vol{S_3 \cap Q} \\
        &\ge \sum_{Q \in \FF_1} \frac{1}{8n} \vol{S_1 \cap Q} \quad \text{(\Cref{lemma: cube isoperimetry})} \\
        &= \frac{1}{8n} \vol{S_1 \cap \FF_1} \ge \frac{1}{16n} \vol{S_1},
    \end{align*}
    and the claimed lower bound on the volume of $S_3$ follows.  Therefore, we
    can assume that
    \begin{equation}
     \vol{\FF_2} \geq   \vol{S_1 \cap \FF_2} \ge (1/2) \vol{S_1}.\label{eq:119}
    \end{equation}
    For any cube $q \in \FF_2$ and facet $f$ of $q$, denote by $\eta_f$ the
    fraction of the facet that is incident on $\FF_1$, and set
    \begin{equation}
      \eta_q \defeq \frac{1}{2n}
      \sum_{\text{facet $f$ of $q$}} \eta_f.\label{eq:4}
    \end{equation}
    We shall show that for some universal constant $c$,
    \begin{equation}
      \label{eq:chr-reduce-to-cube}
      \vol{S_3} \ge \frac{c}{n^{3/2}} \sum_{q \in \FF_2} \eta_q \vol{q}.
    \end{equation}
    Let $q_2 \in \FF_2$ and $f$ a facet of $q_2$ such that $\eta_f \ne 0$.
    \begin{enumerate}
      \item Case 1. $f$ is between two cubes of the same sidelength.
      \begin{figure}[H]
        \centering
        \begin{tikzpicture}
          \coordinate (A1) at (0, 0, 0);
          \coordinate (A2) at (\Depth, 0, 0);
          \coordinate (A3) at (\Depth, 0, \Height);
          \coordinate (A4) at (0, 0, \Height);
          \coordinate (B1) at (0, \Width, 0);
          \coordinate (B2) at (\Depth, \Width, 0);
          \coordinate (B3) at (\Depth, \Width, \Height);
          \coordinate (B4) at (0, \Width, \Height);

          \coordinate (A1') at (\Depth, 0, 0);
          \coordinate (A2') at (2*\Depth, 0, 0);
          \coordinate (A3') at (2*\Depth, 0, \Height);
          \coordinate (A4') at (\Depth, 0, \Height);
          \coordinate (B1') at (\Depth, \Width, 0);
          \coordinate (B2') at (2*\Depth, \Width, 0);
          \coordinate (B3') at (2*\Depth, \Width, \Height);
          \coordinate (B4') at (\Depth, \Width, \Height);

          \draw (A1) -- (A2) -- (A3) -- (A4) -- cycle; \draw (A1) -- (B1) -- (B4) -- (A4) -- cycle; \draw (A1) -- (A2) -- (B2) -- (B1) -- cycle; \draw[fill=red!40,opacity=0.8] (B1) -- (B2) -- (B3) -- (B4) -- cycle; \draw[fill=blue!100,opacity=0.8] (A2) -- (A3) -- (B3) -- (B2) -- cycle; \draw[fill=red!40,opacity=0.8] (A3) -- (B3) -- (B4) -- (A4) -- cycle; 

          \draw (A1') -- (A2') -- (A3') -- (A4') -- cycle; \draw (A1') -- (B1') -- (B4') -- (A4') -- cycle; \draw (A1') -- (A2') -- (B2') -- (B1') -- cycle; \draw[fill=green!10,opacity=0.6] (A3') -- (B3') -- (B4') -- (A4') -- cycle; \draw[fill=green!10,opacity=0.6] (B1') -- (B2') -- (B3') -- (B4') -- cycle; \draw[fill=green!10,opacity=0.6] (A2') -- (A3') -- (B3') -- (B2') -- cycle; \end{tikzpicture}

        {Case 1. The red cube is $q_2 \in \FF_2$ and the green is $q \in \FF_1$. The blue facet is $f$.}
      \end{figure}
      Let $q$ be the cube other than $q_2$ that is bordering $f$. Observe that $\eta_f = 1$. \\
      Using \Cref{lem: axis-disjoint reachability} on $q_2$ and $q$, we get that $\pi_q(S_1 \cup S_3) \ge \pi_{q_2}(S_1) \ge (2/3)$.
By \Cref{cor: minilemma}, $\pi_q(S_3) \ge 1/(24n)$ so
      \begin{equation}
        \label{eqn: chr-isoperimetry-1}
        \vol{S_3 \cap q} \ge \frac{1}{24n} \eta_f \vol{q_2}.
      \end{equation}

      \item Case 2. The cubes other than $q_2$ incident on $f$ are smaller than $q_2$.

      \begin{figure}[H]
        \centering
        \begin{tikzpicture}
\coordinate (A1) at (0, 0, 0);
          \coordinate (A2) at (\Depth, 0, 0);
          \coordinate (A3) at (\Depth, 0, \Height);
          \coordinate (A4) at (0, 0, \Height);
          \coordinate (B1) at (0, \Width, 0);
          \coordinate (B2) at (\Depth, \Width, 0);
          \coordinate (B3) at (\Depth, \Width, \Height);
          \coordinate (B4) at (0, \Width, \Height);
          
          \draw (A1) -- (A2) -- (A3) -- (A4) -- cycle; \draw (A1) -- (B1) -- (B4) -- (A4) -- cycle; \draw (A1) -- (A2) -- (B2) -- (B1) -- cycle; \draw[fill=red!40,opacity=0.8] (B1) -- (B2) -- (B3) -- (B4) -- cycle; \draw[fill=blue!100,opacity=0.8] (A2) -- (A3) -- (B3) -- (B2) -- cycle; \draw[fill=red!40,opacity=0.8] (A3) -- (B3) -- (B4) -- (A4) -- cycle; 

          \coordinate (A11) at (\Depth, 0, 0);
          \coordinate (A12) at (\Depth+0.5*\Depth, 0, 0);
          \coordinate (A13) at (\Depth+0.5*\Depth, 0, 0.5*\Height);
          \coordinate (A14) at (\Depth, 0, 0.5*\Height);
          \coordinate (B11) at (\Depth, 0.5*\Width, 0);
          \coordinate (B12) at (\Depth+0.5*\Depth, 0.5*\Width, 0);
          \coordinate (B13) at (\Depth+0.5*\Depth, 0.5*\Width, 0.5*\Height);
          \coordinate (B14) at (\Depth, 0.5*\Width, 0.5*\Height);

          \draw (A11) -- (A12) -- (A13) -- (A14) -- cycle; \draw (A11) -- (B11) -- (B14) -- (A14) -- cycle; \draw (A11) -- (A12) -- (B12) -- (B11) -- cycle; \draw[opacity=0.6,fill=red!30] (A13) -- (B13) -- (B14) -- (A14) -- cycle; \draw[opacity=0.6,fill=red!30] (B11) -- (B12) -- (B13) -- (B14) -- cycle; \draw[opacity=0.6,fill=red!30] (A12) -- (A13) -- (B13) -- (B12) -- cycle; 

          \coordinate (A21) at (\Depth+0,           0.5*\Width+0,          0);
          \coordinate (A22) at (\Depth+0.5*\Depth,  0.5*\Width+0,          0);
          \coordinate (A23) at (\Depth+0.5*\Depth,  0.5*\Width+0,          0.5*\Height);
          \coordinate (A24) at (\Depth+0,           0.5*\Width+0,          0.5*\Height);
          \coordinate (B21) at (\Depth+0,           0.5*\Width+0.5*\Width, 0);
          \coordinate (B22) at (\Depth+0.5*\Depth,  0.5*\Width+0.5*\Width, 0);
          \coordinate (B23) at (\Depth+0.5*\Depth,  0.5*\Width+0.5*\Width, 0.5*\Height);
          \coordinate (B24) at (\Depth+0,           0.5*\Width+0.5*\Width, 0.5*\Height);

          \draw (A21) -- (A22) -- (A23) -- (A24) -- cycle; \draw (A21) -- (B21) -- (B24) -- (A24) -- cycle; \draw (A21) -- (A22) -- (B22) -- (B21) -- cycle; \draw[fill=green!10,opacity=0.6] (A23) -- (B23) -- (B24) -- (A24) -- cycle; \draw[fill=green!10,opacity=0.6] (B21) -- (B22) -- (B23) -- (B24) -- cycle; \draw[fill=green!10,opacity=0.6] (A22) -- (A23) -- (B23) -- (B22) -- cycle; 

          \coordinate (A31) at (\Depth+0,           0,          0.5*\Height+0);
          \coordinate (A32) at (\Depth+0.5*\Depth,  0,          0.5*\Height+0);
          \coordinate (A33) at (\Depth+0.5*\Depth,  0,          0.5*\Height+0.5*\Height);
          \coordinate (A34) at (\Depth+0,           0,          0.5*\Height+0.5*\Height);
          \coordinate (B31) at (\Depth+0,           0.5*\Width, 0.5*\Height+0);
          \coordinate (B32) at (\Depth+0.5*\Depth,  0.5*\Width, 0.5*\Height+0);
          \coordinate (B33) at (\Depth+0.5*\Depth,  0.5*\Width, 0.5*\Height+0.5*\Height);
          \coordinate (B34) at (\Depth+0,           0.5*\Width, 0.5*\Height+0.5*\Height);

          \draw (A31) -- (A32) -- (A33) -- (A34) -- cycle; \draw (A31) -- (B31) -- (B34) -- (A34) -- cycle; \draw (A31) -- (A32) -- (B32) -- (B31) -- cycle; \draw[fill=green!10,opacity=0.6] (A33) -- (B33) -- (B34) -- (A34) -- cycle; \draw[fill=green!10,opacity=0.6] (B31) -- (B32) -- (B33) -- (B34) -- cycle; \draw[fill=green!10,opacity=0.6] (A32) -- (A33) -- (B33) -- (B32) -- cycle; 

          \coordinate (A41) at (\Depth+0,           0.5*\Width+0,          0.5*\Height+0);
          \coordinate (A42) at (\Depth+0.5*\Depth,  0.5*\Width+0,          0.5*\Height+0);
          \coordinate (A43) at (\Depth+0.5*\Depth,  0.5*\Width+0,          0.5*\Height+0.5*\Height);
          \coordinate (A44) at (\Depth+0,           0.5*\Width+0,          0.5*\Height+0.5*\Height);
          \coordinate (B41) at (\Depth+0,           0.5*\Width+0.5*\Width, 0.5*\Height+0);
          \coordinate (B42) at (\Depth+0.5*\Depth,  0.5*\Width+0.5*\Width, 0.5*\Height+0);
          \coordinate (B43) at (\Depth+0.5*\Depth,  0.5*\Width+0.5*\Width, 0.5*\Height+0.5*\Height);
          \coordinate (B44) at (\Depth+0,           0.5*\Width+0.5*\Width, 0.5*\Height+0.5*\Height);

          \draw (A41) -- (A42) -- (A43) -- (A44) -- cycle; \draw (A41) -- (B41) -- (B44) -- (A44) -- cycle; \draw (A41) -- (A42) -- (B42) -- (B41) -- cycle; \draw[fill=green!10,opacity=0.6] (A43) -- (B43) -- (B44) -- (A44) -- cycle; \draw[fill=green!10,opacity=0.6] (B41) -- (B42) -- (B43) -- (B44) -- cycle; \draw[fill=green!10,opacity=0.6] (A42) -- (A43) -- (B43) -- (B42) -- cycle; \end{tikzpicture}

        {Case 2. The red cubes are in $\FF_2$ and the green in $\FF_1$. The blue facet is $f$.}
      \end{figure}

      Let $T$ be the set of all cubes of smaller size incident on $f$, and set
      $T' \defeq T \cap \FF_1$.  Recall from \cref{item:side-length-ratio} of
      \cref{theo 1.1} that all cubes in $T$ must then have sidelength exactly
      half the sidelength of $q_2$.  We use two consequences of this fact.
      First, that, $\vol{T'} = \eta_f \vol{T} = (1/2) \eta_f \vol{q_2}$.
      Second, that the graph with vertex set equal to the set of cubes in $T$ in
      which two vertices are adjacent if and only if the corresponding cubes are
      adjacent and have a facet in common is exactly the $(n-1)$-dimensional
      Boolean hypercube.
      \begin{enumerate}
      \item Case 2(a). $\eta_f \ge (1/2)$. In this case,
        $\vol{T \setminus T'} \le (1/2) \vol{T}$. Using \Cref{lem: axis-disjoint reachability} on $q_2$ and $T$,
        $\pi_T(S_1 \cup S_3) \ge \pi_{q_2}(S_1) \ge (2/3)$, so $\vol{T \cap (S_1 \cup S_3)} \geq (2/3)\vol{T}$.
        Combining, we get
        \begin{equation*}
          \vol{T' \cap (S_1 \cup S_3)}
          \geq \vol{T \cap (S_1 \cup S_3)}
          - \vol{T\setminus T'} \geq \frac{1}{6}\vol{T} \geq \frac{1}{6}\vol{T'}.
        \end{equation*}
        As a result, $\pi_{T'}(S_1 \cup S_3) \ge (1/6)$. By \Cref{cor:
          minilemma}, $\pi_{T'}(S_3) \ge \frac{1}{96n}$ (since each cube in $T'$
        lies in $\FF_{1}$), so that
        \begin{equation}
          \label{eqn: chr-isoperimetry-2a}
          \vol{S_3 \cap T} \ge \vol{S_3 \cap T'} \ge \frac{1}{96n} \vol{T'} \ge \frac{1}{200n} \eta_f \vol{q_2}.
        \end{equation}

      \item Case 2(b). $\eta_f < (1/2)$.  By \Cref{lem: axis-disjoint reachability},
        for any cube $q$ in $T'$ adjacent to a cube in
        $(T \setminus T') \subseteq \FF_2$, we have
        $\pi_q(S_1 \cup S_3) \ge (2/3)$. Since $q \in T' \subseteq \FF_1$,
        \cref{cor: minilemma} then implies $\pi_q(S_3) \ge 1/(24n)$, so that for
        such cubes
        \begin{equation}
          \label{eq:87}
          \vol{q \cap
            S_{3}} \geq \frac{1}{24n} \cdot \vol{q}.
        \end{equation}
        We shall show that there are many such cubes.\\
        To do this, consider the $(n-1)$-dimensional hypercube graph with vertex
        set equal to the set of cubes $T$ and with two vertices being adjacent if
        and only if the corresponding cubes are adjacent in the sense of sharing a
        facet.
Due to Harper's Theorem \cite{hypercube-vtx-expansion}, the vertex expansion of $T'$ in the hypercube graph is
        $\Omega(n^{-1/2})$ (see \Cref{cor:hypercube-vtx-expansion} for more details). In particular, an $\Omega(n^{-1/2})$ fraction of the cubes in $T'$,
        which therefore constitute at least an $\eta_f \Omega(n^{-1/2})$
        fraction of the cubes in $T$, are adjacent to a cube in
        $(T \setminus T')$.  The total volume of these cubes is thus
        $\eta_f \Omega(n^{-1/2})\cdot{\vol{q_{2}}/2}$. Consequently,
        \begin{equation}
          \label{eqn: chr-isoperimetry-2b}
          \vol{S_3 \cap T} \ge \vol{S_3 \cap T'}
          \stackrel{\text{\cref{eq:87}}}{\ge}
          \frac{1}{48n} \cdot \Omega(n^{-1/2})
          \eta_f \vol{q_2} = \Omega(n^{-3/2}) \eta_f \vol{q_2}.
        \end{equation}
      \end{enumerate}

      \item Case 3. The cube $q_1$ other than $q_2$ incident on $f$ is larger
        than $q_2$.

      \begin{figure}[H]
        \centering
        \begin{tikzpicture}

          \coordinate (A11) at (0, 0, 0);
          \coordinate (A12) at (-0.5*\Depth, 0, 0);
          \coordinate (A13) at (-0.5*\Depth, 0, 0.5*\Height);
          \coordinate (A14) at (0, 0, 0.5*\Height);
          \coordinate (B11) at (0, 0.5*\Width, 0);
          \coordinate (B12) at (-0.5*\Depth, 0.5*\Width, 0);
          \coordinate (B13) at (-0.5*\Depth, 0.5*\Width, 0.5*\Height);
          \coordinate (B14) at (0, 0.5*\Width, 0.5*\Height);

          \draw (A11) -- (A12) -- (A13) -- (A14) -- cycle; \draw (A12) -- (A13) -- (B13) -- (B12) -- cycle; \draw (A11) -- (A12) -- (B12) -- (B11) -- cycle; \draw[fill=green!10,opacity=0.7] (A13) -- (B13) -- (B14) -- (A14) -- cycle; \draw[fill=green!10,opacity=0.7] (B11) -- (B12) -- (B13) -- (B14) -- cycle; \draw[fill=green!10,opacity=0.7] (A11) -- (B11) -- (B14) -- (A14) -- cycle; 

          \coordinate (A21) at (0,           0.5*\Width+0,          0);
          \coordinate (A22) at (-0.5*\Depth,  0.5*\Width+0,          0);
          \coordinate (A23) at (-0.5*\Depth,  0.5*\Width+0,          0.5*\Height);
          \coordinate (A24) at (0,           0.5*\Width+0,          0.5*\Height);
          \coordinate (B21) at (0,           0.5*\Width+0.5*\Width, 0);
          \coordinate (B22) at (-0.5*\Depth,  0.5*\Width+0.5*\Width, 0);
          \coordinate (B23) at (-0.5*\Depth,  0.5*\Width+0.5*\Width, 0.5*\Height);
          \coordinate (B24) at (0,           0.5*\Width+0.5*\Width, 0.5*\Height);

          \draw (A21) -- (A22) -- (A23) -- (A24) -- cycle; \draw (A22) -- (A23) -- (B23) -- (B22) -- cycle; \draw (A21) -- (A22) -- (B22) -- (B21) -- cycle; \draw[fill=red!30,opacity=0.7] (A23) -- (B23) -- (B24) -- (A24) -- cycle; \draw[fill=red!30,opacity=0.7] (B21) -- (B22) -- (B23) -- (B24) -- cycle; \draw[fill=red!30,opacity=0.7] (A21) -- (B21) -- (B24) -- (A24) -- cycle; 

          \coordinate (A31) at (0,           0,          0.5*\Height+0);
          \coordinate (A32) at (-0.5*\Depth,  0,          0.5*\Height+0);
          \coordinate (A33) at (-0.5*\Depth,  0,          0.5*\Height+0.5*\Height);
          \coordinate (A34) at (0,           0,          0.5*\Height+0.5*\Height);
          \coordinate (B31) at (0,           0.5*\Width, 0.5*\Height+0);
          \coordinate (B32) at (-0.5*\Depth,  0.5*\Width, 0.5*\Height+0);
          \coordinate (B33) at (-0.5*\Depth,  0.5*\Width, 0.5*\Height+0.5*\Height);
          \coordinate (B34) at (0,           0.5*\Width, 0.5*\Height+0.5*\Height);

          \draw (A31) -- (A32) -- (A33) -- (A34) -- cycle; \draw (A31) -- (A32) -- (B32) -- (B31) -- cycle; \draw (A32) -- (A33) -- (B33) -- (B32) -- cycle; \draw[fill=red!40,opacity=0.7] (A33) -- (B33) -- (B34) -- (A34) -- cycle; \draw[fill=red!40,opacity=0.7] (B31) -- (B32) -- (B33) -- (B34) -- cycle; \draw[fill=blue!100,opacity=0.7] (A31) -- (B31) -- (B34) -- (A34) -- cycle; 

          \coordinate (A41) at (0,           0.5*\Width+0,          0.5*\Height+0);
          \coordinate (A42) at (-0.5*\Depth,  0.5*\Width+0,          0.5*\Height+0);
          \coordinate (A43) at (-0.5*\Depth,  0.5*\Width+0,          0.5*\Height+0.5*\Height);
          \coordinate (A44) at (0,           0.5*\Width+0,          0.5*\Height+0.5*\Height);
          \coordinate (B41) at (0,           0.5*\Width+0.5*\Width, 0.5*\Height+0);
          \coordinate (B42) at (-0.5*\Depth,  0.5*\Width+0.5*\Width, 0.5*\Height+0);
          \coordinate (B43) at (-0.5*\Depth,  0.5*\Width+0.5*\Width, 0.5*\Height+0.5*\Height);
          \coordinate (B44) at (0,           0.5*\Width+0.5*\Width, 0.5*\Height+0.5*\Height);

          \draw (A41) -- (A42) -- (A43) -- (A44) -- cycle; \draw (A42) -- (A43) -- (B43) -- (B42) -- cycle; \draw (A41) -- (A42) -- (B42) -- (B41) -- cycle; \draw[fill=green!10,opacity=0.7] (A43) -- (B43) -- (B44) -- (A44) -- cycle; \draw[fill=green!10,opacity=0.7] (B41) -- (B42) -- (B43) -- (B44) -- cycle; \draw[fill=green!10,opacity=0.7] (A41) -- (B41) -- (B44) -- (A44) -- cycle; 

\coordinate (A1) at (0, 0, 0);
          \coordinate (A2) at (\Depth, 0, 0);
          \coordinate (A3) at (\Depth, 0, \Height);
          \coordinate (A4) at (0, 0, \Height);
          \coordinate (B1) at (0, \Width, 0);
          \coordinate (B2) at (\Depth, \Width, 0);
          \coordinate (B3) at (\Depth, \Width, \Height);
          \coordinate (B4) at (0, \Width, \Height);
          
          \draw (A1) -- (A2) -- (A3) -- (A4) -- cycle; \draw (A1) -- (B1) -- (B4) -- (A4) -- cycle; \draw (A1) -- (A2) -- (B2) -- (B1) -- cycle; \draw[fill=green!10,opacity=0.6] (B1) -- (B2) -- (B3) -- (B4) -- cycle; \draw[fill=green!10,opacity=0.6] (A2) -- (A3) -- (B3) -- (B2) -- cycle; \draw[fill=green!10,opacity=0.6] (A3) -- (B3) -- (B4) -- (A4) -- cycle; \end{tikzpicture}

        {Case 3. The red cubes are in $\FF_2$ and the green in $\FF_1$.\\
        The blue facet is $f$. $f_1$ is the larger facet it is part of.}
      \end{figure}

      In this case, $\eta_f = 1$ again. Further, from \cref{item:side-length-ratio} of \cref{theo 1.1}, the sidelength of $q_{1}$ is exactly twice that of $q_2$. Let $f_1$ be the facet of $q_1$ that contains $f$, and let $\eta_{f_1}$ be the fraction of $f_1$ that is incident on $\FF_2$. Let $T$ be the set of all cubes of smaller size incident on $f_1$, and $T' = T \cap \FF_2$ (note that $T'$ is now defined in terms of $\FF_2$, unlike in Case 2). Clearly, $\vol{T'} = \eta_{f_1} \vol{T} = (\eta_{f_1}/2) \vol{q_1}$. For each cube $q$ in $T'$, let $f_q$ be the facet of $q$ contained in $f_1$, so $\eta_{f_q} = 1$ for all such $q$. \\
      Noting that $\pi_T(S_1) \ge \pi_{T'}(S_1) \pi_{T}(T') \ge (2/3)\eta_{f_1}$ and using \Cref{lem: axis-disjoint reachability} on $T$ and $q_1$, we get that
      $\pi_{q_1}(S_1 \cup S_3) \ge (2/3)\eta_{f_1}$, so that by
      \Cref{cor: minilemma},
      \[ \pi_{q_1}(S_3 \cap q_1) \ge \frac{\eta_{f_1}}{24n}. \]
      So,
      \begin{equation}
        \label{eqn: chr-isoperimetry-3}
        \vol{S_3 \cap q_1} \ge \frac{\eta_{f_1}}{24n} \vol{q_1} = \frac{1}{12n} \vol{T'} = \frac{1}{12n} \sum_{q \in T'} \eta_{f_q} \vol{q}.
      \end{equation}
      This allows us to associate a volume of $S_3 \cap q_1$ of measure
      $\frac{1}{12n}\eta_{f_q} \vol{q}$ to each cube $q \in T'$, such that the
      volumes associated with distinct cubes in $T'$ are disjoint.
    \end{enumerate}
From \cref{eqn: chr-isoperimetry-1,eqn: chr-isoperimetry-2a,eqn:
      chr-isoperimetry-2b,eqn: chr-isoperimetry-3}, we see that each facet $f$
    of a cube $q$ in $\FF_2$ which abuts a cube in $\FF_1$ can be associated
    with a volume of $S_3$ which lies in cubes in $\FF_1$ that abut $f$ and
    which is of measure at least $c\cdot{n^{-3/2}}\cdot\eta_f \vol{q}$, for some
    universal constant $c$.

    Observe also that for distinct facets $f,f'$ of cubes in $\FF_2$ whose
    normal vectors pointing out of their respective cubes point in the same
    direction, these $S_3$ volumes are disjoint.  Indeed, if $f$ falls in Cases
    1 or 2 above, then the cubes in $\FF_1$ which contain the $S_3$ volume
    associated with $f$ and $f'$ are disjoint.  If $f$ falls in Case 3, $f'$ can
    share the cube in $\FF_1$ which contains the $S_3$-volume associated with
    $f$ only if $f'$ abuts the same cube $q_1$ in $\FF_1$ as $f$ does.  But in
    this case, the remark after \cref{eqn: chr-isoperimetry-3} shows that the
    $S_3$ volumes associated with $f,f'$ are still distinct.

    Since there are only $2n$ distinct directions for the outward normal of a
    facet (two each in each of the $n$ dimensions) we therefore get
    \begin{align}
      2n \cdot \vol{S_3}
      &\ge \frac{c}{n^{3/2}}
        \sum_{q \in \FF_2}
        \sum_{\text{facet $f$ of $q$}}
        \eta_f \vol{q}  \\
      &= 2n \cdot \frac{c}{2n^{3/2}}
        \sum_{q \in \FF_2}
        \eta_q \vol{q}, \text{using the definition of $\eta_q$ in \cref{eq:4},}
        \label{eqn: chr-isoperimetry-penultimate}
    \end{align}
    proving \cref{eq:chr-reduce-to-cube}.
    Now, note that $\pi(\FF_2) \le (3/2) \pi(S_1) \le (3/4)$.  In the multiscale
    walk discussed in earlier sections, the ergodic flow out of
    $\FF_2 \subseteq \FF$ was (noting that $\FF_1 = \FF \setminus \FF_2$)
    \begin{align*}
      \Psi_{{\MM_\infty}}(\FF_2 , \FF \setminus \FF_2)
      &= \sum_{q \in \FF_2} \pi(q) \cdot \frac{1}{2} \sum_{\substack{q' \in \FF_1\\\text{$q'$ abuts $q$}}} \frac{\vol[n-1]{\partial q' \cap \partial q}}{\vol[n-1]{\partial q}} \min\left\{ 1 , \frac{\sidelen{q'}}{\sidelen{q}} \right\} \\
      &\le \sum_{q \in \FF_2} \eta_q \pi(q).
    \end{align*}
    In particular, because $\pi(\FF_2) \le (3/2) \pi(S_1) \le (3/4)$, we have that
    \[ \sum_{q \in \FF_2} \eta_q \pi(q) \ge \Psi_{\MM_\infty}(\FF_2, \FF
      \setminus \FF_2) \ge \frac{1}{4} \Phi_{\MM_\infty}\cdot \pi(\FF_2) \ge
      \frac{1}{8} \Phi_{\MM_\infty} \cdot \pi(S_1). \] Here, the second inequality
    uses
    $\min\inb{\pi(\FF_2), \pi(\FF \setminus \FF_2)} \geq \pi(\FF_2)\pi(\FF
    \setminus \FF_2) \geq \frac{1}{4} \pi(\FF_2)$
    (which follows since $\pi(\FF_2) \leq 3/4$), while the third inequality uses
    \cref{eq:119} (where it was argued that it is enough to consider the case
    where $\pi(\FF_2 \cap S_1) \geq \pi(S_1)/2$). Plugging this back into
    \cref{eqn: chr-isoperimetry-penultimate}, we get
    \[ \vol{S_3} \ge \Omega\left(\frac{\Phi_{\MM_\infty}}{n^{3/2}}\right) \vol{S_1}. \qedhere \]
  \end{proof}

  \subsection{Rapid mixing of CHR from a cold start: Proof of
    \texorpdfstring{\cref{thm:intro-chr}}{Theorem 1.1}}

  Armed with the new isoperimetric inequality for axis-disjoint sets given by
  \cref{theo: axis-disjoint isoperimetry}, we can now replicate the argument of
  \cite{laddha_convergence_2021} to get a conductance lower-bound bound even for
  small sets, in place of the $s$-conductance bound obtained in that paper, which
  approached zero as the size of the set approached zero.

  \begin{theorem}\label{thm:CHR-cond}
    Let $K$ be a convex body in $\R^n$, and let $\Phi_{\MM_\infty}$ denote the
    conductance of the Markov chain $\MM_\infty$ on the Whitney decomposition
    $\FF^{(\infty)}$ of $K$.  Then, the conductance of the coordinate
    hit-and-run chain on $K$ is $\Omega(\Phi_{\MM_\infty}\,n^{-5/2})$.
  \end{theorem}
  \begin{proof} As stated above, the strategy of the proof is essentially
    identical to that of \cite{laddha_convergence_2021}, with the new ingredient
    being the isoperimetry for axis-disjoint sets given by
    \cref{theo: axis-disjoint isoperimetry}.  Let $K = S_1 \cup S_2$ be a
    partition of $K$ into two parts with $\pi(S_1) \le \pi(S_2)$. For $i = 1,2$,
    let
    \[ S_i' = \left\{ x \in S_i \st P_{\CHR}(x,S_{3-i}) < \frac{1}{4n} \right\}. \]
    We claim that $S_1'$ and $S_2'$ are axis-disjoint. Suppose instead that they
    are not, and there is an axis parallel line $\ell$ intersecting both of
    them, with $x_i \in \ell \cap S_i'$ for $i=1,2$, say. Then,
    \[ \frac{1}{4n} > P_{\CHR}(x_i, S_{3-i}) \ge \frac{1}{2n} \cdot
      \frac{\vol[1]{\ell \cap S_{3-i}}}{\vol[1]{\ell \cap K}}, \] so
    $\vol[1]{\ell \cap K} > 2 \vol[1]{\ell \cap S_i}$ for $i = 1,2$. However,
    this is clearly impossible as
    $\vol[1]{\ell \cap S_1} + \vol[1]{\ell \cap S_2} = \vol[1]{\ell \cap
      K}$.

    Now, if $\vol{S_1'} \le (1/2) \vol{S_1}$ (or similarly
    $\vol{S_2'} \le (1/2) \vol{S_2}$), then (here $\Psi_{\CHR}(\cdot,\cdot)$ denotes the ergodic
    flow for the coordinate hit-and-run chain),
    \[ \Psi_{\CHR}(S_1,S_2) \ge \Psi_{\CHR}((S_1 \setminus S_1') , S_2) \ge \frac{1}{4n}
      \pi(S_1 \setminus S_1') \ge \frac{1}{8n} \pi(S_1) \] and we are done. So,
    assume that $\vol{S_i'} \ge (1/2) \vol{S_i}$ for $i = 1,2$. In this case
    \begin{align*}
      \Psi_{\CHR}(S_1,S_2)
      &\ge \frac{1}{2} \left(\Psi_{\CHR}(S_1 \setminus S_1' , S_2)
        + \Psi_{\CHR}(S_1 , S_2 \setminus S_2') \right)\; \text{by reversibility,} \\
      &\ge \frac{1}{8n} \left( \pi(S_1 \setminus S_1') + \pi(S_2 \setminus S_2') \right) \\
      &= \frac{1}{8n} \pi(K \setminus (S_1' \cup S_2')),\;\text{since $K = S_1
        \sqcup S_2$,} \\
      &\ge \frac{c\Phi_{\MM_\infty}}{n^{5/2}} \pi(S_1)\;
        \text{(from \Cref{theo: axis-disjoint isoperimetry}, since $\vol{S_i'} \ge (1/2) \vol{S_i}$)}
    \end{align*}
    for some universal constant $c$, completing the proof.
  \end{proof}

\begin{corollary}\label{thm:chr-l2-mixing}
  Let $K$ be a convex body such that
  $r_\infty \cdot B_{\infty} \subseteq K \subseteq R_{\infty} \cdot B_{\infty}$.
  Let $\rho_\infty \defeq r_\infty/R_{\infty}$.  Let $\pi$ denote the uniform
  measure on $K$. Let $\one_K$ denote the indicator of $K$. Let
  $0 < \eps < 1/2$.  The number of steps $T$ needed for CHR to achieve a density
  $\eta_T$ with respect to $\pi$ such that
  $\|\eta_T - \one_K\|_{L^2(\pi)} < \eps$, from a starting density $\eta_0$
  (with respect to $\pi$) that satisfies $\|\eta_0 - \one_K\|_{L^2(\pi)} < M$
  obeys \beqs T \leq O\left(\frac{n^9}{\rho_\infty^2}\log
  \frac{M}{\eps}\right).\eeqs
\end{corollary} 
\begin{proof}
  \cref{lem:ls-l2-mixing} applied to CHR gives the following.  Let the CHR walk
  be started from a density $\eta_0 \in L^2(\pi)$ with respect to the uniform
  measure on $K$.  Then, after $T$ steps, the density $\eta_T$, of the measure
  at time $T$ with respect to $\pi$, satisfies
  $$\langle \eta_T - \one_K, \eta_T - \one_K\rangle_{L^2(\pi)} \leq \left(1 -
    \frac{\Phi_{\CHR}^2}{2}\right)^{2T} \langle \eta_0 - \one_K, \eta_0 -
  \one_K\rangle_{L^2(\pi)}.$$ The corollary now follows by applying
  \cref{thm:CHR-cond} and the fact that
  $$\Phi_{\MM_\infty} \geq \frac{\rho_\infty}{O(n^{2})},$$ as shown in
  \cref{eq:86} in \cref{thm:conductance}.
\end{proof}

\begin{proof}[Proof of \cref{thm:intro-chr}] The theorem follows immediately
  from \cref{thm:chr-l2-mixing} after a few substitutions.  Let $\eta_0$ be the
  density of the initial $M$-warm start $\nu_{0}$, with respect to $\pi$.  This
  means $\|\eta_0\|_{\infty} \leq M$, which implies that
  $\norm[L^2(\pi)]{\eta_0 - 1} \leq M + 1$.  Next, note that $\rho_\infty = R/r$
  by definition, where $R$ and $r$ are as in the statement of
  \cref{thm:intro-chr}.  The theorem now follows from \cref{thm:chr-l2-mixing}
  by noting that that the density $\eta_T$ (with respect to $\pi$) of the
  measure $\nu_T$ obtained after $T$ steps of the CHR walk satisfies
  $$\|\eta_T - \one_K\|_{L^1(\pi)} \leq \|\eta_T - \one_K\|_{L^2(\pi)},$$ so
  that \cref{thm:chr-l2-mixing} implies
  $d_{TV}(\nu_T, \pi) = \frac{1}{2}\|\eta_T - \one_K\|_{L^1(\pi)} \leq \eps/2$
  for the same $T$.
\end{proof}

\subsection{Mixing of CHR from a point}
\label{sec:chr-point}

As discussed in the introduction, a cold start is often trivial to achieve.  For
example, in the context of \cref{thm:intro-chr} for the coordinate hit-and-run
(CHR) walk, where the body $K$ satisfies
$r\cdot B_\infty \subseteq K \subseteq R \cdot B_\infty$, an
$\exp(\poly{n})$-warm start can be generated simply by sampling the initial
point uniformly at random from $r\cdot B_{\infty}$, provided that the mild condition
that $R/r \leq \exp(\poly{n})$ is satisfied.  However, for aesthetic reasons,
one may want to prove that the chain mixes rapidly even when started from a
given point (i.e., when the initial distribution is concentrated on a point).
To avoid pathological issues that may arise at a ``corner'' of the body, it is
usual to assume that this initial point is somewhat far from the boundary of the
body: say at a distance of at least $R\exp(-\poly{n})$.

In this section, we prove the following corollary, which shows that the CHR walk
mixes in a polynomial number of steps even when started from a distribution
concentrated on a single point of $K$, provided that that point is not too close
to the boundary of $K$.
\begin{corollary}[Mixing time of CHR from a point] There is a universal constant
  $C$ such that the following is true.  Let $K$ be a convex body such that
  $r\cdot B_{\infty} \subseteq K \subseteq R\cdot B_{\infty}$.  Consider the
  coordinate hit-and-run chain on $K$ started from a point $X_0$ satisfying
  $\dist_{\ell_\infty}(X_0, \bdry{K}) \geq \delta$, and let $X_T$ denote the random
  state of the chain after $T$ steps. Let $\epsilon \in (0, 1/2)$ be given, and
  set $\tau = \ceil{2n \log (6n/\epsilon)}$. Then, provided that
  \begin{equation}
    T \geq C\cdot \frac{n^{10} R^2}{r^2}
    \cdot\inp{
      \log \frac{R}{2\delta}
      + \tau \cdot \log \frac{6\tau}{\epsilon}
    },
  \end{equation}
  the total variation distance $d_{TV}(X_T, \pi_K)$ is less than $\eps$.
  \label{thm:chr-mixing-point}
\end{corollary}

Lovász and Vempala~\cite[Corollary 1.2]{lovasz_hit-and-run_2006} pointed out
that their mixing time result for the usual hit-and-run walk from a cold
start~\cite[Theorem 1.1]{lovasz_hit-and-run_2006} immediately implies a
corresponding mixing time result for hit-and-run starting from a point at
distance $\delta > 0$ from the boundary as well, by considering the ``warmth''
of the distribution obtained by taking a single step of the hit-and-run walk
from such a point. The intuition for the above result for the CHR walk is
similar, with the only apparent difficulty being the fact that after any
constant number of steps from a single point, the distribution generated by the
CHR walk is a finite mixture of distributions on lower dimensional subsets of
$\R^n$, and hence cannot be $M$-warm with respect to $\pi_K$ for any finite $M$
(since all such lower dimensional sets have zero probability mass under
$\pi_K$).  The obvious fix is to consider the distribution after about
$2n\log (n/\epsilon)$ steps: by this time, with probability at least
$1-\epsilon$, the chain would have taken a step in each of the coordinate
directions.  All we need to check is that in this time, the chain does not come
too close to the boundary either.  The proof of \cref{thm:chr-mixing-point}
formalizes this intuition.

In the proof of \cref{thm:chr-mixing-point}, we will need the following simple
observation.
\begin{observation}\label{obv:basic-chord}
  Let $K$ be a convex body, and let $x$ be a point in $K$ such that
  $\dist_{\ell_\infty}(x, \bdry{K}) \geq \delta$.  Consider a chord $PQ$ of $K$
  that passes through $x$ (so that $P, Q \in \bdry{K}$).  For
  $\alpha \in [0, 1]$, define $P_\alpha \defeq \alpha x + (1-\alpha)P$ and
  $Q_\alpha \defeq \alpha x + (1-\alpha)Q$. Then the segment
  $P_\alpha Q_\alpha$ covers a $(1-\alpha)$-fraction of the length of $PQ$,
  and for any $y$ on the segment $P_\alpha Q_\alpha$,
  \begin{equation}
    \dist_{\ell_\infty}(y, \bdry{K}) \geq \alpha \delta.
  \end{equation}
\end{observation}
\begin{proof}
  The claim about the ratios of the length of $P_\alpha Q_\alpha$ and $P Q$
  follows immediately from the definition of $P_\alpha$ and $Q_\alpha$.  For the
  second claim, assume without loss of generality that the point $y$ lies
  between $x$ and $P_\alpha$ (the argument when it lies between $x$ and
  $Q_\alpha$ is identical).  Then, there exists $\beta$ satisfying
  $1 \geq \beta \geq \alpha$ such that $y = \beta x + (1 - \beta)P$.  We now
  use the concavity of $\dist_{\ell_\infty}(x, \bdry{K})$ (\cref{lem: l1
    concave}) to get
  \begin{equation*}
    \dist_{\ell_\infty}(y, \bdry{K})
    = \dist_{\ell_\infty}(\beta x + (1-\beta) P, \bdry{K})
    \geq \beta \dist_{\ell_\infty}(x, \bdry{K})
    + (1-\beta) \dist_{\ell_\infty}(P, \bdry{K})
    \geq \alpha \delta. \qedhere
  \end{equation*}
\end{proof}

We are now ready to prove \cref{thm:chr-mixing-point}.
\begin{proof}[Proof of \cref{thm:chr-mixing-point}]
  For convenience of notation, we assume, after scaling the body if necessary,
  that $R = 1$ (to translate the calculations below back to the general setting,
  we will need to replace $\delta$ by $\delta/R$).  Let
  $\tau = \ceil{2n \log (6n/\epsilon)}$ be as in the statement of the corollary.
  We will show that after $\tau$ steps, the probability distribution $\nu$ of the
  state $X_\tau$ of the chain can be written as a convex combination of two
  probability measures $\nu_{\textup{cold}}$ and $\nu_{\textup{rest}}$:
  \begin{equation}
    \nu = p \nu_{\textup{cold}} + (1-p) \nu_{\textup{rest}},\label{eq:chr:123}
  \end{equation}
  where $1 \geq p \geq  1- \epsilon/3$, $\nu_{\textup{rest}}$ is an arbitrary
  probability distribution on $K$, and $\nu_{\textup{cold}}$ is a probability
  distribution on $K$ which is $M$-warm with respect to $\pi_K$, for an $M$
  satisfying
  \begin{equation}
    \label{eq:chr:122}
    \log M \leq 4n \cdot\inp{
      \log \frac{1}{2\delta}
      + \tau \cdot \log \frac{6\tau}{\epsilon}
    },
.
  \end{equation}
  Given this, the claimed mixing time result will follow from a direct
  application of the mixing time bound for CHR from a cold start
  (\cref{thm:intro-chr}).  We now proceed to prove the decomposition claimed in
  \cref{eq:chr:123}.  The intuition is that such a decomposition with a
  reasonable upper bound on $M$ should follow whenever during the first $\tau$
  steps, the chain, with high probability, (i) takes a step at least once in
  each of the $n$ coordinate directions, and (ii) does not come too close to the
  boundary of $K$.  For the above choice of $\tau$, both of these conditions are
  seen to be true with high probability.  We now proceed to formalize
  this intuition.

  To do this, we consider the following alternative description of the run
  of the chain till time $\tau$.  Let $D_i$, for $1 \leq i \leq \tau$, be
  i.i.d.~random variables taking values in $\inb{0, 1, 2, \dots, n}$ such that
  for each $i$, $\Pr{D_i = 0} = 1/2$, and $\Pr{D_i = j} = 1/(2n)$ for
  $1 \leq j \leq n$ (the random variables $D_i$ capture the direction the CHR
  walk chooses when generating the $i$th point, with $D_i = 0$ corresponding to
  the lazy choice of not making a step).  Fix $\kappa \defeq \epsilon/(6\tau)$, and
  let $G_i$, for $1 \leq i \leq \tau$, be i.i.d.~Bernoulli random
  variables with a parameter of $(1-\kappa)$.

  Now, the distribution of the random variables $X_1, X_2, \dots, X_{\tau}$
  representing the first $\tau$ states of the CHR walk started at $X_0$ can be
  described as follows.  Given $X_{i-1}$, where $1 \leq i \leq \tau$, $X_i$ is
  generated as follows: if $D_i = 0$, set $X_i = X_{i-1}$.  Otherwise, let
  $P^{(X_{i-1})}Q^{(X_{i-1})}$ be the chord of $K$ (with
  $P^{(X_{i-1})},Q^{(X_{i-1})} \in \bdry{K}$) through $X_{i-1}$, parallel to
  the standard basis vector $e_{D_{i}}$ of $\R^n$.  Define $P_\kappa^{(X_{i-1})}$,
  $Q_\kappa^{(X_{i-1})}$ as in \cref{obv:basic-chord}.  Let $A_i$ be a point chosen
  uniformly at random from the segment $P_\kappa^{(X_{i-1})}Q_\kappa^{(X_{i-1})}$, and let
  $B_i$ be a point chosen uniformly at random from the remaining length
  $P^{(X_{i-1})}Q^{(X_{i-1})} \setminus P_\kappa^{(X_{i-1})}Q_\kappa^{(X_{i-1})}$ of
  $P^{(X_{i-1})}Q^{(X_{i-1})}$.  Set
  \begin{equation}
    \label{eq:chr:124}
    X_i =
    \begin{cases}
      A_i & \text{if $G_i = 1$, and}\\
      B_i & \text{if $G_i = 0$.}\\
    \end{cases}
  \end{equation}
  The independence of the $D_i$ and the $G_i$, taken together with the fact that
  $P_\kappa^{(X_{i-1})}Q_\kappa^{(X_{i-1})}$ constitutes a $(1-\kappa)$ fraction
  of the chord $P^{(X_{i-1})}Q^{(X_{i-1})}$, then implies that the
  distribution of $(X_0, X_1, X_2, \dots, X_\tau)$ generated in the above way is
  identical to the distribution of the first $\tau$ steps of the CHR walk
  started from $X_0$.  Note also that for any realizations
  $g \in \inb{0, 1}^\tau$ of the $G_i$ and
  $d \in \inb{0, 1, 2, \dots, n}^\tau$ of the $D_i$, the above process induces a
  probability distribution $\nu_{d, g}$ on $K$, such that the probability
  distribution $\nu$ of $X_\tau$ can be written as the convex combination
  \begin{equation}
    \label{eq:chr:126}
    \nu = \sum_{\substack{d \in \inb{0, 1, 2, \dots, n}^\tau\\g \in \inb{0,
          1}^\tau}}
    \Pr{\vec{D} = d \text{ and } \vec{G} = g}\nu_{d,g}.
  \end{equation}

  If $G_i = 1$, then \cref{obv:basic-chord} implies that
  $\dist_{\ell_\infty}(X_i, \bdry{K}) \geq
  \kappa\cdot\dist_{\ell_\infty}(X_{i-1}, \bdry{K}).$ Thus, if $G_i = 1$ for all
  $i$, $1 \leq i \leq \tau$, then
  \begin{equation}
    \dist_{\ell_\infty}(X_i, \bdry{K}) \geq \kappa^i\delta \geq \kappa^\tau \delta
    \; \text{for all } 1 \leq i \leq \tau.
    \label{eq:chr:125}
  \end{equation}
  Call a vector $d \in \inb{0, 1, 2,\dots, n}^{\tau}$ \emph{complete} if it
  includes each of the numbers $\inb{1, 2, \dots, n}$ at least once.  Let $p$ be
  the probability that $G_i = 1$ for all $1 \leq i \leq \tau$, and that
  $(D_1, D_2, \dots, D_\tau)$ is complete.  Then, breaking apart the sum in
  \cref{eq:chr:126} into those $d$ and $g$ that satisfy this requirement and
  those that do not, we get
  \begin{equation}
    \label{eq:chr:127}
    \nu = p \nu_{\text{good}} + (1-p)\nu_{\text{rest}},
  \end{equation}
  where $\nu_{\text{good}}$ is a convex combination of those $\nu_{d, g}$ for
  which $d$ is complete and $g = \vec{1}$, while $\nu_{\text{rest}}$ is a convex
  combination of the remaining $\nu_{d, g}$.  From a simple union bound
  argument, we also get
  \begin{equation}
    1- p \leq \tau\kappa + n \inp{1 - \frac{1}{2n}}^\tau \leq \epsilon/3.\label{eq:chr:130}
  \end{equation}
  Thus, to establish the decomposition in \cref{eq:chr:123}, we only need to
  show that when $d$ is complete and $g = \vec{1}$, $\nu_{d, g}$ is $M$-warm
  with respect to $\pi_K$ for an $M$ satisfying \cref{eq:chr:122} (since a
  convex combination of $M$-warm distributions is also $M$-warm).

  We now prove that for any realization $d$ of the $D_i$ and $g$ of the $G_i$,
  such that $d$ is complete and $g = \vec{1}$, $\nu_{d, g}$ satisfies the
  required warmth condition. For $0 \le i \le \tau$, let $\nu_{d, g, i}$ be the
  distribution of $X_i$ under this realization.  Let $S(i)$ denote the set of
  coordinate directions corresponding to the distinct non-zero values among
  $d_1, d_2, \dots, d_i$.  For convenience of notation, we relabel the
  coordinate directions in the order of their first appearance in $d$, so that
  $S(i) = \inb{e_j \st 1 \leq j \leq |S(i)|}$. Since $d$ is complete, $\abs{S(\tau)} = n$. Define
  \begin{equation}
    M_i \defeq (1-\kappa)^{-i}(2\kappa^\tau\delta)^{-\abs{S(i)}}.\label{eq:chr:122x}
  \end{equation}
  Note that the probability distribution $\nu_{d, g,i}$ is supported on the
  $\abs{S(i)}$-dimensional set
  $(X_0 + \R^{\abs{S(i)}}\times\inb{0}^{n - \abs{S(i)}}) \cap K$.  We will show
  by induction that $\nu_{d, g, i}$ is $M_i$-warm with respect to the standard
  $\abs{S(i)}$-dimensional Lebesgue measure on
  $(X_0 + \R^{\abs{S(i)}}\times\inb{0}^{n - \abs{S(i)}})$.  Let $\rho_{i}$
  denote the density of $v_{d, g, i}$ with respect to the $\abs{S(i)}$
  dimensional Lebesgue measure on
  $(X_0 + \R^{\abs{S(i)}}\times\inb{0}^{n - \abs{S(i)}})$ (implicit in the
  inductive proof of warmth below will be an inductive proof of the existence of
  these densities).  Since $g = \vec{1}$, we have from \cref{eq:chr:125} that
  for all $0 \leq j \leq \tau$,
  \begin{equation}
    \label{eq:chr:131}
    \rho_j(y) > 0\; \text{only if}\; \dist_{\ell_\infty}(y, \bdry{K}) \geq
    \kappa^j\delta \geq \kappa^\tau\delta.
  \end{equation}
  In particular, if $\rho_j(y) > 0$, then the length of any chord of $K$ through
  $y$ in any coordinate direction is at least
  $2\kappa^j\delta \geq 2\kappa^\tau\delta$.

  In the base case $i = 0$, we have $S(i)= 0$ and $M_i = 1$, and
  $\nu_{d, g, 0}$ trivially has density $\rho_0 \equiv 1$ (supported on the single
  point $X_0$). In the inductive case $i \geq 1$, the density does not change
  when from step $i - 1$ to $i$ if $d_i = 0$, so that we get using the
  induction hypothesis and the form of the $M_i$ that at each point $z$,
  $\rho_{i}(z) = \rho_{i-1}(z) \leq M_{i-1} \leq M_{i}$. We are thus left with the case
  $d_i \ge 1$.  We break this case into further cases.
  \begin{description}
  \item[Case 1: $S(i) = S(i-1)$.]  Note that $X_i = z$ is possible only if
    $X_{i-1} = y$ lies on the chord $\tilde{P}\tilde{Q}$ of $K$ (where
    $\tilde{P}, \tilde{Q} \in \bdry{K}$) through $z$ in the direction $e_{d_i}$.
    For such a $y$, we also have $\tilde{P}\tilde{Q} = P^{(y)}Q^{(y)}$ (where,
    as defined in the description before \cref{eq:chr:124}, $P^{(y)}Q^{(y)}$ is
    the chord through $y$ in the direction $e_{d_i}$). Further, since $g_i = 1$,
    we see from \cref{eq:chr:124} that the density at $z$ is given by:
    \begin{equation}
      \rho_{i}(z) = \int\limits_{y \in \tilde{P}\tilde{Q}}\rho_{i-1}(y)
      \cdot \frac{
      I[z \in P_\kappa^{(y)}Q_\kappa^{(y)}]
      }{
      \vol[1]{P_\kappa^{(y)}Q_\kappa^{(y)}}
      }dy.\label{eq:chr:132}
    \end{equation}
    Using $\tilde{P}\tilde{Q} = P^{(y)}Q^{(y)}$, we get from
    \cref{obv:basic-chord} that
    \begin{equation}
      \vol[1]{P_\kappa^{(y)}Q_\kappa^{(y)}} =
      (1-\kappa)\vol[1]{P^{(y)}Q^{(y)}} =
      (1-\kappa)\vol[1]{\tilde{P}\tilde{Q}}. \label{eq:117}\end{equation}
  Substituting this along with the induction hypothesis in \cref{eq:chr:132},
  and using $S(i) = S(i-1)$, we get that for every $z$,
    \begin{equation}
      \label{eq:chr:117}
      \rho_{i}(z)
      \leq M_{i-1}\cdot
      \frac{
        \vol[1]{\tilde{P}\tilde{Q}}
      }{
        (1-\kappa)
        \vol[1]{\tilde{P}\tilde{Q}}
      } = \frac{M_{i-1}}{1-\kappa} =
      (1-\kappa)^{-i}(2\kappa^\tau\delta)^{-\abs{S(i-1)}} = M_i,
    \end{equation}
    which completes the induction in this case.

  \item[Case 2: $\abs{S(i)} = \abs{S(i-1)} + 1$.] In this case, the $i$th step
    is the first occurrence of the direction $e_{d_i}$ in the run of the
    chain. Again, $X_i = z$ is possible only if $X_{i-1} = y$ lies on the chord
    $\tilde{P}\tilde{Q}$ through $z$ in the direction $e_{d_i}$.  But in this
    case, for each
    $z \in (X_0 + \R^{\abs{S(i)}}\times\inb{0}^{n - \abs{S(i)}}) \cap K$, there
    is at most one $y = y(z)$, obtained by setting the $d_i$th coordinate of $z$
    equal to the $d_i$th coordinate of $x_0$, for which $\rho_{i-1}(y) > 0$ and
    which lies on a chord through $z$ in the direction $e_{d_i}$.  Since we also
    have $g_i = 1$, we therefore get from \cref{eq:chr:124} that the density
    $\rho_{i}$ (which is now with respect to an
    $(\abs{S_{i-1}} + 1)$-dimensional Lebesgue measure), is given by
    \begin{equation}
      \label{eq:chr:120}
      \rho_{i}(z) = \rho_{i-1}(y)
      \cdot \frac{
        I[z \in P_\kappa^{(y)}Q_\kappa^{(y)}]
      }{
        \vol[1]{P_\kappa^{(y)}Q_\kappa^{(y)}}
      }.
    \end{equation}
    Again, since we have $g = \vec{1}$, we see from \cref{eq:chr:131} (applied
    with $j = i - 1$) that if $\rho_{i-1}(y) > 0$ then
    $\vol[1]{P_\kappa^{(y)}Q_\kappa^{(y)}} =
    (1-\kappa)\vol[1]{P^{(y)}Q^{(y)}} \geq 2(1-\kappa)(\kappa^\tau\delta)$.
    Substituting this along with the induction hypothesis in \cref{eq:chr:120},
    we thus get (using $\abs{S(i)} = \abs{S(i-1)} + 1$)
    \begin{equation}
      \label{eq:chr:121}
      \rho_{i}(z) \leq \frac{M_{i-1}}{2(1-\kappa)\kappa^\tau\delta}
      \leq (1-\kappa)^{-i}(2\kappa^\tau\delta)^{-(\abs{S(i-1)} + 1)}= M_i,
    \end{equation}
    which completes the induction in this case as well.
  \end{description}

  We thus get that when $d$ is complete and $g = \vec{1}$, then
  $\nu_{d,g} = \nu_{d, g, \tau}$ is $M_\tau$-warm with respect to the
  $n$-dimensional Lebesgue measure on $\R^n$ (and is, of course, supported only
  on $K$).  Since $\pi_K$ has density $I[ \cdot \in K]/\vol[n]{K}$ with respect
  to the Lebesgue measure, it follows that $\nu_{d, g}$ (for such $d$ and $g$)
  is $M_\tau\vol{K}$-warm with respect to $\pi_K$.  Now, since we assumed
  $K \subseteq R B_{\infty}$ with $R \leq 1$, we have $\vol{K} \leq 2^n$.  Thus,
  we see that when $d$ is complete and $g = \vec{1}$, $\nu_{d, g}$ is
  $\inp{(1 - \kappa)^{-\tau}\cdot 2^n \cdot (2\kappa^\tau\delta)^{-n}}$-warm
  with respect to $\pi_K$.  Given the discussion around
  \cref{eq:chr:125,eq:chr:127}, this completes the proof of the decomposition in
  \cref{eq:chr:123} (after recalling that
  $0 < \epsilon < 1/2, \tau = \ceil{2n\log(6n/\epsilon)}$, and
  $\kappa = \epsilon/(6\tau)$).

Given the decomposition in \cref{eq:chr:123}, the claimed bound on the mixing
time now follows using an application of \cref{thm:intro-chr} for $M$-warm
starts.  From that theorem, we see that that there is an absolute constant $C'$
such that when $t \geq T' \defeq C' \frac{n^9R^2}{r^2}\log(6M/\epsilon)$ (where
$M$ is as in \cref{eq:chr:122}) then
  \begin{equation}
    d_{TV}(\nu_{\textup{cold}}\MM_{CHR}^t, \pi_{K}) \leq \frac{\epsilon}{6}.\label{eq:chr:128}
  \end{equation}
  Using the decomposition of $\nu$ in
  \cref{eq:chr:123} and the triangle inequality, this implies that
  \begin{align}
    \label{eq:chr:129}
    d_{TV}(\nu\MM_{CHR}^t, \pi_{K})
    &\leq p d_{TV}(\nu_{\textup{cold}}\MM_{CHR}^t, \pi_{K})
      + (1 - p) d_{TV}(\nu_{\textup{rest}}\MM_{CHR}^t, \pi_{K})\\
    &\stackrel{\textup{\cref{eq:chr:128}}}{\leq}
      p\frac{\epsilon}{6} + (1-p)
      \stackrel{\textup{\cref{eq:chr:130}}}{\leq} \frac{\epsilon}{2}.
  \end{align}
  Since $\nu = \Delta_{X_0}\MM_{CHR}^\tau$, where $\Delta_{X_0}$ is the
  probability distribution on $K$ which puts probability mass $1$ on $X_{0}$, we
  therefore get that the mixing time from $X_0$ up to a total variation distance
  of $\epsilon$ to $\pi_K$ is at most $T' + \tau$, and the claimed bound in the
  statement of the theorem follows from the expressions given above for $M, T'$
  and $\tau$.
\end{proof}

\addcontentsline{toc}{section}{Acknowledgements}
\subsection*{Acknowledgements} \grantacknowledgement{}

\newpage
\appendix
\section*{Appendix}
\addcontentsline{toc}{section}{Appendix}
\section{Proofs omitted from main text}
\subsection{Geometry}
\label{sec:geometry}

Here, we provide the omitted proofs of
\cref{lem:surface-area-infty,lem-small-metric-distances}, both of which we
restate here for convenience.
\surface*
\begin{proof}
  From \cref{item:side-length-ratio} of \cref{theo 1.1}, we know that each
  axis-aligned dyadic cube in $S$ that abuts the axis-aligned dyadic cube $Q$
  must have sidelength whose ratio with the sidelength of $Q$ lies in
  $\inb{2^{-1}, 1, 2}$.  Thus, the area $\bdry{Q} \cap \bdry{S}$ can be
  partitioned into a finite set $T$ of disjoint $(n-1)$-dimensional axis-aligned
  cuboids lying on the surface of $Q$, each having a non-zero surface area. Let
  $a > 0$ be the minimum sidelength over all the sidelengths of cuboids in $T$,
  and let $A \defeq \vol[n-1]{\bdry{Q} \cap \bdry{S}}$ denote the total surface
  area of cuboids in $T$. Let $\epsilon$ be small enough: e.g.
  $\epsilon \leq a/100$ suffices.  By noting that the part of
  $(Q + \epsilon B_{\infty})\setminus Q$ that does \emph{not} project to
  $\bdry{Q}$ along any coordinate direction has volume at most $C \epsilon^2$
  (where $C > 0$ does not depend upon $\epsilon$), and that the part of
  $(Q + \epsilon B_{\infty})\setminus S$ that \emph{does} project to $\bdry{Q}$
  along some coordinate direction has volume $A\epsilon$, we can
  sandwich the $n$-dimensional volume of the set
  $(Q + \epsilon B_{\infty})\setminus S$ as follows:
  \begin{equation}
    A\epsilon \leq \vol[n]{(Q + \epsilon B_{\infty})\setminus S}
    \leq A\epsilon + C\epsilon^2.\label{eq:109}
  \end{equation}
  The lemma now follows by dividing by $\epsilon$ throughout in \cref{eq:109}
  and then taking the limit as $\epsilon \downarrow 0$.
\end{proof}
\metricdist*
\begin{proof}
  Pick $\delta \defeq 1/(2n^{1/p})$.  Consider the cube
  $L \defeq x + 0.9\epsilon \dist_{\ell_p}(x, \bdry{K})B_\infty$.  We then have
  $x \in L$ and $y \not\in L$.  Further, for any point $z \in L$, we have
  \begin{equation}
    \label{eq:14}
    \dist_{\ell_p}(z, \bdry{K}) \leq \dist_{\ell_p}(x, \bdry{K}) + 0.9\epsilon
    \dist_{\ell_p}(x, \bdry{K})\cdot n^{1/p}
    \leq 1.5 \dist_{\ell_p}(x, \bdry{K}),
  \end{equation}
  since $\epsilon \leq \delta = 1/(2n^{1/p})$.  Now, let $\gamma: [0, 1] \to K^\circ$
  be any piecewise continuously differentiable curve with $\gamma(0) = x$ and
  $\gamma(1) = y$.  To prove the lemma, we only need to show that the length of
  any such curve (in the $g_p$ metric) is at least $\epsilon/2$.  We now proceed
  to do so, by considering the part of the curve that lies within $L$. Let
  $t_0 \defeq \inf\inb{t \st \gamma(t) \notin L}$. Note that since
  $\gamma(0) = x \in L$ and $\gamma(1) = y \not \in L$, $t_0 \in (0, 1)$ and
  $\gamma(t_0) \in \bdry{L}$.  We then have
  \begin{align*}
    \len_{g_p}(\gamma)
    & \geq
      \int\limits_{0}^{t_0}\frac{
      \norm[\infty]{\gamma'(t)}
      }{
      \dist_{\ell_p}(\gamma(t), \bdry{K})
      }dt\\
    &\geq \frac{2}{3\dist_{\ell_p}(x,
      \bdry{K})}\int\limits_{0}^{t_0}\norm[\infty]{\gamma'(t)} dt,
      \text{ applying \cref{eq:14} to $\gamma(t) \in L$,}\\
    &\geq \frac{2}{3\dist_{\ell_p}(x, \bdry{K})}
      \cdot \norm[\infty]{x - \gamma(t_0)}
      \geq \frac{\epsilon}{2},
  \end{align*}
  where the last inequality uses that $\gamma(t_0) \in \bdry{L}$ so that
  $\norm[\infty]{x - \gamma(t_0)} = 0.9\epsilon\dist_{\ell_p}(x, \bdry{K})$.
  Since the curve $\gamma$ is an arbitrary piecewise continuously differentiable curve
  connecting $x$ to $y$, this completes the proof.
\end{proof}

\subsection{Properties of  Whitney cubes}
\label{sec:prop-whitn-cubes}
Here, we provide a proof of \cref{theo 1.1}, which we restate here for convenience.

\whitneythm*

\begin{proof}[Proof of \cref{theo 1.1}] Throughout the proof, we use
  $\lambda = 1/2$.  At several places in the proof, we will also use the
  following simple fact without comment: If $Q$ is an axis-aligned cube and
  $x \in Q$, then for any $1 \leq p \leq \infty$,
  $\dist_{\ell_p}(\cntr(Q), x) \leq (1/2)\diam_{{\ell_p}}(Q)$.

  \ben

\item Let $x \in K^\circ$. Choose a positive integer $k$ such that
  $n^{1/p}/2^k < \lambda\dist_{\ell_p}(x,\partial K)/4$. Let $Q_k' \in \cube_k$
  such that $x \in Q_k'$, and let $(Q_i')_{i=0}^{k-1}$ be cubes such that
  $Q_i' \in \cube_{i}$ is the cube whose child is $Q_{i+1}' \in \cube_{i+1}$. Note that
  each $Q_{i}'$ also contains $x$. Suppose, if possible, that
  $x \not\in \bigcup_{Q \in \FF} Q$. Now consider $Q_0' \in \cube_0$ as above,
  which, by its definition, contains $x$. Observe that
  \[ \dist_{\ell_p}(\cntr(Q_0'),K) \le \dist_{\ell_p}(\cntr(Q_0'),x) \le (1/2)
    \diam_{\ell_p}(Q_0') = n^{1/p}/2, \] so $Q_0' \in \FF_0$. On the other hand,
  since $k$ was chosen so that
  $(1/2)\cdot\diam_{\ell_p}(Q_{k-1}') = \diam_{\ell_p}(Q_{k}') = n^{1/p}/2^k < \lambda\dist_{\ell_p}(x,\partial
  K)/4$, we have
\begin{align*}
  \lambda \dist_{\ell_p}(\cntr(Q_{k-1}'),\partial K) &\ge \lambda \dist_{\ell_p}(x,\partial K) - \lambda\dist_{\ell_p}(x,\cntr(Q_{k-1}')) \\
    &\ge \lambda \dist_{\ell_p}(x,\partial K) - (\lambda/2)\diam_{\ell_p}(Q_{k-1}') \\
    &> (2 - \lambda/2)\diam_{\ell_p}(Q_{k-1}') > \diam_{\ell_p}(Q_{k-1}'),
\end{align*}
so that even if $Q_{k-1}'$ were present in $\FF_{k-1}$, it would not be
subdivided into its children: this implies that $Q_k' \not\in \FF_k$. Since
$Q_{0}' \in \FF_0$, it follows that there exists a $j$ satisfying $0 \leq j < k$
such that $Q_j' \in \FF_j$ but $Q_{j+1}' \not\in \FF_{j+1}$. We shall show that
$Q_j' \in \FF$.

To do this, it suffices to show that $\cntr(Q_j') \in K^\circ$. By the definition of
$j$, $Q_{j}'$ is not sub-divided, so that we have
$\lambda \dist_{\ell_p}(\cntr(Q_j'),\partial K) \ge
\diam_{\ell_p}(Q_j')$. Suppose, if possible, that $\cntr(Q_j') \in \Kb$.  Then,
$\dist_{\ell_p}(\cntr(Q_j'),\partial K) = \dist_{\ell_p}(\cntr(Q_j'),K)$, so
that
\[ (1/2)\diam_{\ell_p}(Q_j') \ge \dist_{\ell_p}(\cntr(Q_j'),x) \ge
  \dist_{\ell_p}(\cntr(Q_j'),K) \ge (1/\lambda) \diam_{\ell_p}(Q_j'), \] which
is a contradiction since $\lambda = 1/2$. We have thus shown that
$K^\circ \subseteq \bigcup_{Q \in \FF} Q$.

The reverse containment follows because for any cube $Q \in \FF$,
\begin{align*}
  \dist_{\ell_p}(\cntr(Q),\Kb)
  &= \dist_{\ell_p}(\cntr(Q),\partial K) & \text{(because $\cntr(Q) \in K^\circ$)} \\
  &\ge (1/\lambda) \diam_{\ell_p}(Q) & \text{($Q$ is not further subdivided)} \\
  &> \diam_{\ell_p}(Q), &\text{(since $\la = 1/2$)}
\end{align*}
and as a result, $Q \cap (\Kb) = \emptyset$.  This proves the first part of
\cref{item:Whitney-cubes}.  The second part then follows since
$K \subseteq R_{\infty}\cdot B_{\infty}$ with $R_\infty < 1$ implies
that $K$ cannot contain any cube in $\cube_0$.

\item If possible, let $Q \in \FF_{i} \subseteq \cube_i$ and
  $Q' \in \FF_{j} \subseteq \cube_{j}$ be distinct Whitney cubes with $i \le j$,
  such that $Q^\circ \cap (Q')^\circ \ne \emptyset$.  Note that the interiors of
  any two distinct cubes in $\cube_i$ do not intersect, so it must be the case
  that $i < j$. Then, since cubes in $\cube_{j}$ are obtained by a sequence of
  subdivisions of cubes in $\cube_{i}$, it must be the case that $Q'$ is a
  descendant of $Q$, i.e., obtained by a sequence of subdivisions of $Q$.
  However, since $Q \in \FF_i$, the construction of $\FF_{i}$ implies that no
  child of $Q$ can be present in $\FF_{{i+1}}$. This implies that no cube
  obtained by subdivisions of $Q$ is present in any $\FF_j$ for $j > i$, and
  thus leads to a contradiction since $Q' \in \FF_j$ was required to be a
  descendant of $Q$.

\item The first inequality is direct since if it did not hold, we would have
  further subdivided $Q$, so that $Q$ would not be in $\FF$. From
  \cref{item:Whitney-cubes}, we know that $Q \not\in \cube_0 \supseteq \FF_0$.  So, let $k \geq 0$
  be such that $Q \in \FF_{k+1}$, and let $Q' \in \FF_{k}$ be its parent whose
  subdivision led to the inclusion of $Q \in \FF_{k+1}$. We then have
  \begin{align*}
    2\diam_{\ell_p}(Q)
    &= \diam_{\ell_p}(Q') \\
    &> \lambda\dist_{\ell_p}(\cntr(Q'),\Kb) \qquad \text{(since $Q'$ was subdivided)} \\
    &\ge \lambda\dist_{\ell_p}(\cntr(Q),\Kb) - \lambda\dist_{\ell_p}(\cntr(Q'),\cntr(Q)) \\
    &\ge \lambda\dist_{\ell_p}(\cntr(Q),\Kb) - (\lambda/2) \diam_{\ell_p}(Q),
  \end{align*}
  where the last inequality uses the fact that $\cntr(Q')$ is a vertex of $Q$.
  This implies
  \begin{equation}
    \dist_{\ell_p}(\cntr(Q),\Kb) <  \left( \frac{2}{\lambda} + \frac{1}{2} \right) \diam_{\ell_p}(Q)
    = \frac{9}{2} \diam_{\ell_P}(Q).
  \end{equation}

\item Set $x = \cntr(Q)$. Also let $z_x,z_y \in \Kb$ such that $\dist_{\ell_p}(x,z_x) = \dist_{\ell_p}(x,\Kb)$ and $\dist_{\ell_p}(y,z_y) = \dist_{\ell_p}(y,\Kb)$. Then,
  \begin{equation*}
    \dist_{\ell_p}(y,z_y)
    \ge \dist_{\ell_p}(x,z_y) - \dist_{\ell_p}(x,y)
    \ge \dist_{\ell_p}(x,z_x) - \dist_{\ell_p}(x,y),
  \end{equation*}
  and similarly,
  \begin{equation*}
    \dist_{\ell_p}(y,z_y)
    \le \dist_{\ell_p}(y,z_x)
    \le \dist_{\ell_p}(x,z_x) + \dist_{\ell_p}(x,y).
  \end{equation*}
  Using the upper and lower bounds on $\dist_{\ell_p}(x,z_x)$ derived in
  \cref{item: diameter of cube center} along with the above inequalities and the
  bound $\dist_{\ell_p}(x,y) \leq (1/2)\diam_{\ell_p}(Q)$ yields the claimed
  bounds.

\item Let $Q_1,Q_2 \in \FF$ be abutting cubes with $\diam(Q_1) > \diam(Q_2)$,
  and let $y \in Q_1 \cap Q_2$. \Cref{item:diamater-of-cube} applied to both
  $Q_1$ and $Q_2$ then gives
  \begin{equation*}
    \frac{3}{2} \diam_{\ell_p}(Q_1)
    \le \dist_{\ell_p}(y,\Kb) 
    \le
    5 \diam_{\ell_p}(Q_2).
  \end{equation*}
  This implies
  \[ 1 > \frac{\diam_{\ell_p}(Q_2)}{\diam_{\ell_p}(Q_1)}
    =\frac{\sidelen{Q_2}}{\sidelen{Q_1}}
    \ge \frac{3}{10}. \] Since the ratio of
  sidelengths of any two Whitney cubes is an integral power of two, this forces
  the ratio of the sidelengths of $Q_2$ and $Q_1$ to be $1/2$.  We conclude that
  if two cubes $Q_{1}, Q_{2} \in \FF$ are abutting, then the ratio of their
  sidelengths is an element of the set $\inb{1/2, 1, 2}$. \qedhere \een
\end{proof}

\section{Some results used in proofs}
\subsection{The isoperimetric inequality of Kannan, Lovász and Montenegro}
\label{sec:isop-ineq-kann}
In the proof of \cref{theo isoperimetry}, an isoperimetric inequality due to
Kannan, Lovász and Montenegro~\cite{KLM06} was used.  In their paper, Kannan,
Lovász and Montenegro state their inequality only when the distance between the
sets and the diameter of the convex body are both measured using the
$\ell_2$-norm.  However, since their proof uses the localization lemma
framework of Lovász and Simonovits~\cite{LS93} to reduce the problem to the
setting of a line segment, it applies without any changes even when the
corresponding quantities are measured in any other $\ell_p$ norm, where
$1 \leq p \leq \infty$ (see, e.g., the statement of Corollary 2.7 of
\cite{LS93}). For completeness, we reproduce the statement of the isoperimetric
inequality of Kannan, Lovász, and Montenegro in this more general form, and also
provide a short sketch of how their proof applies also in this setting.

\begin{theorem}[\textbf{Kannan, Lovász, and Montenegro~\cite[Theorem
    4.3]{KLM06}}] Fix $p$ satisfying $1 \leq p \leq \infty$.  Let $K$ be a
  convex body, and let $S_1, S_2$ be disjoint measurable subsets of $K$.  Define
  $S_3\defeq K \setminus (S_1 \cup S_2)$.  Let $\epsilon, D > 0$ be such that
  for any two points $x, y \in K$, $\dist_{\ell_p}(x, y) \leq D$, and such that
  $\dist_{\ell_p}(S_1, S_2) \geq \epsilon$.  Then
  \begin{equation}
    \label{eq:85}
    \vol{S_3} \geq \frac{\epsilon}{D}\cdot \frac{\vol{S_1} \vol{S_2}}{\vol{K}}
    \cdot \log\inp{
      1 + \frac{
        \vol{K}^2
      }{
        \vol{S_1}\vol{S_2}
      }
    }.
  \end{equation}
\end{theorem}
\begin{proof}[Proof sketch] The theorem above is stated and proved by Kannan,
  Lovász and Montenegro for the case $p = 2$ in Section 6.4 of \cite{KLM06}.  To
  prove the result for other $p$, simply repeat the same proof (with $\epsilon$
  and $D$ defined with respect to $\ell_p$ instead of $\ell_{2}$).  The only
  point one has to note is that when \cite{KLM06} apply their Lemma 6.1 in the
  last paragraph of their proof, they only need to consider ratios of lengths
  which lie along the same line segment, and such a ratio does not depend upon
  which $\ell_p$-norm is used to measure the corresponding lengths.
\end{proof}

\subsection{Vertex expansion of the hypercube graph}

In this section, we elaborate on the bound on the vertex expansion used in Case
2(b) of the proof of \Cref{theo: axis-disjoint isoperimetry}, and restate the
relevant results in \cite{harper2}.  We begin by restating some of the notation
from \cite{harper2}.

Let $\mathscr{Y}$ be a finite set, $W$ a probability distribution on
$\mathscr{Y}$, and define the product probability distribution $W^n$ on
$\mathscr{Y}^n$ as
  \[ W^n(y) = \prod_{i=1}^n W(y_i) \]
  for $y \in \mathscr{Y}^n$. For $y,y' \in \mathscr{Y}^n$, introduce the Hamming distance
  \[ d(y,y') \defeq |\{ 1 \le i \le n \st y_i \ne y_i' \}|. \]
  For a set $\mathscr{B} \subseteq \mathscr{Y}^n$, define the Hamming neighbourhood $\Gamma\mathscr{B}$ of $\mathscr{B}$ as
  \[ \Gamma\mathscr{B} \defeq \{ y \in \mathscr{Y}^n \st d(y,y') \le 1 \text{ for some } y' \in \mathscr{B} \} \]
  and the inner boundary $\partial\mathscr{B}$ of $\mathscr{B}$ as
  \[ \partial\mathscr{B} \defeq \Gamma\mathscr{B}^{c} \cap \mathscr{B}. \]
  Also set
  \begin{align*}
    \varphi(t) &= (2\pi)^{-1/2} e^{-t^2/2}, \\
    \Phi(t) &= \int_{-\infty}^t\varphi(x)dx, \text{ and} \\
    f(s) &= \varphi(\Phi^{-1}(s)).
  \end{align*}
  Setting $\mathscr{X} = \{0\}$ in Theorem 5 of \cite{harper2}, we obtain the following.

  \begin{theorem}[\textbf{\cite{harper2}}]
    There is a constant $a$ depending only on $W$ such that for any
    $\mathscr{B} \subseteq \mathscr{Y}^n$,
    \[ W^n(\partial\mathscr{B}) \geq an^{-1/2} f(W^n(\mathscr{B})). \]
  \end{theorem}

  \begin{corollary}
    \label{cor:hypercube-vtx-expansion}
    Let $Q_n$ be the $n$-dimensional hypercube graph $(V,E)$, where
    $V = \{0,1\}^n$ and vertices $u,v$ are adjacent if and only if their Hamming
    distance is $1$. Set $\mu$ to be the uniform distribution on
    $\{0,1\}^n$. For any $S \subseteq V$ with $\mu(S) \le (1/2)$, there exists a
    universal constant $c$ such that
    \[ \frac{\mu(\Gamma(S^{c}) \cap S)}{\mu(S)} \ge cn^{-1/2}, \]
    where $\Gamma(S)$ denotes the neighbourhood of vertices in $S$. 
  \end{corollary}
  \begin{proof}
    Setting $\mathscr{Y} = \{0,1\}$, $\mathscr{B} = S$, and $W$ as the uniform distribution on $\{0,1\}$ in the previous theorem, we get that
    \[ \mu(\Gamma(S^c) \cap S) \ge c n^{-1/2} f(\mu(S)). \]
    To conclude, we note that for $\mu(S) \le (1/2)$, $f(\mu(S)) \ge c' \mu(S)$ for a universal constant $c$. 
  \end{proof}

\addcontentsline{toc}{section}{References}
\bibliographystyle{alphaabbrv}
\bibliography{convex}

\newcommand{\etalchar}[1]{$^{#1}$}
\begin{thebibliography}{CDWY18}

\bibitem[AGK76]{harper2}
R.~Ahlswede, P.~G{\'a}cs, and J.~K{\"o}rner.
\newblock Bounds on conditional probabilities with applications in multi-user
  communication.
\newblock {\em Zeitschrift f{\"u}r Wahrscheinlichkeitstheorie und Verwandte
  Gebiete}, 34(2):157--177, June 1976.
\newblock DOI: \href{https://doi.org/10.1007/BF00535682}{10.1007/BF00535682}.

\bibitem[AK91]{AK91}
D.~Applegate and R.~Kannan.
\newblock Sampling and integration of near log-concave functions.
\newblock In {\em Proc. 23rd {ACM} Symposium on {Theory} of {Computing}
  (STOC)}, pages 156--163. ACM, January 1991.
\newblock DOI:
  \href{https://doi.org/10.1145/103418.103439}{10.1145/103418.103439}.

\bibitem[BC13]{burgisser13:_condit}
P.~Bürgisser and F.~Cucker.
\newblock {\em Condition}.
\newblock Springer, 2013.
\newblock DOI:
  \href{https://doi.org/10.1007/978-3-642-38896-5}{10.1007/978-3-642-38896-5}.

\bibitem[BDJ98]{bubley_elementary_1998}
R.~Bubley, M.~Dyer, and M.~Jerrum.
\newblock An elementary analysis of a procedure for sampling points in a convex
  body.
\newblock {\em Random Structures \& Algorithms}, 12(3):213--235, 1998.
\newblock DOI:
  \href{https://doi.org/10.1002/(SICI)1098-2418(199805)12:3<213::AID-RSA1>3.0.CO;2-Y}{10.1002/(SICI)1098-2418(199805)12:3<213::AID-RSA1>3.0.CO;2-Y}.

\bibitem[CDWY18]{chen2018fast}
Y.~Chen, R.~Dwivedi, M.~J. Wainwright, and B.~Yu.
\newblock Fast {MCMC} sampling algorithms on polytopes.
\newblock {\em JMLR}, 19:86, 2018.
\newblock URL: \url{https://jmlr.org/papers/v19/18-158.html}.

\bibitem[CE22]{CE22}
Y.~Chen and R.~Eldan.
\newblock Localization {Schemes}: {A} {Framework} for {Proving} {Mixing}
  {Bounds} for {Markov} {Chains}.
\newblock In {\em Proc. 63rd IEEE Symposium on Foundations of Computer Science
  (FOCS)}, pages 110--122. IEEE, October 2022.
\newblock DOI:
  \href{https://doi.org/10.1109/FOCS54457.2022.00018}{10.1109/FOCS54457.2022.00018}.
\newblock \arxiv{2203.04163}.

\bibitem[CZ52]{CZ}
A.~P. Calder\'{o}n and A.~Zygmund.
\newblock {On the existence of certain singular integrals}.
\newblock {\em Acta Mathematica}, 88:85 -- 139, 1952.
\newblock DOI: \href{https://doi.org/10.1007/BF02392130}{10.1007/BF02392130}.

\bibitem[DF91]{DF91}
M.~Dyer and A.~Frieze.
\newblock Computing the volume of convex bodies: a case where randomness
  provably helps.
\newblock In B.~Bollobás, editor, {\em Proceedings of {Symposia} in {Applied}
  {Mathematics}}, volume~44, pages 123--169. American Mathematical Society,
  1991.
\newblock DOI:
  \href{https://doi.org/10.1090/psapm/044/1141926}{10.1090/psapm/044/1141926}.

\bibitem[DFK91]{DFK91}
M.~Dyer, A.~Frieze, and R.~Kannan.
\newblock A random polynomial-time algorithm for approximating the volume of
  convex bodies.
\newblock {\em J. ACM}, 38(1):1--17, January 1991.
\newblock DOI:
  \href{https://doi.org/10.1145/102782.102783}{10.1145/102782.102783}.

\bibitem[EG15]{coarea}
L.~Evans and R.~Gariepy.
\newblock {\em Measure Theory and Fine Properties of Functions, Revised
  Edition}.
\newblock Chapman and Hall/CRC, 2015.
\newblock DOI: \href{https://doi.org/10.1201/b18333}{10.1201/b18333}.

\bibitem[Fed96]{federer_geometric_1996}
H.~Federer.
\newblock {\em Geometric measure theory}.
\newblock Classics in mathematics. Springer, 1996.
\newblock DOI:
  \href{https://doi.org/10.1007/978-3-642-62010-2}{10.1007/978-3-642-62010-2}.
\newblock Reprint of the 1969 edition.

\bibitem[Fef05]{Feff0}
C.~L. Fefferman.
\newblock A sharp form of {Whitney}'s extension theorem.
\newblock {\em Annals of Mathematics}, 161(1):509--577, 2005.
\newblock URL: \url{http://www.jstor.org/stable/3597349}.

\bibitem[FK09]{Feff}
C.~L. Fefferman and B.~Klartag.
\newblock {Fitting a $C^m$-Smooth Function to Data II}.
\newblock {\em Revista Matemática Iberoamericana}, 25(1):49 -- 273, 2009.
\newblock DOI: \href{https://doi.org/10.4171/RMI/569}{10.4171/RMI/569}.

\bibitem[FSA20]{fallahi_comparison_2020}
S.~Fallahi, H.~J. Skaug, and G.~Alendal.
\newblock A comparison of {Monte} {Carlo} sampling methods for metabolic
  network models.
\newblock {\em PLOS ONE}, 15(7):e0235393, July 2020.
\newblock DOI:
  \href{https://doi.org/10.1371/journal.pone.0235393}{10.1371/journal.pone.0235393}.

\bibitem[FV23]{FV23}
M.~Fernandez~V.
\newblock On the {$\ell_0$} isoperimetric coefficient of measurable sets.
\newblock {\em arXiv preprint arXiv:2312.00015}, 2023.

\bibitem[Har66]{hypercube-vtx-expansion}
L.~H. Harper.
\newblock Optimal numberings and isoperimetric problems on graphs.
\newblock {\em Journal of Combinatorial Theory, Series A}, 1:385--393, 1966.

\bibitem[HCT{\etalchar{+}}17]{H+17}
H.~S. Haraldsdóttir, B.~Cousins, I.~Thiele, R.~M.~T. Fleming, and S.~S.
  Vempala.
\newblock {CHRR}: Coordinate hit-and-run with rounding for uniform sampling of
  constraint-based models.
\newblock {\em Bioinformatics}, 33(11):1741--1743, June 2017.
\newblock DOI:
  \href{https://doi.org/10.1093/bioinformatics/btx052}{10.1093/bioinformatics/btx052}.

\bibitem[JS88]{jerrum_conductance_1988}
M.~Jerrum and A.~Sinclair.
\newblock Conductance and the rapid mixing property for {Markov} chains: the
  approximation of permanent resolved.
\newblock In {\em Proc. 20th {ACM} Symposium on {Theory} of Computing (STOC)},
  pages 235--244. ACM, January 1988.
\newblock DOI: \href{https://doi.org/10.1145/62212.62234}{10.1145/62212.62234}.

\bibitem[KLM06]{KLM06}
R.~Kannan, L.~Lovász, and R.~Montenegro.
\newblock Blocking conductance and mixing in random walks.
\newblock {\em Combinatorics, Probability and Computing}, 15(4):541--570, July
  2006.
\newblock DOI:
  \href{https://doi.org/10.1017/S0963548306007504}{10.1017/S0963548306007504}.

\bibitem[KN12]{kannan2012random}
R.~Kannan and H.~Narayanan.
\newblock Random walks on polytopes and an affine interior point method for
  linear programming.
\newblock {\em Mathematics of Operations Research}, 37(1):1--20, 2012.
\newblock URL: \url{www.jstor.org/stable/41412339}.

\bibitem[LK99]{LK}
L.~Lov\'{a}sz and R.~Kannan.
\newblock Faster mixing via average conductance.
\newblock In {\em Proc. 31st Annual ACM Symposium on Theory of Computing
  (STOC)}, pages 282--287. ACM, 1999.
\newblock DOI:
  \href{https://doi.org/10.1145/301250.301317}{10.1145/301250.301317}.

\bibitem[LLV20]{Laddha}
A.~Laddha, Y.~T. Lee, and S.~S. Vempala.
\newblock Strong self-concordance and sampling.
\newblock In {\em Proc. 52nd {ACM} {SIGACT} {Symposium} on {Theory} of
  {Computing} (STOC)}, pages 1212--1222. ACM, June 2020.
\newblock URL: \url{https://doi.org/10.1145/3357713.3384272}.
\newblock \arxiv{1911.05656}.

\bibitem[Lov90]{L90-ICM}
L.~Lovász.
\newblock Geometric algorithms and algorithmic geometry.
\newblock In {\em Proceedings of the International Congress of Mathematicians
  (ICM)}, volume~1, pages 139--154, 1990.

\bibitem[Lov99]{L99}
L.~Lovász.
\newblock Hit-and-run mixes fast.
\newblock {\em Math. Program.}, 86(3):443--461, December 1999.
\newblock DOI:
  \href{https://doi.org/10.1007/s101070050099}{10.1007/s101070050099}.

\bibitem[LS90]{LS90}
L.~Lovász and M.~Simonovits.
\newblock The mixing rate of {Markov} chains, an isoperimetric inequality, and
  computing the volume.
\newblock In {\em Proc. 31st IEEE {Symposium} on {Foundations} of {Computer}
  {Science} (FOCS)}, pages 346--354 vol. 1. IEEE, October 1990.
\newblock DOI:
  \href{https://doi.org/10.1109/FSCS.1990.89553}{10.1109/FSCS.1990.89553}.

\bibitem[LS93]{LS93}
L.~Lovász and M.~Simonovits.
\newblock Random walks in a convex body and an improved volume algorithm.
\newblock {\em Random Structures \& Algorithms}, 4(4):359--412, 1993.
\newblock DOI:
  \href{https://doi.org/10.1002/rsa.3240040402}{10.1002/rsa.3240040402}.

\bibitem[LV06a]{lovasz_hit-and-run_2006}
L.~Lovász and S.~S. Vempala.
\newblock Hit-and-run from a corner.
\newblock {\em SIAM J. Comput.}, 35(4):985--1005, January 2006.
\newblock DOI:
  \href{https://doi.org/10.1137/S009753970544727X}{10.1137/S009753970544727X}.

\bibitem[LV06b]{lovasz_simulated_2006}
L.~Lovász and S.~S. Vempala.
\newblock Simulated annealing in convex bodies and an {$O^*(n^4)$} volume
  algorithm.
\newblock {\em Journal of Computer and System Sciences}, 72(2):392--417, March
  2006.
\newblock DOI:
  \href{https://doi.org/10.1016/j.jcss.2005.08.004}{10.1016/j.jcss.2005.08.004}.

\bibitem[LV17]{lee2017geodesic}
Y.~T. Lee and S.~S. Vempala.
\newblock Geodesic walks in polytopes.
\newblock In {\em Proc. 49th {ACM} {Symposium} on {Theory} of {Computing}
  (STOC)}, pages 927--940. ACM, June 2017.
\newblock DOI:
  \href{https://doi.org/10.1145/3055399.3055416}{10.1145/3055399.3055416}.

\bibitem[LV22]{LV-survey}
Y.~T. Lee and S.~S. Vempala.
\newblock The manifold joys of sampling.
\newblock In {\em Proc. 49th International Colloquium on Automata, Languages,
  and Programming (ICALP)}, volume 229 of {\em Leibniz International
  Proceedings in Informatics (LIPIcs)}, pages 4:1--4:20. Schloss Dagstuhl --
  Leibniz-Zentrum f{\"u}r Informatik, 2022.
\newblock DOI:
  \href{https://doi.org/10.4230/LIPIcs.ICALP.2022.4}{10.4230/LIPIcs.ICALP.2022.4}.
\newblock Invited talk.

\bibitem[LV23]{laddha_convergence_2021}
A.~Laddha and S.~S. Vempala.
\newblock Convergence of {Gibbs} sampling: {Coordinate} hit-and-run mixes fast.
\newblock {\em Discrete \& Computational Geometry}, 70:406--425, 2023.
\newblock \arxiv{2009.11338}.

\bibitem[MS01]{MS01}
R.~Montenegro and J.-B. Son.
\newblock Edge isoperimetry and rapid mixing on matroids and geometric {Markov}
  chains.
\newblock In {\em Proc. 33rd {ACM} Symposium on {Theory} of Computing (STOC)},
  {STOC} '01, pages 704--711, New York, NY, USA, July 2001. Association for
  Computing Machinery.
\newblock DOI:
  \href{https://doi.org/10.1145/380752.380876}{10.1145/380752.380876}.

\bibitem[MV19]{Mangoubi}
O.~Mangoubi and N.~K. Vishnoi.
\newblock Faster polytope rounding, sampling, and volume computation via a
  sub-linear ball walk.
\newblock In {\em Proc. 60th IEEE {Symposium} on {Foundations} of {Computer}
  {Science} ({FOCS})}, pages 1338--1357. IEEE, November 2019.
\newblock DOI:
  \href{https://doi.org/10.1109/FOCS.2019.00082}{10.1109/FOCS.2019.00082}.

\bibitem[Nar16]{HN2}
H.~Narayanan.
\newblock Randomized interior point methods for sampling and optimization.
\newblock {\em The Annals of Applied Probability}, 26(1):597--641, February
  2016.
\newblock DOI: \href{https://doi.org/10.1214/15-AAP1104}{10.1214/15-AAP1104}.

\bibitem[NS22]{narayanan_srivastava_2022}
H.~Narayanan and P.~Srivastava.
\newblock On the mixing time of coordinate hit-and-run.
\newblock {\em Combinatorics, Probability and Computing}, 31(2):320–332,
  2022.
\newblock DOI:
  \href{https://doi.org/10.1017/S0963548321000328}{10.1017/S0963548321000328}.
\newblock \arxiv{2009.14004}.

\bibitem[Smi84]{smith_efficient_1984}
R.~L. Smith.
\newblock Efficient {Monte} {Carlo} procedures for generating points uniformly
  distributed over bounded regions.
\newblock {\em Operations Research}, 32(6):1296--1308, 1984.
\newblock URL: \url{https://www.jstor.org/stable/170949}.

\bibitem[Ste70]{Stein}
E.~M. Stein.
\newblock {\em Singular Integrals and Differentiability Properties of Functions
  (PMS-30)}.
\newblock Princeton University Press, 1970.
\newblock URL: \url{http://www.jstor.org/stable/j.ctt1bpmb07}.

\bibitem[TV17]{tsukerman_brunn-minkowski_2017}
E.~Tsukerman and E.~Veomett.
\newblock Brunn-{Minkowski} theory and {Cauchy’s} surface area formula.
\newblock {\em American Mathematical Monthly}, 124(10):922--929, 2017.
\newblock DOI:
  \href{https://doi.org/10.4169/amer.math.monthly.124.10.922}{10.4169/amer.math.monthly.124.10.922}.

\bibitem[Whi34]{Whitney}
H.~Whitney.
\newblock Analytic extensions of differentiable functions defined in closed
  sets.
\newblock {\em Transactions of the American Mathematical Society},
  36(1):63--89, 1934.
\newblock URL: \url{http://www.jstor.org/stable/1989708}.

\end{thebibliography}

\end{document}